\documentclass[envcountsame,runningheads,notitlepage]{llncs}

\usepackage{amsmath, amssymb, amsfonts,stmaryrd, xspace}
\SetSymbolFont{stmry}{bold}{U}{stmry}{m}{n}

\usepackage{thm-restate}

\usepackage[utf8]{inputenc}
\usepackage[T1]{fontenc}
\usepackage{color}
\usepackage{enumitem}
\usepackage{ellipsis}
\usepackage{etoolbox} \usepackage[misc,geometry]{ifsym} \usepackage{mathrsfs}
\usepackage{mathtools}
\usepackage[normalem]{ulem}

\usepackage{makecell}
\usepackage[mathscr]{eucal}

\usepackage{tikz}
\usetikzlibrary{matrix}

\usepackage{todonotes}

\usepackage{lipsum}

\usepackage{breakcites}
\usepackage{float}
\usepackage{booktabs}
\usepackage{threeparttable}

\usepackage{hyperref}
\usepackage[normalem]{ulem}
\newtoggle{longversion}
\toggletrue{longversion}

\definecolor{ao}{rgb}{0.4, 0.8, 0.0}

\newcommand{\word}[1]{\ensuremath{\boldsymbol{\rm #1}}}
\renewcommand{\vec}[1]{\word{#1}}
\newcommand{\bfa}{\word{a}}
\newcommand{\bfb}{\word{b}}

\newcommand{\bfe}{\word{e}}
\newcommand{\bfg}{\word{g}}
\newcommand{\bfh}{\word{h}}

\newcommand{\bfm}{\word{m}}
\newcommand{\bfs}{\word{s}}

\newcommand{\bfv}{\word{v}}

\newcommand{\bfx}{\word{x}}
\newcommand{\bfy}{\word{y}}
\newcommand{\bfz}{\word{z}}

\newcommand{\mat}[1]{\ensuremath{\boldsymbol{\rm #1}}}
\newcommand{\bfA}{\mat{A}}

\newcommand{\bfG}{\mat{G}}
\newcommand{\bfH}{\mat{H}}

\newcommand{\bfGamma}{\mat{\Gamma}}

\newcommand{\Gal}{{\ensuremath{\rm Gal}}}
\newcommand{\Aut}{{\ensuremath{\rm Aut}}}

\newcommand{\goth}[1]{\ensuremath{\mathfrak{#1}}}

\newcommand{\OO}{\mathcal{O}} 
 \newcommand{\gothP}{\mathfrak{P}}

\newcommand{\F}[1]{\mathbb{F}_{#1}}
\newcommand{\Fp}{\mathbb{F}_p}
\newcommand{\Fq}{\mathbb{F}_q}
\newcommand{\FF}{\mathbb{F}}
\newcommand{\NN}{\mathbb{N}}
\newcommand{\N}{\mathbb{N}}
\newcommand{\QQ}{\mathbb{Q}}

\newcommand{\CC}{\mathbb{C}}
\newcommand{\GG}{\mathbb{G}}
\newcommand{\ZZ}{\mathbb{Z}}

\newcommand{\Ring}{\mathcal{R}}

\newcommand{\ideal}[1]{\langle #1 \rangle}

\mathchardef\mhyphen="2D
\newcommand{\PCG}{{\sf PCG}\xspace}
\newcommand{\PCGs}{{\sf PCG}'s\xspace}
\newcommand{\PRG}{{\sf PRG}\xspace}
\newcommand{\Expand}{\ensuremath{\mathsf{Expand}}\xspace}
\newcommand{\SPFSS}{{\sf SPFSS}\xspace}
\newcommand{\key}{{\sf k}}

\newcommand{\RSample}{{\sf RSample}\xspace}
\newcommand{\cF}{\mathcal{F}}

\newcommand{\R}{R}

\newcommand{\cC}{\mathcal{C}}
 
\newcommand{\FullEval}{\ensuremath{\mathsf{FullEval}}\xspace}
\newcommand{\Eval}{\ensuremath{\mathsf{Eval}}\xspace}
\newcommand{\Sim}{\ensuremath{\mathsf{Sim}}\xspace}
\newcommand{\Leak}{\ensuremath{\mathsf{Leak}}\xspace}
\newcommand{\abs}[1]{\ensuremath{|#1|}\xspace}
\newcommand{\class}{\ensuremath{\mathcal{C}}\xspace}

\newcommand{\FSS}{{\sf FSS}\xspace}
\newcommand{\Gen}{{\sf Gen}\xspace}
\newcommand{\Share}{{\sf Share}}

\newcommand{\bit}{\{0,1\}}

\newcommand{\secpar}{{\lambda}\xspace}
\newcommand{\getsr}{\stackrel{{}_\$}{\leftarrow}}
\newcommand{\Dist}{\ensuremath{\mathcal{D}}\xspace}
\newcommand{\UniDist}{\ensuremath{\mathcal{U}}\xspace}

\newcommand{\suchthat}{\ensuremath{\phantom{i}|\phantom{i}}\xspace}
\newcommand{\Adv}{\ensuremath{\mathcal{A}}\xspace}
\newcommand{\bias}{\ensuremath{\mathsf{bias}}\xspace}
\newcommand{\poly}{{\mathsf{poly}}}

\newcommand{\HW}[1]{\ensuremath{\text{wt}(#1)}\xspace}

\newcommand{\indist}{\stackrel{\rm c}{\approx}}

\newcommand{\Trans}{\ensuremath{{\intercal}}\xspace}
\newcommand{\QACSD}{\ensuremath{\mathsf{QA}\text{-}\mathsf{SD}}\xspace}
\newcommand{\QASD}{\ensuremath{\mathsf{QA}\text{-}\mathsf{SD}}\xspace}
\newcommand{\QCSD}{\ensuremath{\mathsf{QC}\text{-}\mathsf{SD}}\xspace}
\newcommand{\Expect}{\mathbb{E}}

\newcommand{\mindist}{\ensuremath{\mathsf{d}}\xspace}
\newcommand{\wt}[1]{\HW{#1})}

\newcommand{\code}[1]{\mathscr{#1}}

\newcommand{\eqdef}{\stackrel{\textrm{def}}{=}}
\renewcommand{\leq}{\leqslant}
\renewcommand{\le}{\leqslant}
\renewcommand{\geq}{\geqslant}
\renewcommand{\ge}{\geqslant}
\newcommand{\map}[4]{\left\{
    \begin{array}{ccc}
      #1 & \longrightarrow & #2\\
      #3 & \longmapsto & #4
    \end{array}\right.
}
\newcommand{\ie}{{\em i.e. }}

\newcommand{\eps}{\varepsilon}
\newcommand{\Prob}{\mathbb{P}}
\newcommand{\adv}{Adv}

\newcommand{\distrib}{\psi}
\newcommand{\sample}{\leftarrow}
\newcommand{\ffd}[1]{\mathcal{F}_{#1}}

\newcommand{\Ber}{{\rm Ber}}

\newcommand{\fract}[2]{\hbox{\leavevmode
\kern.1em \raise .5ex \hbox{\the\scriptfont0 $#1$}\kern-.1em }/
\hbox{\kern-.15em \lower .25ex \hbox{\the\scriptfont0 $#2$}}
}

\newcommand{\LPN}{\ensuremath{\mathsf{LPN}}\xspace}
\newcommand{\VDLPN}{\ensuremath{\mathsf{VDLPN}}\xspace}
\newcommand{\EALPN}{\ensuremath{\mathsf{EALPN}}\xspace}

\newcommand{\SD}{\ensuremath{{\mathsf{SD}}}\xspace}
\newcommand{\FFDP}{\ensuremath{{\mathsf{FF\mhyphen DP}}}\xspace}
\newcommand{\RLWE}{\textsf{Ring-}\ensuremath{\mathsf{LWE}}\xspace}
\newcommand{\LWE}{\ensuremath{\mathsf{LWE}}\xspace}
\newcommand{\OLE}{\ensuremath{{\mathsf{OLE}}}\xspace}
\newcommand{\OLEs}{\ensuremath{{\mathsf{OLE}}}'s\xspace}
\newcommand{\OT}{\ensuremath{{\mathsf{OT}}}\xspace}
\newcommand{\OTs}{\ensuremath{{\mathsf{OT}}}'s\xspace}
\newcommand{\DOOM}{\ensuremath{{\mathsf{DOOM}}}\xspace}

\newcommand{\BIKE}{{\rm BIKE}\xspace}
\newcommand{\HQC}{{\rm HQC}\xspace}

 \author{\iftoggle{longversion}{Maxime Bombar\inst{1,2}\and Geoffroy Couteau\inst{3,4}\and Alain
      Couvreur\inst{2,1}\and Clément Ducros\inst{4,2}}{Maxime
      Bombar\inst{1,2}\orcidID{0000-0001-9081-6094}\textsuperscript{({\tiny\Letter})
      } \and Geoffroy Couteau\inst{3,4}\orcidID{0000-0002-6645-0106}
      \and Alain Couvreur\inst{2,1}\orcidID{0000-0003-4554-6720} \and
      Clément Ducros\inst{4,2}} \thanks{This work was funded by the
      French Agence Nationale de la Recherche through the France 2030
      ANR Projects ANR-22-PECY-003 SecureCompute and ANR-22-PETQ-0008
      PQ-TLS, the ANR BARRACUDA (ANR-21-CE39-0009-BARRACUDA), the ANR
      SCENE (ANR-20-CE39-0001-SCENE), and by the DIM RFSI through the
      project LICENCED. }}

  \institute{Laboratoire LIX, \'Ecole Polytechnique, Institut
    Polytechnique de Paris, 1 rue Honor\'e d'Estienne d'Orves, 91120
    {\sc Palaiseau Cedex} \and INRIA \and CNRS \and IRIF, Université
    Paris Cité \\
    \email{\{maxime.bombar, alain.couvreur\}@inria.fr}\\
    \email{\{couteau, cducros\}@irif.fr}
  }

\iftoggle{longversion}{
\usepackage{geometry}
\geometry{
  a4paper,         textwidth=15cm,  textheight=24cm, heightrounded,   hratio=1:1,      vratio=2:3,      }}

\begin{document}

\makeatletter
\renewcommand*\l@author[2]{}
\renewcommand*\l@title[2]{}
\makeatother

\title{Correlated Pseudorandomness from the Hardness of Quasi-Abelian Decoding}
\iftoggle{longversion}{\date{\today}}{\date{}}

\maketitle

\begin{abstract}
  \iftoggle{longversion}{Secure computation often benefits from the
      use of correlated randomness to achieve fast, non-cryptographic
      online protocols. }{}A recent paradigm put forth by Boyle
    \emph{et al.} (CCS 2018, Crypto 2019) showed how
    \emph{pseudorandom correlation generators} (\PCG) can be used to
    generate large amounts of useful forms of correlated
    (pseudo)randomness, using minimal interactions followed solely by
    local computations, yielding \emph{silent} secure two-party
    computation protocols\iftoggle{longversion}{ (protocols where the
        preprocessing phase requires almost no communication).
        Furthermore, \emph{programmable} \PCGs can be used similarly
        to generate multiparty correlated randomness to be used in
        silent secure $\textup{N}$-party protocols.}{. This can be
        extended to $\textup{N}$-party using \emph{programmable}
        \PCGs.} Previous works constructed very efficient
      (non-programmable) \PCGs for correlations such as random
      oblivious transfers. However, the situation is less satisfying
      \iftoggle{longversion}{for the case of \emph{random oblivious
            linear evaluation} (\OLE), which generalises oblivious
          transfers over large fields, and are a core resource for
          secure computation of arithmetic circuits.}{for \emph{random
            oblivious linear evaluations} (\OLEs), their
          generalisation over large fields.} The state-of-the-art work
        of Boyle \emph{et al.} (Crypto 2020) constructed programmable
        \PCGs for \OLE, but their work suffers from two important
        downsides: (1) it only generates \OLEs over \emph{large
          fields}, and (2) it relies on a relatively new
        ``splittable'' ring-\LPN assumption, which lacks strong
        security foundations.

        In this work, \iftoggle{longversion}{we construct new
            programmable \PCGs for the \OLE correlation, that overcome
            both limitations. To this end,}{} we introduce the
          \emph{quasi-abelian syndrome decoding problem} (\QASD), a
          family of assumptions which generalises the well-established
          quasi-cyclic syndrome decoding
          assumption\iftoggle{longversion}{. Building upon \QASD, we
              construct new programmable \PCGs for \OLEs over any
              field $\FF_q$ with $q > 2$.}{ and allows to construct
              new programmable \PCGs for \OLE over any field $\FF_{q}$
              with $q>2$.} Our analysis also sheds light on the security of the ring-\LPN
        assumption used in Boyle \emph{et al.} (Crypto 2020). Using
        our new \PCGs, we obtain the first efficient
        $\textup{N}$-party silent secure computation protocols for
        computing general arithmetic circuit over $\FF_q$ for any
        $q > 2$.
\end{abstract}

\begin{keywords}
  Pseudorandom correlation generators, oblivious linear evaluation,
  quasi-abelian codes, silent secure computation
\end{keywords}

\iftoggle{longversion}{
\setcounter{tocdepth}{2}
  \tableofcontents}{}

\section{Introduction}
\label{sec:introduction}

Correlated randomness is a powerful resource in secure computation.
Following the seminal work of Beaver~\cite{C:Beaver91b}, many lightweight,
concretely efficient secure computation protocols have been designed
in a model where the parties have access to long trusted correlated
random strings: $\Omega(n)$-length instances of a simple correlation
enable securely computing circuits with $n$ gates. Depending on the
setting, various correlations are used: for example, oblivious
transfer (\OT) correlations are used for two-party (semi-honest) secure
computation of Boolean circuits, and oblivious linear evaluation (\OLE{})
correlations, which generalize \OT over arbitrary fields, enable
2-party semi-honest secure computation of arithmetic circuits.
Eventually, $n$-party Beaver triples enable $n$-party semi-honest
secure computation of arithmetic circuits, and authenticated Beaver
triples enable maliciously secure computation of arithmetic circuits.

Since protocols in the correlated randomness paradigm are lightweight
and very efficient, they gave rise to a popular, two-stage approach:
first, the parties run an input-independent \emph{preprocessing
  phase}, which securely generates and distributes the correlated
strings, and second, these strings are consumed by an \emph{online}
protocol. Traditional approaches for implementing the preprocessing
phase had $\Omega(n)$ communication~\cite{C:IKNP03,C:DPSZ12,EC:KelPasRot18} and formed the efficiency
bottleneck of the overall protocol. The situation changed recently
with a new approach, introduced in~\cite{CCS:BCGIO17,CCS:BCGI18,C:BCGIKS19} and further refined in many
subsequent works~\cite{CCS:BCGIKRS19,CCS:SGRR19,C:BCGIKS20,FOCS:BCGIKS20,CCS:YWLZW20,C:CouRinRag21,C:BCGIKRS22}, with appealing efficiency features such as a one-time, $o(n)$-communication phase followed solely by local computation. At the heart of this approach is the notion of \emph{pseudorandom
  correlation generators} (\PCGs{}). Informally, a \PCG{} has two
algorithms: $\Gen(1^\secpar)$ outputs two \emph{short correlated keys}
$(\key_0,\key_1)$, and $R_\sigma \gets \Expand(\key_\sigma)$ stretches
$\key_\sigma$ into a long string $R_\sigma$, such that $(R_0, R_1)$ is
a pseudorandom instance of a target correlation. \PCGs{} enable an
efficient, two-stage \emph{silent} preprocessing phase:
\begin{enumerate}
  \item First, the parties securely distribute the short \PCG{} seeds,
        using a small amount of work and communication (often
        independent of the circuit size).
  \item Second, the parties locally stretch the \PCGs{} into long
        correlated pseudorandom strings: this part is the bulk of the
        computation, but does not require any further communication
        among the parties.
\end{enumerate}
This is the model of secure computation with \emph{silent
  preprocessing} (or \emph{silent secure computation} in short), where most of the preprocessing phase is pushed
offline. Previous works gave efficient constructions of \PCGs{} for
various correlations such as \OTs{}~\cite{CCS:BCGIKRS19,CCS:SGRR19,C:CouRinRag21,C:BCGIKRS22}, vector-\OLE{}~\cite{CCS:BCGI18}, \OLEs{}
over large fields~\cite{C:BCGIKS20}, authenticated Beaver triples~\cite{C:BCGIKS20} and many more. These \PCGs{} all build upon a common template, which combines
function secret sharing (\FSS{}) for simple function classes with
suitable variants of the syndrome decoding assumption.

\subsection{\PCGs{}: State of the Art and Challenges}

Very efficient constructions of \PCGs{} for the \OT{} correlations have been proposed~\cite{CCS:BCGIKRS19,CCS:SGRR19,C:CouRinRag21,C:BCGIKRS22}. The most recent constructions (see~\cite{C:CouRinRag21,C:BCGIKRS22}) allow to generate millions of random \OTs{} per second on one core of a standard laptop. Combined with the GMW protocol, they effectively enable extremely efficient two-party secure computation of Boolean circuits in the semi-honest model, with minimal communication in the preprocessing phase (a few dozen of kilobytes, independent of the circuit size), followed by cheap local computation, and a fast online phase (exchanging four bits per AND gate).

The situation, however, is much less satisfactory in essentially all other standard settings of secure computation, where the \OT{} correlation is not the best choice of correlation\footnote{While the \OT{} correlation is complete even for $N$-party malicious secure computation of arithmetic circuits, its use induces large overheads in the online phase: an $\Omega(N^2)$ communication overhead for handling $N$ parties, an $\Omega(\log^2|\FF|)$ overhead for handling larger fields $\FF$, and an $\Omega(\secpar)$ overhead for handling malicious parties. In contrast, other choices of correlated randomness can avoid each of these overheads.}, and one of the major open problems in this line of work is to improve this state of affair. Concretely, when targeting any one of \emph{multiparty} computation (with $N > 2$ parties), \emph{arithmetic} computation (for arithmetic circuits over a field $\F{}$ of size $|\F{}| > 2$), or \emph{malicious} security, the best-known \PCG{}-based solutions lag way behind the state of the art for 2-party, semi-honest secure computation of Boolean circuits. At a high level, the problem is twofold:
\begin{itemize}
  \item Secure computation of arithmetic circuits requires the \OLE{} correlation rather than the \OT{} correlation, and the constructions of~\cite{CCS:BCGIKRS19,CCS:SGRR19,C:CouRinRag21,C:BCGIKRS22} are inherently limited to the \OT{} correlation. To handle \OLE{}, a fundamentally different approach is required.
  \item Additionally, handling $N> 2$ parties or achieving malicious security both require the underlying \PCG{} for \OLE{} (or \OT{}) to satisfy a property known as \emph{programmability} (at a high level, programmability allows both to generate $N$-party correlations from $O(N^2)$ 2-party correlations, which is required because all known \PCGs{} are inherently restricted to the 2-party setting, and to \emph{authenticate} 2-party correlations with a MAC, which is needed for malicious security). Unfortunately, the constructions of~\cite{CCS:BCGIKRS19,CCS:SGRR19,C:CouRinRag21,C:BCGIKRS22} cannot (by design) achieve programmability.
\end{itemize}

These two limitations were addressed in the recent work of~\cite{C:BCGIKS20}, which introduced the first (reasonably efficient) construction of programmable \PCG{} for the \OLE{} correlation. While not as efficient as the best known \PCGs{} for \OT{}, it can produce around $10^5$ \OLEs{} per second on a standard laptop. However, the result of~\cite{C:BCGIKS20} suffers from two important downsides:
\begin{itemize}
  \item it can only produce \OLEs{} over \emph{large enough fields} (concretely, the field size must be larger than the circuit size). This leaves open the question of designing efficient programmable \PCGs{} for \OLE{} over small fields.
\item it relies on a relatively new \emph{ring-\LPN{} with splittable polynomial} assumption which states, in essence, that $(a, as+e)$ is hard to distinguish from $(a,b)$, where $a,b$ are random polynomials from a ring $\Ring = \Fp[X]/(P(X))$ where $P$ splits into $\deg(P)$ linear factors, and $s,e$ are random \emph{sparse} polynomials from $\Ring$. The ring-\LPN{} assumption was introduced a decade ago in~\cite{FSE:HKLPP12} to build efficient authentication protocols, and it has received some attention from the cryptography community~\cite{bernstein2012never,EPRINT:DamPar12,CANS:LipPav15,guo2015new,C:BCGIKS20,C:BomCouDeb22}. However, so far, we lack both a principled understanding of \emph{which} choice of the underlying polynomial $P$ yield solid instances (beyond the observation that reducible polynomials seem to enable more efficient attacks~\cite{guo2015new,C:BCGIKS20}), and a general methodology to argue the security of ring-\LPN{} assumptions.
\end{itemize}

At a high level, the construction of \PCG{} for \OLE{} from~\cite{C:BCGIKS20} proceeds by generating a single large pseudorandom \OLE{} correlation over a polynomial ring $\Ring = \Fp[X]/(P(X))$, assuming the hardness of the ring-\LPN{} assumption over $\Ring$. When $P$ splits into $N = \deg(P)$ linear factors, the Chinese Remainder Theorem permits to convert this large \OLE{} correlation over $\Ring$ into $N$ \OLE{} correlations over $\Fp$ (by reducing it modulo each of the factors of $P$). Note that the condition that $P$ splits requires $|\Fp| \geq N$, hence the restriction to large fields. Because the ring-\LPN{} assumption with a splittable polynomial is relatively new, the authors also provided a broad overview of its security against standard attacks and provided an ad-hoc analysis of the relation between the choice of the polynomial $P$ and the security strength of this assumption.

\subsection{Our Contributions}

In this work, we put forth and analyze a new general family of cryptographic assumptions related to the hardness of decoding codes defined over group algebras. A problem called \emph{quasi-abelian syndrome decoding} (\QACSD). Our family of assumptions builds upon quasi-abelian codes, a well-known family of codes in algebraic coding theory. It generalizes both the ring-\LPN{} assumption from~\cite{C:BCGIKS20} under some conditions on the underlying choice of polynomial and the quasi-cyclic syndrome decoding assumption.
The latter assumption was in particular used in several recent works~\cite{aragon2017bike,aguilar2018efficient,melchor2018hamming,CCS:BCGIKRS19}, including prominent submissions to the NIST post-quantum competition.
We show that working over group algebras presents several advantages:
\begin{enumerate}
\item a broad family of possible instantiations;
\item a rich structure that allows stronger security foundations;
\item a group algebra contains a canonical basis given by
  the group itself, providing a canonical notion of sparsity.
\end{enumerate}
Building on our new family of assumptions, we overcome both downsides of the recent work of~\cite{C:BCGIKS20} and obtain \PCGs{} for \OLEs{} over \emph{general fields} with \emph{solid security foundations}. In more details:

\subsubsection{A Template for Building New \PCGs{}.} We revisit and generalize the approach of~\cite{C:BCGIKS20} for building pseudorandom correlation generators for \OLE{} from ring-\LPN{}. We show that any choice of quasi-abelian code yields a \PCG{} for \OLE{} over a group algebra $\Ring$ under the corresponding \QACSD assumption. We identify natural instances of our framework such that the group algebra $\Ring$:
\begin{enumerate}
  \item supports fast operations via generalizations of the Fast Fourier Transform (which allows to achieve efficiency comparable to that of~\cite{C:BCGIKS20}), and
  \item is isomorphic to a product $\Fq \times \cdots \times \Fq$ of
        $N$ copies of $\Fq$ for arbitrary small $q>2$ and arbitrary
        large $N$ and therefore yields an efficient \PCG{} for
        generating $N$ copies of \OLE{} over $\FF_q$ for any $q> 2$.
\end{enumerate}
Therefore, we obtain new constructions of efficient programmable
\PCG{} over small fields, circumventing the main limitation of the
work of~\cite{C:BCGIKS20}. Our \PCGs{} enable for the first time
secure computation of arithmetic circuits over fields $\F{}$ of any
size $|\F{}|>2$ in the silent preprocessing model. This holds for two
or more parties, in the semi-honest or in the malicious setting. The
concrete efficiency of our construction is comparable to that
of~\cite{C:BCGIKS20} (we refer the reader to Table~\ref{tab:param} for
details on the seed size and stretch of our \PCGs{}). Concretely, our
costs are essentially identical, up to the fact that~\cite{C:BCGIKS20}
uses FFT's over cyclotomic rings, while our generalization to
arbitrary field relies on a generic FFT. Because FFT's over cyclotomic
rings have been thoroughly optimized in hundreds of papers, we expect
that using generic FFT's will be noticeably slower. Still, we identify
some concrete FFT-friendly choices of quasi-abelian codes where fast
FFT algorithms comparable to cyclotomic FFT's could in principle be
designed. We leave the concrete optimization of these FFT algorithms
to future work.

\subsubsection{Strong Security Foundations.} Building upon recent
results on the minimum distance of quasi-abelian codes, we give
evidence that the assumptions from our family cannot be broken by any
attack from the \emph{linear test
  framework}~\cite{FOCS:BCGIKS20,C:CouRinRag21}, a broad framework
that encompasses essentially all known attacks on \LPN{} and syndrome
decoding (including ISD, Gaussian elimination, BKW, and many more).
Our approach also sheds light on the security of the ring-\LPN{}
assumption. In essence, a conceptual message from our new approach is
that some choices of $P$ in the ring $\FF_q[X]/(P(X))$ yield an
instance of \QASD, and as such inherit our arguments of resistance
against linear attacks. In contrast, other (seemingly very similar)
choices of $P$ yield instances that are \emph{completely broken} by
linear attacks. This suggests that choosing instantiations of the
ring-LPN assumption should be done with care, and our framework yields
a way to do it with strong security guarantees.

\iftoggle{longversion}{As a contribution of independent interest,}{In a long version of the present article available on eprint,} we
  also complement our security analysis by showing, for all concrete
  instantiations of our framework that we use in our new \PCG{}
  constructions, a search-to-decision reduction for the underlying
  assumption. Therefore, we reduce the security of all our new \PCGs{}
  to (instances of) the \emph{search} $\QASD$ assumption.

  \subsubsection{The Case of $\FF_2$.} Perhaps intriguingly, the most
  natural way to instantiate our framework goes all the way to
  $\F{3}$, but breaks down over $\F{2}$. We prove a theorem that
  states that this is in fact inherent to the approach. Basically, the
  reason why the construction is not adaptable to $\F{2}$ is due to
  the fact that the product ring
  $\F{2}^N = \F{2} \times \cdots \times \F{2}$ has only one invertible
  element and hence can never be realised as a group algebra but in
  the irrelevant case of $N=1$. \iftoggle{longversion}{We then
      discuss}{In a long version of the present article available on
      eprint, we further discuss} a general methodology toward
    circumventing this limitation over $\F{2}$. While our approach
    falls short of providing a full-fledged solution, it highlights a
    possible avenue towards the intriguing goal of one day getting an
    efficient programmable \PCG{} for \OLEs{} over $\FF_2$.

\subsubsection{Applications.} Building upon our new programmable \PCGs{}, we obtain
\begin{itemize}
  \item (via Beaver triples) secure $N$-party computation of arithmetic circuits over $\FF_q$, for any $q > 2$, with silent preprocessing and communication $N^2\cdot \poly(\secpar)\cdot \log s$ bits (preprocessing phase) plus $2Ns$ field elements (online phase), where $s$ is the number of multiplication gates. The silent preprocessing phase involves $O(N\poly(\secpar)s\log s)$ work per party. For small numbers of parties, the $N^2\cdot \poly(\secpar)\cdot \log s$ is dominated by the $2Ns$ field elements for values of $s$ as low as $2^{25}$.
  \item (via circuit-dependent correlated randomness) secure $N$-party computation of a batch of $T$ arithmetic circuits over $\FF_q$, for any $q>2$, with silent preprocessing and communication $N^2\cdot \poly(\secpar)\cdot s \log T$ bits (preprocessing phase) plus $NTs$ field elements (online phase) , where $s$ is the number of multiplication gates in each circuit. The silent preprocessing phase involves $O(N\poly(\secpar)s T\log T)$ work per party.
\end{itemize}

As in~\cite{C:BCGIKS20}, our protocols extend to the malicious setting by generating \emph{authenticated} correlated randomness instead, which our \PCGs{} allow as well, and using a maliciously secure seed distribution protocol. Since the extension to authenticated correlated randomness and the seed distribution protocols in~\cite{C:BCGIKS20} are oblivious to the concrete choice of underlying ring $\Ring$, they directly apply to our new \PCGs{} from \QASD.

\subsection{Related Works}

Traditional constructions of \OLE{} protocols require communication for each \OLE{} produced. The work of~\cite{C:Gilboa99} requires $\Omega(\log |\FF|)$ string-\OTs{} per \OLE{}\footnote{This approach crucially requires structured \OTs{}, hence we cannot remove the communication by using pseudorandom \OTs{}.}. \OLEs{} can also be produced using state-of-the-art protocols based on homomorphic encryption~\cite{EC:KelPasRot18,CCS:HIMV19}, {\em e.g.} producing 64MB worth of \OLEs{} requires about 2GB of communication with Overdrive~\cite{EC:KelPasRot18}. A recent direct construction of \OLE{} from \RLWE{} has also been described in~\cite{SCN:BEPST20}. Using their construction, generating a batch of \OLEs{} has an amortized communication of about $8$ elements of $\FF$ over a large enough field.

\PCGs{} for \OLEs{} allow removing most of the communication overhead, by generating a large number of pseudorandom \OLEs{} using sublinear communication. The work of~\cite{C:BCGIKS20}, which is our starting point, has a computational cost comparable to that of recent \OLE{} protocols~\cite{EC:KelPasRot18}, but a considerably lower communication ; however, it only works over large fields. There has been several attempts to build \PCGs{} for \OLEs{} over small fields, but all suffer from severe downsides. The work of~\cite{C:BCGIKS19} describes a \PCG{} construction that combines BGV-based somewhat homomorphic encryption (under ring-LWE) and a new, ad-hoc variant of the multivariate quadratic assumption with sparse secrets. Their \PCG{}'s require very large seed sizes and are only efficient when generating huge batches (\cite{C:BCGIKS19} estimates about 7.000 \OLEs{} per second using a 3GB seed size when producing 17GB worth of triples).

In an appendix, the work of~\cite{C:BCGIKS20} shows that the standard
variant of syndrome decoding with quasi-cyclic code yields a \PCG{}
for \OLEs{} over arbitrary fields (including small fields). At a high
level, the construction uses the fact that given two pseudorandom
vectors $\vec x^\Trans = \bfH\cdot \vec e^\Trans_x$ and
$\vec y = \bfH\cdot \vec e^\Trans_y$, generating shares of their
pointwise products (\emph{i.e.} a batch of pseudorandom \OLE{}
correlations) reduces to generating shares of the diagonal of
$\vec x^\Trans \cdot \vec y = \bfH\cdot (\vec e^\Trans_x\cdot \vec e_y) \cdot \bfH^\Trans$,
and the term $(\vec e^\Trans_x\cdot \vec e_y)$ can be shared
efficiently with \FSS{} for point functions. However, the
computational cost of generating $n$ \OLEs{} this way scales as
$\Omega(n^2\log n)$ (ignoring $\poly(\secpar)$ factors), which makes
it entirely impractical in practice (the sublinearity in these
protocols only ``kicks in'' for values of $n$ above about $2^{30}$).

Eventually, two recent works on \PCGs{}~\cite{FOCS:BCGIKS20,C:BCGIKRS22} have introduced new variants of syndrome
decoding called respectively \emph{variable-density} and
\emph{expand-accumulate} \LPN{}. Each of these variants can actually be
used to construct programmable \PCGs{} for \OLE{} over small fields (though that was not their primary purpose: \VDLPN{} was introduced to construct
pseudorandom correlation \emph{functions}, and \EALPN{} to obtain more
efficient ``online-offline'' \PCGs{} for \OT{}). The intuition is that both assumptions can be formulated as the hardness of distinguishing $\bfH\cdot \vec e^\Trans$ from random, where $\bfH$ is a \emph{sparse} matrix, and the noise distribution is such that the term $(\vec e^\Trans_x\cdot \vec e_y)$ can still be shared efficiently using some appropriate \FSS{}. In this case, extracting the diagonal of $\bfH\cdot (\vec e^\Trans_x\cdot \vec e_y) \cdot \bfH^\Trans$ does not require computing the full square matrix, and scales only as $\poly(\secpar)\cdot \tilde{\Omega}(n)$. However, the hidden costs remain prohibitively large. Concretely, for both the \EALPN{} assumption and the \VDLPN{} assumption, the row-weight of $\bfH$ must grow as $\secpar \cdot \log n$~\cite{FOCS:BCGIKS20,C:BCGIKRS22,PKC:CouDuc23} (for some specific security parameter $\secpar$), hence the cost of generating $n$ \OLEs{} boils down to $\secpar^2\cdot \log^2 n$ invocations of an \FSS{} scheme, where the concrete security parameter $\secpar$ must be quite large: the recent analysis of~\cite{PKC:CouDuc23} estimates $\secpar \approx 350$. For $n = 2^{30}$, this translates to around $10^8$ invocation of an \FSS{} scheme for \emph{each} \OLE{} produced, which is nowhere near practical.

\subsection{Organization}

We provide a technical overview of our results in
Section~\ref{sec:overview}, and preliminaries in
Section~\ref{sec:preliminaries}. Section~\ref{sec:qac} is devoted to
introducing group algebras, quasi-abelian codes, and our new \QASD
family of assumptions. Section~\ref{sec:pcg} uses our new \QASD assumption to build
programmable \PCGs{}, adapting and generalising the template
of~\cite{C:BCGIKS20}.
\iftoggle{longversion}{Section~\ref{sec:cryptanalysis} covers the
    concrete security analysis of \QASD against various known attacks,
    and in particular against \emph{folding attacks}, which exploit
    the structure of the assumption to reduce the dimension of the
    instances. Finally, in Section~\ref{sec:applications} we elaborate
    on the applications of our new \PCGs{} to secure computation.
    Appendix~\ref{app:prelims} provides more detailed preliminaries on
    \FSS and \PCGs. Appendix~\ref{app:std} complements our study of
    \QASD by providing a search-to-decision reduction for the subset
    of the \QASD family used to construct our \PCGs.
    Appendix~\ref{sec:ntff} provides some background on function field
    theory, which is used in the analysis of some of our results.
    Appendix~\ref{sec:carlitz} adds background on the Carlitz module,
    which is at the heart of our (ultimately unsuccessful) attempt to
    extend our framework to \OLEs over $\FF_2$. Appendix~\ref{sec:f2}
    covers our approach for building \OLEs over $\FF_2$ and identifies
    the missing ingredient.}{In a longer version available on eprint
    we complement our study of \QASD by providing a search-to-decision
    reduction for the subset of the \QASD used to construct our \PCGs.
    We also further analyse the concrete security of \QASD against
    various known attacks and in particular against \emph{folding
      attacks} which exploit the structure of the assumption to reduce
    the dimension of the instances. Finally, in this longer version we
    also explore a potential way to extend our framework to the case
    of $\FF_{2}$, using number theory in function fields and the
    notion of Carlitz modules.}

\section{Technical Overview}
\label{sec:overview}

\subsection{Generating Pseudorandom Correlations: a Template}

A general template to construct \PCGs{} was put forth
in~\cite{CCS:BCGI18}, and further refined in subsequent works. At a
high level, the template combines two ingredients: a method that uses
\emph{function secret sharing} to generate a \emph{sparse} version of
the target correlation, and a carefully chosen linear code for which
the syndrome decoding problem is conjectured to be intractable. To
give a concrete example let us consider the task of generating an
\OLE{} correlation over a large polynomial ring $\Ring = \Fp[X]/(P)$,
where $P$ is some degree-$N$ split polynomial, and $\Fp$ is a field. In a
ring-\OLE{} correlation, each party $P_\sigma$ receives
$(x_\sigma,y_\sigma)\in\Ring^2$ for $\sigma=0,1$, which are random
conditioned on $x_0+x_1 = y_0\cdot y_1$.

\paragraph{Sparse correlations from \FSS{}.} Informally, \FSS{} for a
function class $\cF$ allows to share functions
$f:\bit^\ell \mapsto \GG$ (where $\GG$ is some group) from $\cF$ into
$(f_0,f_1) \gets \Share(f)$ such that
\begin{enumerate}[label=(\arabic*)]
\item $f_\sigma$ hides $f$ (computationally), and
\item for any
$x\in\bit^\ell$, $f_0(x) + f_1(x) = f(x)$.
\end{enumerate}
Since \FSS{} can always be achieved trivially by sharing the truth
table of $f$, one typically wants the shares to be compact
(\emph{i.e.} not much larger than the description of $f$). Efficient
constructions of \FSS{} from a length-doubling pseudorandom generator
are known for some simple function classes, such as \emph{point
  functions} (functions $f_{\alpha,\beta}$ that evaluate to $\beta$ on
$x= \alpha$, and to $0$ otherwise). \FSS{} for point functions can be
seen as a succinct way to privately share a long unit vector. More
generally, \FSS{} for $t$-point functions yield a succinct protocol
for privately sharing a long $t$-sparse vector.

An \FSS{} for multipoint functions immediately gives a strategy to
succinctly distribute a \emph{sparse} ring-\OLE{} correlation: sample
two random $t$-sparse polynomials $y_0, y_1$ (\emph{i.e.} polynomials
with $t$ nonzero coefficients in the standard basis), and define $f$
to be the $t^2$-point function whose truth table are the coefficients
of $y_0\cdot y_1$ (over $\Fp[X]$). Each party $P_\sigma$ receives
$\key_\sigma = (y_\sigma, f_\sigma)$, where $(f_0, f_1) = \Share(f)$.
With standard constructions of multipoint \FSS{}, the size of
$\key_\sigma$ is $O(t^2\cdot \log N)$ (ignoring $\secpar$ and $\log p$
terms): whenever $t$ is small, this is an exponential improvement over
directly sharing $y_0\cdot y_1$ (which would yield keys of length
$O(N)$).

\paragraph{From sparse to pseudorandom using syndrome decoding.} It
remains to convert the sparse correlation into a pseudorandom
correlation. This step is done non-interactively, by locally
\emph{compressing} the sparse correlation using a suitable linear
mapping. Viewing the compressed vector as the syndrome of a linear
code $\code{C}$ (the compressive linear mapping
is the parity-check matrix of $\code{C}$). The
mapping must satisfy two constraints: it should be \emph{efficient}
(linear or quasi-linear in its input size), and its output on a sparse
vector should be \emph{pseudorandom}. Fortunately, decades of research
on coding theory have provided us with many linear mappings which are
conjectured to satisfy the latter; the corresponding assumptions are
usually referred to as (variants of) the \emph{syndrome decoding}
(\SD{}) assumption, or as (variants of) the \emph{learning parity with
  noise} (\LPN{}) assumption\footnote{The name \LPN{} historically
  refers to the hardness of distinguishing oracle access to samples
  $(\vec a, \langle\vec a, \vec s\rangle + e)$ (for a fixed secret
  $\vec s$) from samples $(\vec a, b)$ where $\vec a, \vec s$ are
  random vectors, $e$ is a biased random bit, and b is a uniform
  random bit. This becomes equivalent to the syndrome decoding
  assumption when the number of calls to the oracle is \emph{a priori
    bounded}, hence the slight abuse of terminology. Since we will
  mostly use tools and results from coding theory in this work, we
  will use the standard coding theoretic terminology ``syndrome
  decoding'' to refer to the variant with bounded oracle access, which
  is the one used in all works on \PCGs{}.}.

Going back to our example, we will use two instances
$(x^0_\sigma, y^0_\sigma)_{\sigma\in\bit}$ and
$(x^1_\sigma, y^1_\sigma)_{\sigma\in\bit}$ of a sparse ring-\OLE{}
correlation. Fix a random element $a\getsr \Ring$. Each party
$P_\sigma$ defines
\iftoggle{longversion}{\[}{\(}
  y_\sigma \gets (1, \vec a) \cdot (y^0_\sigma, y^1_\sigma)^\intercal = y^0_\sigma + a\cdot y^1_\sigma \bmod P(X).
\iftoggle{longversion}{\]}{\)}
The assumption that $y_\sigma$ is indistinguishable from random is
known in the literature as the \emph{ring-\LPN{} assumption}, and has
been studied in several previous works~\cite{C:BCGIKS20} (for an
appropriate choice of $P$, it is also equivalent to the quasi-cyclic
syndrome decoding assumption, used in NIST submissions such as
BIKE~\cite{aragon2017bike} and HQC~\cite{melchor2018hamming}).
Furthermore, using FFT, the mapping can be computed in time
$\tilde{O}(N)$. Then, observe that we have
\[y_0y_1 = (y^0_0 + a\cdot y^1_0)\cdot (y^0_1 + a\cdot y^1_1) = y_0^0\cdot y^0_1 + a\cdot (y^0_1\cdot y^1_0 + y^1_0\cdot y^1_1) + a^2\cdot (y^1_0\cdot y^1_1),\]
where the polynomials
$y_0^0\cdot y^0_1, y^0_1\cdot y^1_0, y^1_0\cdot y^1_1,$ and
$y^1_0\cdot y^1_1$ are all $t^2$-sparse. Hence, each of these four
polynomials can be succinctly shared using \FSS{} for a $t^2$-point
function. Therefore, shares of $y_0y_1$ can be reconstructed using a
local linear combination of shares of sparse polynomials, which can be
distributed succinctly using \FSS{} for multipoint functions.

\paragraph{Wrapping up.} The final \PCG{} looks as follows: each party
$P_\sigma$ gets $(y^0_\sigma, y^1_\sigma)$ together with four \FSS{}
shares of $t^2$-point functions whose domain correspond to these four
terms. The \PCG{} key size scales as $O(t^2\log N)$ overall. Expanding
the keys amounts to locally computing the shares of the sparse
polynomial products (four evaluations of the \FSS{} on their entire
domain, in time $O(N)$) and a few $\tilde{O}(N)$-time polynomial
multiplications with $a$ and $a^2$ (which are public parameters).
Observe that when $P$ splits into $N$ linear factors over $\Fp[X]$, a
single pseudorandom ring-\OLE{} correlation as above can be locally
transformed into $N$ instances of pseudorandom \OLEs{} over $\Fp$:
this is essentially the construction of \PCG{} for \OLE{}
of~\cite{C:BCGIKS20}. However, this requires $p$ to be larger than
$N$, restricting the construction to generating \OLEs{} over large
fields. Furthermore, the requirement of a splitting $P$ makes the
construction rely on a less-studied variant of ring-\LPN{}.

\subsection{Quasi-Abelian Codes to the Rescue}

We start by abstracting out the requirement of the construction
of~\cite{C:BCGIKS20}. In coding theoretic terms, the hardness of
distinguishing $(a,a\cdot e + f)$ with sparse $(e,f)$ is an instance
of the (decisional) \emph{syndrome decoding problem} with respect to a
code with parity check matrix $(1,a)$. At a high level, and sticking
to the coding-theoretic terminology, we need a ring $\Ring$ such that
\begin{enumerate}
  \item the (decisional) syndrome decoding problem with respect to the
        matrix $(1,a)$ is intractable with high probability over the
        random choice of $a\getsr\Ring$;
  \item given \emph{sparse} elements $(e,f)$ of $\Ring$, it is
        possible to succinctly share the element $e\cdot f \in \Ring$;
  \item operations on $\Ring$, such as products, can be computed
        efficiently (\emph{i.e.} in time quasilinear in the
        description length of elements of $\Ring$);
  \item eventually, $\Ring$ is isomorphic to
        $\FF \times \cdots \times \FF$ for some target field $\FF$ of
        interest.
 \end{enumerate}

 We identify \emph{quasi-abelian codes} as a family of codes that
 simultaneously satisfy all the above criteria. At a high level, a
 quasi-abelian code of index $\ell$ has codewords of the form
 \iftoggle{longversion}{
     \[
       \{(\bfm\bfGamma_{1},\dots,\bfm\bfGamma_{\ell}) \mid \bfm = (m_{1},\dots, m_{\ell})\in(\Fq[G])^{k}\},
     \]
   }{$\{(\bfm\bfGamma_{1},\dots,\bfm\bfGamma_{\ell}) \mid \bfm = (m_{1},\dots, m_{\ell})\in(\Fq[G])^{k}\},$}
   where each $\bfGamma_i$ is an element of $\FF_q[G]^k$. Here,
   $\FF_q[G]$ denotes the \emph{group algebra}:
   \iftoggle{longversion}{\[\FF_q[G]\eqdef \left\{ \sum_{g\in G}a_{g}g \mid a_{g}\in \Fq \right\},\]}{
       $\FF_q[G]\eqdef \left\{ \sum_{g\in G}a_{g}g \mid a_{g}\in \Fq \right\}$,}
     where $G$ is a finite abelian group. Quasi-abelian codes
     generalise quasi-cyclic codes in a natural way: a quasi-cyclic
     code is obtained by instantiating $G$ with $\ZZ/n\ZZ$. We define
     the \emph{quasi-abelian syndrome decoding problem} (\QASD) as the
     natural generalisation of the syndrome decoding problem to
     quasi-abelian codes. This encompasses both quasi-cyclic syndrome
     decoding and plain syndrome decoding. The properties of
     quasi-abelian codes have been thoroughly studied in algebraic
     coding theory. We elaborate below on why quasi-abelian codes turn
     out to be precisely the right choice given our constraints 1--4
     above.

     \subsubsection{Security Against Linear Tests.} The linear test
     framework from~\cite{FOCS:BCGIKS20,C:CouRinRag21} provides a
     unified way to study the resistance of \LPN-style and syndrome
     decoding-style assumptions against a wide family of \emph{linear}
     attacks, which includes most known attacks on \LPN and syndrome
     decoding. We refer the reader to
     Section~\ref{subsec:linear_test_framework} for a detailed
     coverage. At a high level, in our setting, security against
     linear attacks boils down to proving that $(1, a)$ generates a
     code with large minimum distance. On one hand, a recent result of
     Fan and Lin~\cite{FL15} proves that quasi-Abelian codes
     asymptotically meet the Gilbert-Varshamov bound when the code
     length goes to infinity and the underlying group is fixed. On the
     other hand, Gaborit and Z\'emor \cite{GZ06} prove a similar
     result when the size of the group goes to infinity but restricted
     to the case where the group is cyclic. We conjecture an extension
     of Gaborit and Zémor result to arbitrary abelian groups. The
     latter conjecture entails that the \QASD problem cannot be broken
     by any attack from the linear test framework, for {any} choice of
     the underlying group $G$. This is the key to circumvent the
     restrictions of~\cite{C:BCGIKS20}.

 \subsubsection{Distribution of Products of Sparse Elements.} Using
 quasi-abelian codes, the ring $\Ring$ is therefore a group algebra
 $\FF_q[G]$. Now, given $e = \sum_{g\in G}e_{g}g$ and
 $f = \sum_{g\in G}f_{g}g$ any two $t$-sparse elements of $\Ring$
 (that is, such that $(e_g)_{g\in G}$ and $(f_g)_{g\in G}$ have
 Hamming weight $t$), the product $e\cdot f$ can be rewritten
 as
 \[
   e\cdot f = \sum_{e_g, f_{h} \neq 0} e_gf_h \cdot
   gh,
 \] which is a $t^2$-sparse element of the group algebra. In other
 words, the product of two sparse elements in a group algebra is
 always a sparse element. In the context of building \PCGs{}, this
 implies that we can directly distribute elements $ef \in \Ring$ using
 Sum of Point Function Secret Sharing ($\SPFSS$) for $t^2$-point
 functions. This allows us to generalise the template \PCG{}
 construction of~\cite{C:BCGIKS20} to the setting of arbitrary
 quasi-abelian code, with essentially the same efficiency (in a sense,
 the template is ``black-box'' in the ring: it only relies on the
 ability to distribute sparse elements via \FSS{}).

 We note that our generalised template differs slightly from the
 approach of~\cite{C:BCGIKS20}: in this work, the authors work over
 rings of the form $\Ring = \FF_p[X]/(P(X))$, where $P$ is some
 polynomial. However, in general, this ring is not a group algebra,
 and the product of sparse elements of $\Ring$ might not be sparse.
 They circumvented this issue by sharing directly the product over
 $\FF_p[X]$ (where the product of sparse polynomials remains sparse)
 and letting the parties reduce locally modulo $P$. Doing so, however,
 introduces a factor 2 overhead in the expansion (and a slight
 overhead in the seed size). Our approach provides a cleaner solution,
 using a structure where sparsity is natively preserved through
 products inside the ring.

 \subsubsection{Fast Operations on Group Algebras.} We observe that,
 by folklore results, operations over a group algebra $\FF_q[G]$ admit
 an FFT algorithm (using a general form of the FFT which encompasses
 both the original FFT of Cooley and Tuckey, and the Number Theoretic
 Transform). When using this general FFT, setting $G = \ZZ/2^t\ZZ$
 recovers the usual FFT from the literature. In full generality, given
 any abelian group $G$ of cardinality $n$ with $\gcd(n,q)=1$ and
 exponent $d$, if $\FF_q$ contains a primitive $d$-th root of unity,
 then the Discrete Fourier Transform and its inverse can be computed
 in time $O(n\cdot \sum_i p_i)$, where the $p_i$ are the prime factors
 appearing in the Jordan-Hölder series of $G$; we refer the reader to
 Section~\ref{subsec:fft} for a more detailed coverage. For several
 groups of interest in our context, this appears to yield very
 efficient FFT variants. For example, setting $q=3$ and
 $G = (\ZZ/2\ZZ)^d$, the resulting FFT is a $d$-dimensional FFT over
 $\FF_3$ and it can be computed in time $O(n\cdot \log n)$ (the group
 algebra $\FF_3[(\ZZ/2\ZZ)^d]$ is the one that yields a \PCG{} for $n$
 copies of \OLE{} over $\FF_3$).

 We note that FFT's over cyclotomic rings, such as those used
 in~\cite{C:BCGIKS20}, have been heavily optimised in hundreds of
 papers, due to their wide use (among other things) in prominent
 cryptosystems. As such, it is likely that even over ``FFT-friendly''
 choices of group algebras, such as $\FF_3[(\ZZ/2\ZZ)^d]$, the general
 FFT construction described above will be in practice significantly
 less efficient than the best known FFT's implementations over
 cyclotomic rings.  Hence, computationally, we expect that
 state-of-the-art implementations of the \PCG{} of~\cite{C:BCGIKS20}
 over large fields $\FF$ using a cyclotomic ring $\Ring$ for the
 ring-\LPN{} assumption will be noticeably faster than state-of-the-art
 implementations of our approach to generate \OLEs{} over a small
 field, such as $\FF_3$.  There is however nothing inherent to this:
 the efficiency gap stems solely from the years of effort that have
 been devoted to optimising FFT's over cyclotomic rings, but we expect
 that FFT's over other FFT-friendly group algebra such as
 $\FF_3[(\ZZ/2\ZZ)^d]$ could be significantly optimised in future
 works. We hope that our applications to silent secure computation
 over general fields will motivate such studies in the future.

 \subsubsection{From Quasi-Abelian Codes to \OLEs{} over $\FF_q$.} Our
 general \PCG{} template allows to generate a pseudorandom \OLE{} over
 an arbitrary group algebra $\FF_q[G]$. Then, when using
 $G = (\ZZ/(q-1)\ZZ)^d$, we have that $\FF_q[G] \simeq \FF_q^n$ (with
 $n = (q-1)^d$). Therefore, a single pseudorandom \OLE{} over
 $\FF_q[G]$ can be \emph{locally} converted by the parties into
 $(q-1)^d$ copies of a pseudorandom \OLE{} over $\FF_q$. Furthermore,
 for these concrete choices of $G$, we complement our security
 analysis by proving a search-to-decision reduction, showing that the
 decision \QASD problem over $\FF_q[G]$ with $G = (\ZZ/(q-1)\ZZ)^d$ is
 as hard as the \emph{search} \QASD problem. This provides further
 support for the security of our instantiations.

 In addition, our framework provides a way to investigate
 different instantiations of the ring-\LPN{} problem through the
 lens of quasi-abelian codes. This turns out to play an important role
 in understanding the basis for the security of ring-\LPN{}: seemingly
 very similar choices of the underlying polynomial can yield secure
 instances in one case, and completely broken instances in the other
 case. While the work of~\cite{C:BCGIKS20} gave a heuristic
 cryptanalysis of ring-\LPN{}, it fails to identify the influence of the
 choice of the polynomial.

 Concretely, consider the ring $\Ring = \FF_q[X]/(P(X))$ with either
 $P(X) = X^{q-1}-1$ or $P(X) = X^q - X$. The latter is a natural
 choice, as it has the largest possible number of factors over $\FF_q$
 (which controls the number of \OLEs{} produced over $\FF_q$).
 $\Ring = \FF_q[X]/(P(X))$ with $P(X) = X^{q-1}-1$ is a group algebra,
 and the ring-\LPN{} assumption with ring $\Ring$ reduces to
 $\QASD(\Ring)$. Hence, it is secure against all attacks from the
 linear test framework (and admits a search-to-decision reduction) by
 our analysis. On the other hand, ring-\LPN{} over the ring
 $\Ring = \FF_q[X]/(P(X))$ with $P(X) = X^{q}-X$ does not fit in our
 framework, and turns out to be \emph{completely broken} by a simple
 linear attack: given $(a,b)$ where $b$ is either random or equal to
 $a\cdot e+f \bmod X^q-X$, it holds that $e(0) = f(0) = 0 \bmod X^q-X$
 with high probability, because $e(0) = f(0) = 0$ over $\FF_q[X]$ with
 high probability (since $e,f$ are sparse, their constant coefficient
 is likely to be zero), and reduction modulo $X^q-X$ does not change
 the constant coefficient. Hence, the adversary can distinguish $b$
 from random simply by computing $b(0)$ (since $b(0)$ is nonzero with
 probability $(q-1)/q$ for a random $b$).

 The above suggests that settling for $\Ring = \FF_q[X]/(X^{q-1}-1)$
 is a conservative choice to instantiate the \PCG{}
 of~\cite{C:BCGIKS20} with strong security guarantees. We note
 that~\cite{C:BCGIKS20} recommended instead
 $\Ring = \FF_p[X]/(X^{n}+1)$ with $n$ being a power of $2$ and $p$ a
 large prime for efficiency reasons (since it is a cyclotomic ring, it
 admits fast FFT's). We believe that a natural generalisation of our
 framework should also encompass this ring, and allow proving that it
 also yields a flavor of ring-\LPN{} which is immune to linear attacks.
 However, this is beyond the scope of our paper, and we leave it to
 future work.

 \subsubsection{Considerations on the Case of $\FF_2$.} Interestingly,
 the aforementioned instance allows generating many \OLEs{} over
 $\FF_q$ for any $q > 2$; for $q = 2$, however, the term $n = (q-1)^d$
 becomes equal to $1$; that is, we only get a single \OLE{} over
 $\FF_2$ this way.  This is in fact inherent to our approach: the
 product ring $\FF_2^n$ has only one invertible element, and therefore
 can never be realised as a group algebra unless $n=1$. Hence,
 somewhat surprisingly, our general approach circumvents the size
 limitation of~\cite{C:BCGIKS20} and gets us all the way to $\FF_3$ or
 any larger field, but fails to provide a construction in the
 (particularly interesting) case $\FF_2$.

 Motivated by this limitation of our framework,
 \iftoggle{longversion}{}{in the long version} we devise a strategy to
   further generalise our approach through the theory of algebraic
   function fields (in essence, our generalisation is to quasi-abelian
   codes what quasi-negacyclic codes are to quasi-cyclic codes; we
   note that this is also close in spirit to the instance chosen
   in~\cite{C:BCGIKS20}: for their main candidate, they suggest using
   the ring $\Ring = \FF_p[X]/(X^{n}+1)$, which is a module over a
   group algebra and yields a {\em quasi-negacyclic code}). Alas, we
   did not manage to get a fully working candidate. At a (very) high
   level, our generalised framework produces pseudorandom elements
   $x = a\odot e_x + 1\odot f_x$ and $y = a\odot e_y + 1\odot f_y$
   where $e_x,e_y,f_x,f_y$ are sparse. However, the product $\odot$ is
   now \emph{not} the same product as the group algebra product
   $x\cdot y$. Concretely, to share $x\cdot y$, we need to share terms
   of the form $(u\odot e) \cdot (v \odot f)$ (where $u,v$ can be $a$
   or $1$). However, unlike the case of our previous instantiation,
   this does not rewrite as a term of the form $uv\cdot ef$ (which we
   could then share by sharing the sparse term $ef$, as $uv$ is
   public). Still, we believe that our approach could serve as a
   baseline for future works attempting to tackle the intriguing
   problem of building efficient programmable \PCGs{} for \OLE{} over
   $\FF_2$. In particular, our unsuccessful attempts show that to get
   such a \PCG{}, it suffices to find a way to succinctly share terms
   of the form $(u\odot e) \cdot (v \odot f)$ where $u,v$ are public,
   and $e,f$ are sparse. While \FSS{} do not provide an immediate
   solution to this problem, this reduces the goal to a ``pure MPC
   problem'' which could admit an efficient solution.

   \subsubsection{Concrete Cryptanalysis.}
   \iftoggle{longversion}{Eventually,}{In the long version of the
       present article} we complement our study by a concrete analysis
     of the security of our assumptions. As in previous works, the
     bounds derived from the resistance to linear attacks are quite
     loose, because they cover a \emph{worst-case} choice of linear
     attack. We cover standard attacks, such as information set
     decoding. A particularity of both ring-\LPN{} with splittable
     polynomial and our new family of \QASD assumption is that it
     grants the adversary some additional freedom: the adversary can,
     informally, transform a \QASD instance into an instance with
     reduced dimension (in the case of ring-\LPN{}, by reducing modulo
     factors of $P$; for $\QASD$, by quotienting by subgroups of $G$).
     This turns out to be equivalent to the concept of \emph{folding
       attacks}, which have been recently studied both in the context
     of code-based cryptography~\cite{CT19} and of lattice-based
     cryptography~\cite{BCV20}. We analyse the effect of folding
     attacks on our instances and discuss the impact on our parameter
     choices. In particular, the instances of \QASD used in our \PCG
     construction closely resemble the Multivariate \LWE assumption
     (with sparse noise instead of small-magnitude noise), which was
     shown in~\cite{BCV20} to be broken by folding attacks. We note
     (but this is well-known~\cite{CT19}) that folding attacks are
     much less devastating on \LPN- and syndrome decoding-style
     assumptions, essentially because folding yields a very slight
     increase of the noise magnitude in the \LWE setting (the sum of
     \LWE error terms has small magnitude), but increases the noise
     rate very quickly in the coding setting (the sum of sparse noises
     very quickly becomes dense).

\section{Preliminaries}
\label{sec:preliminaries}
\label{sec:prelimFSS}

\subsubsection{Function Secret Sharing.} Function secret sharing
(\FSS{}), introduced in~\cite{EC:BoyGilIsh15,CCS:BoyGilIsh16}, allows
to succinctly share functions. An \FSS{} scheme splits a secret
function $f:I\to\GG$, where $\GG$ is some Abelian group, into two
functions $f_0,f_1$, each represented by a key $K_0,K_1$, such that:
(1) $f_0(x) + f_1(x)=f(x)$ for every input $x\in I$, and (2) each of
$K_0,K_1$ individually hides $f$.

An \SPFSS is an \FSS scheme for the class of \emph{sums of point
  functions}: functions of the form $f(x) = \sum_if_{s_i,y_i}(x)$
where each $f_{s_i,y_i}(\cdot)$ evaluates to $y_i$ on $s_i$, and to
$0$ everywhere else. As in previous works, we will use efficient
constructions of $\SPFSS$ in our constructions of PCGs. Such efficient
constructions are known from any length-doubling pseudorandom
generator~\cite{CCS:BoyGilIsh16}. \iftoggle{longversion}{We refer the
    reader to Appendix~\ref{app:prelims} for more details on \FSS and
    \SPFSS.}{}

  \subsubsection{Pseudorandom Correlation Generators.} A pseudorandom
  correlation generator (\PCG{}) for some target ideal correlation
  takes as input a pair of short, correlated seeds and outputs long
  correlated pseudorandom strings, where the expansion procedure is
  deterministic and can be applied locally. In slightly more details,
  a $\PCG$ is a pair $(\Gen,\Expand)$ such that $\Gen(1^\secpar)$
  produces a pair of short seeds $(\key_0,\key_1)$ and
  $\Expand(\sigma, \key_\sigma)$ outputs a string $R_\sigma$. A \PCG
  is \emph{correct} if the distribution of the pairs $(R_0,R_1)$
  output by $\Expand(\sigma, \key_\sigma)$ for $\sigma=0,1$ is
  indistinguishable from a random sample of the target correlation. It
  is \emph{secure} if the distribution of
  $(\key_{1-\sigma}, R_\sigma)$ is indistinguishable from the
  distribution obtained by first computing $R_{1-\sigma}$ from
  $\key_{1-\sigma}$, and sampling a uniformly random $R_\sigma$
  conditioned on satisfying the target correlation with $R_{1-\sigma}$
  (for both $\sigma = 0$ and $\sigma = 1$). In this work, we will
  mostly consider the \OLE correlation, where the parties $P_0, P_1$
  receive random vectors $\vec x_0, \vec x_1 \in \FF^n$ respectively,
  together with random shares of $\vec x_0 * \vec x_1$, where $*$
  denotes the component-wise (\emph{i.e.} Schur) product.

Eventually, \emph{programmable} \PCGs allow generating multiple \PCG
keys such that part of the correlation generated remains the same
across different instances. Programmable \PCGs are necessary to
construct $n$-party correlated randomness from the $2$-party
correlated randomness generated via the \PCG. Informally, this is
because when expanding $n$-party shares (e.g. of Beaver triples) into
a sum of $2$-party shares, the sum will involve many ``cross terms'';
using programmable \PCGs allows maintaining consistent pseudorandom
values across these cross terms. \iftoggle{longversion}{We refer the
    reader to Appendix~\ref{app:prelims} for more details on \PCGs and
    programmable \PCGs.}{}

\iffalse
Given $t$ distributions $(\Dist_1, \cdots, \Dist_t)$ over $\F{2}^n$, we denote by $\bigoplus_{i\leq t} \Dist_i$ the distribution obtained by {\em independently} sampling $\vec v_i \getsr \Dist_i$ for $i=1$ to $t$ and outputting $ \vec v \gets\vec v_1 \oplus \cdots \oplus \vec v_t$. We will use the following bias of the exclusive-or (cf.~\cite{shpilka2009constructions}).
\begin{lemma}
\label{lemma:bias}
Let $t\in \N$ be an integer, and let $(\Dist_1, \cdots, \Dist_t)$ be $t$ independent distributions over $\F{2}^n$. Then $\bias( \bigoplus_{i\leq t} \Dist_i ) \leq 2^{t-1}\cdot \prod_{i=1}^t \bias(\Dist_i) \leq \min_{i \leq t} \bias(\Dist_i)$.
\end{lemma}

Finally, let $\Ber_r(\F{2})$ denote the Bernoulli distribution that outputs $1$ with probability $r$, and $0$ otherwise. More generally, we denote by $\Ber_r(\mathcal{R})$ the distribution that outputs a uniformly random nonzero element of a ring $\mathcal{R}$ with probability $r$, and $0$ otherwise.

We will use a standard simple lemma for computing the bias of a XOR of Bernoulli samples:

\begin{lemma}[Piling-up lemma]
\label{lemma:pilingup}
For any $0 < r < 1/2$ and any integer $n$, given $n$ random variables $X_1, \cdots, X_n$ i.i.d. to $\Ber_r(\F{2})$, it holds that $\Pr[\bigoplus_{i=1}^n X_i = 0] = 1/2 + (1-2r)^n/2$.
\end{lemma}
\fi

\subsection{Syndrome Decoding Assumptions}

The syndrome decoding assumption over a field $\F{}$ states,
informally, that no adversary can distinguish $(\bfH, \bfH\cdot \vec e^\Trans)$
from $(\bfH, \vec b^\Trans)$, where $\bfH$ is sampled from the set of parity-check matrices of some family  of linear
  codes, and $\vec e$ is a \emph{noise vector} sampled from some
distribution over $\F{}$-vectors and typically sparse. The vector
$\vec b$ is a uniform vector over $\F{}^n$. More formally, we define the
\SD assumption over a ring $\Ring$ with dimension $k$, code length $n$, w.r.t. a family $\code{F}$ of linear codes, and a
noise distribution $\Dist$:

\begin{definition}[Syndrome Decoding]
  Let $k, n\in \NN$, and let
  $\code{F} = \code{F}_{n,k}\subset \mathcal{R}^{(n-k)\times n}$ be a family of
  parity-check matrices of codes over some ring $\mathcal{R}$. Let
  $\Dist$ be a noise distribution over $\mathcal{R}^{n}$. The
  $(\Dist, \code{F}, \mathcal{R})\mhyphen\SD(k, n)$ assumption states
  that
  \begin{align*} \{(\bfH, \bfH\cdot\vec{e}^{\Trans}) \suchthat \bfH\getsr\code{F},
    \vec{e}\getsr\Dist \} \indist \{(\bfH, \vec{b}^\Trans) \suchthat \bfH\getsr\code{F}, \vec{b}\getsr\Ring^n\},
  \end{align*}
  where ``$\indist$'' denotes the computational indistiguishability.
\end{definition}

Denoting $t$ a parameter which governs the average density of nonzero
entries in a random noise vector, common choices of noise distribution
are Bernoulli noise (each entry is sampled from a Bernoulli
distribution with parameter $t/n$), exact noise (the noise vector is
uniformly random over the set of vectors of Hamming weight $t$), and
regular noise (the noise vector is a concatenation of $t$ random unit
vectors). The latter is a very natural choice in the construction of
pseudorandom correlation generators as it significantly improves
efficiency~\cite{CCS:BCGI18,C:BCGIKS19,CCS:BCGIKRS19} without harming
security (to the best of our knowledge; the recent work \cite{BO23}
being efficient for very low code rates, which is not our setting).

Many codes are widely believed to yield secure instances of the
syndrome decoding assumption, such as setting $\bfH$ to be a uniformly
random matrix over $\F{2}$ (the standard \SD{} assumption), the
parity-check matrix of an LDPC code~\cite{FOCS:Alekhnovich03} (the
``Alekhnovich assumption''), a
quasi-cyclic code (as used in several recent submissions to the NIST
post-quantum
competition, see e.g.~\cite{aragon2017bike,aguilar2018efficient,melchor2018hamming}
and in previous works on pseudorandom correlation generators, such
as~\cite{CCS:BCGIKRS19}), Toeplitz
matrices~\cite{FC:GilRobSeu08,C:LyuMas13} and more. All these variants
generalize naturally to larger fields (and are conjectured to remain
secure over arbitrary fields).

In the context of \PCG{}'s, different codes enable different
applications: advanced \PCG{} constructions, such as \PCG{}s for \OLE{},
require codes with structure. When designing new \PCG{}s, it is common to rely on syndrome decoding for codes which have not
been previously analyzed in the literature -- hence, unlike the ones
listed above, they did not withstand years or decades of
cryptanalysis. To facilitate the systematic analysis of new proposals,
recent works~\cite{FOCS:BCGIKS20,C:CouRinRag21} have put forth a
framework to automatically establish the security of new variants of
the syndrome decoding assumption against a large class of standard
attacks.

\subsection{The Linear Test
  Framework}\label{subsec:linear_test_framework}

The linear test framework provides a unified template to analyze the security of variants of the \LPN or syndrome decoding assumption against the most common attacks. It was first put forth explicitly in~\cite{FOCS:BCGIKS20,C:CouRinRag21} (but similar observations were implicit in many older works). Concretely, an attack against syndrome decoding in the linear test framework proceeds in two stages:
\begin{enumerate}
\item First, a matrix $\bfH$ is sampled from $\code{F}$, and fed to the (unbounded) adversary $\Adv$. The adversary returns a (nonzero) \emph{test vector} $\vec v = \Adv(\bfH)$.
\item Second, a noise vector $\vec e$ is sampled. The \emph{advantage} of the adversary $\Adv$ in the linear test game is the bias of the induced distribution $\vec v \cdot \bfH\cdot \vec e^\Trans$.
\end{enumerate}

\noindent To formalize this notion, we recall the definition of the bias of a distribution:

\begin{definition}[Bias of a Distribution] \label{def:bias} Given a
  distribution $\Dist$ over $\F{}^n$ and a vector
  $\vec u \in \F{}^n$, the bias of $\Dist$ with respect to $\vec u$,
  denoted $\bias_{\vec u}(\Dist)$, is equal to
\[\bias_{\vec u}(\Dist) = \left|\Prob_{\vec x \sim \Dist}[\vec u \cdot \vec x^\Trans=0] - \Prob_{\vec x \sim \UniDist_n}[\vec u\cdot \vec x^\Trans=0]  \right| = \left|\Prob_{\vec x \sim \Dist}[\vec u \cdot \vec x^\Trans=0] - \frac{1}{|\F{}|} \right|,\]
where
$\UniDist_n$ denotes the uniform distribution over $\F{}^n$. The bias
of $\Dist$, denoted $\bias(\Dist)$, is the maximum bias of $\Dist$
with respect to any nonzero vector $\vec u$.
\end{definition}

We say that an instance of the syndrome decoding problem is
\emph{secure against linear test} if, with very high probability over
the sampling of $\bfH$ in step $1$, for any possible adversarial choice
of $\vec v = \Adv(\bfH)$, the bias of $\vec v \cdot \bfH\cdot \vec e^\Trans$
induced by the random sampling of $\vec e$ is negligible. Intuitively,
the linear test framework captures any attack where the adversary is
restricted to computing a linear function of the syndrome
$\vec b^\Trans = \bfH\cdot \vec e^\Trans$, but the choice of the linear
function itself can depend arbitrarily on the code. Hence, the
adversary is restricted in one dimension (it has to be linear in
$\vec b^\Trans$), but can run in unbounded time given $\bfH$.

The core observation made in~\cite{FOCS:BCGIKS20,C:CouRinRag21} (and
also implicit in previous works) is that almost all known attacks
against syndrome decoding (including, but not limited to, attacks
based on Gaussian elimination and the BKW
algorithm~\cite{STOC:BluKalWas00,AR:L05,SCN:LevFou06,C:EssKubMay17}
and variants based on covering
codes~\cite{EC:ZhaJiaWan16,AC:BogVau16,bogos2016solving,JC:GuoJohLon20},
the ISD family of information set decoding
attacks~\cite{P62,stern1988method,AC:FinSen09,C:BerLanPet11,AC:MayMeuTho11,EC:BJMM12,EC:MayOze15,C:EssKubMay17,PQCRYPTO:BotMay18},
statistical decoding
attacks~\cite{al2001statistical,fossorier2006modeling,ACISP:Overbeck06,debris2017statistical},
generalized birthday attacks~\cite{C:Wagner02,EPRINT:Kirchner11},
linearization attacks~\cite{EC:BelMic97,INDOCRYPT:Saarinen07a},
attacks based on finding low weight code vectors~\cite{Zichron17}, or
on finding correlations with low-degree
polynomials~\cite{ITCS:ABGKR14,EPRINT:BogRos17}) fit in the above
framework. Therefore, provable resistance against linear test implies
security against essentially all standard attacks.

\subsubsection{Security Against Linear Tests.} Resistance against
linear test is a property of both the code distribution (this is the
``with high probability over the choice of $\bfH$'' part of the
statement) and of the noise distribution (this is the ``the bias of
the distribution induced by the sampling of $\vec e$ is low'' part of
the statement). It turns out to be relatively easy to give sufficient
conditions for resistance against linear tests. At a high level, it
suffices that
\begin{enumerate}
  \item the \emph{code generated by $\bfH$} has large minimum
        distance, and
  \item for any large enough subset $S$ of coordinates, with high
        probability over the choice of $\vec e$, one of the
        coordinates of $\vec e$ indexed by $S$ will be nonzero.
\end{enumerate}
The above characterization works for any noise distribution whose
nonzero entries are uniformly random over $\Ring\setminus\{0\}$, which
is the case for all standard choices of noise distributions. To see
why these conditions are sufficient, recall that the adversarial
advantage is the bias of $\vec v \cdot \bfH \cdot \vec e^\Trans$. By
condition (2), if the subset $S$ of nonzero entries of
$\vec v \cdot \bfH$ is sufficiently large, then $\vec e$ will ``hit'' one
of these entries with large probabilities, and the output will be
uniformly random. But the condition that $S$ is sufficiently large
translates precisely to the condition that $\vec v\cdot \bfH$ has large
Hamming weight for any possible (nonzero) vector $\vec v$, which is
equivalent to saying that $\bfH$ generates a code with large minimum
distance. We recall the formalization below:

\begin{definition}[Security against Linear
  Tests] \label{def:security_linear_test} Let $\Ring$ be a ring, and
  let $\Dist$ denote a noise distribution over $\Ring^n$. Let
  $\code{F} \subset \Ring^{(n-k)\times k}$ be a family of
  (parity-check matrices of) linear codes. Let
  $\eps, \eta: \N \mapsto [0,1]$ be two functions. We say that the
  $(\Dist,\code{F},\Ring)\mhyphen\SD(k,n)$ problem is
  \emph{$(\eps,\eta)$-secure against linear tests} if for any
  (possibly inefficient) adversary $\Adv$ which, on input $\bfH$
  outputs a nonzero $\vec v \in \Ring^n$, it holds that
\[\Pr[\bfH \getsr \code{F}, \vec v \getsr \Adv(\bfH)\;:\; \bias_{\vec v}(\Dist_{\bfH}) \geq \eps(\secpar) ] \leq \eta(\secpar),\]
where $\lambda$ denotes the security parameter and
$\Dist_{\bfH}$ denotes the distribution which samples
$\vec e \gets \Dist$ and outputs $\bfH\cdot\vec e^\Trans$.
\end{definition}

The \emph{minimum distance} of a matrix $\bfH$, denoted
$\mindist(\bfH)$, is the minimum weight of a nonzero vector in its row-span.
Then, we have the following straightforward lemma:

\begin{lemma} \label{lem:security_linear_test} Let $\Dist$ denote a
  noise distribution over $\Ring^n$. Let
  $\code{F}\subset\Ring^{(n-k)\times k}$ be a family of parity-check
  matrices of linear codes. Then for any integer $d\in \N$, the
  $(\Dist,\code{F},\Ring)\mhyphen\SD(k,n)$ problem is
  $(\eps_d,\eta_d)$-secure against linear tests, where
\[
    \eps_d = \max_{\HW{\vec v} > d}\bias_{\vec v}(\Dist), \quad \text{ and }
    \quad \eta_d = \Pr_{\bfH \getsr \code{F}}[\mindist(\bfH) \geq d].
\]
\end{lemma}

The proof is folklore, and can be found e.g. in~\cite{C:CouRinRag21}.
For example, using either Bernoulli, exact, or regular noise
distributions with expected weight $t$, for any $\vec v$ of weight at
least $d$, the bias against $\vec v$ is bounded by $e^{-2td/n}$.
Hence, if the code is a good code ({\em i.e.} $d = \Omega (n)$), the bias is
of the form $2^{-\Omega(t)}$.

\paragraph{When security against linear attacks does not suffice.}
There are two important cases where security against linear test does
not yield security against \emph{all} attacks.
\begin{enumerate}
  \item When the code is strongly algebraic. For example, Reed-Solomon
        codes, which have a strong algebraic structure, have high
        dual minimum distance, but can be decoded efficiently with the
        Welch--Berlekamp algorithm, hence they do not lead to a secure
        syndrome decoding instance (and indeed, Welch--Berlekamp does
        not fit in the linear test framework).
      \item When the noise is structured (e.g. for regular noise) and
        the code length is at least quadratic in the dimension. This
        opens the door to algebraic attacks such as the Arora-Ge
        attack~\cite{arora2011new} or the recent attack from
        Briaud and {\O}ygarden \cite{BO23}. However, when $n = O(k)$ (which is
        the case in all our instances), these
        attacks do not apply.
\end{enumerate}

The above are, as of today, the only known cases where security
against linear attacks is known to be insufficient. Algebraic decoding
techniques have a long history and are only known for very restricted
families of codes, and the aforementioned algebraic attacks typically
never applies in the $n = O(k)$ regime which we usually consider for
\PCG{}'s. Therefore, a reasonable rule of thumb is that a variant of
syndrome decoding yields a plausible assumption if (1) it provably
resists linear attacks, and (2) finding an algebraic decoding
algorithm is a longstanding open problem.

\section{Group Algebras and Quasi-Abelian Codes}
\label{sec:qac}

\subsection{Quasi-Abelian Codes}

Quasi-abelian codes have been first introduced in \cite{Wasan74}, and,
since then, have been extensively studied in coding theory.

\subsubsection{Group Algebras.} Let $\Fq$ denote the finite field with
$q$ elements, and let $G$ be a finite abelian group of cardinality
$n$. The group algebra of $G$ with coefficients in $\Fq$ is the free
algebra with generators $G$. More precisely, it is the set $\Fq[G]$ of
formal linear combinations
\[
  \Fq[G] \eqdef \left\{ \sum_{g\in G}a_{g}g \Bigm\vert a_{g}\in \Fq \right\},
\]
endowed with an $\Fq-$vector space structure in the natural way, and
the multiplication is given by the convolution:
\[
  \left(\sum_{g\in G}a_{g}g\right)\left(\sum_{g\in G}b_{g}g\right) \eqdef \sum_{g\in G}\left(\sum_{h\in G}a_{h}b_{h^{-1}g}\right)g.
\]
It is readily seen that $\Fq[G]$ is commutative if and only if the
group $G$ is abelian, which will always be the case in this article.

Once an ordering $g_{0},\dots,g_{n-1}$ of the elements of $G$ is
chosen, the group algebra $\Fq[G]$ is isomorphic (as an $\Fq$--linear
space) to $\Fq^{n}$ via
$\varphi\colon \sum_{i=0}^{n-1}a_{i}g_{i} \mapsto (a_{0},\dots,a_{n-1})$.
This isomorphism is not canonical since it depends on the ordering,
but changing it only leads to a permutation of the coordinates, and
many groups (especially {\em abelian} groups) come with a canonical
ordering. This isomorphism allows us to endow $\Fq[G]$ with the
Hamming metric, making $\varphi$ an {\em isometry}: The weight
$\wt{a}$ of $a\in\Fq[G]$ is defined as the Hamming weight of
$\varphi(a)$ (Note that changing the ordering of the group does not
impact the weight of an element, which is thus well-defined).

\begin{example}\label{example:group_algebras}
  The simplest example to have in mind is the case of cyclic groups.

  \begin{itemize}
    \item Let $G = \{1\}$ be the trivial group with one element. Then
          the group algebra $\Fq[G]$ is isomorphic to the finite field
          $\Fq$.
    \item Let $G = \ZZ/n\ZZ$ be the cyclic group with $n$ elements.
          Assuming that $q$ is coprime to $n$, it is easy to see that
          the group algebra $\Fq[G]$ is nothing else than the usual
          polynomial ring $\Fq[X]/(X^{n}-1)$. The isomorphism is given
          by $k\mapsto X^{k}$ extended by linearity.
  \end{itemize}
\end{example}

\begin{remark}
  The above example shows that our framework will only be a
  generalisation of known constructions. This generality will be
  crucial though, because all the instances we introduce in the
  present article and which will be proved to resist to linear attacks
  will arise from group algebras.
\end{remark}

Example~\ref{example:group_algebras} shows that the group algebra of a
cyclic group can be seen as a (quotient of a) polynomial ring in
one variable. For a general finite abelian group, this is not always
so simple, however there is also an explicit nice representation. This
uses the following standard fact from the theory of group algebras.
\begin{proposition}\label{prop:tensor_product_group_algebras}
  Let $G_{1}, G_{2}$ be two finite groups. Then
  \[
    \Fq[G_{1}\times G_{2}] \simeq \Fq[G_{1}] \otimes_{\Fq} \Fq[G_{2}].
  \]
\end{proposition}

\begin{example}\label{example:tensor_product_group_algebras}
  Let $G = \ZZ/n\ZZ\times \ZZ/m\ZZ$. Then,
  Proposition~\ref{prop:tensor_product_group_algebras} entails that
  \begin{align*}
    \Fq[G] = \Fq[\ZZ/n\ZZ] \otimes_{\Fq} \Fq[\ZZ/m\ZZ] & = \Fq[X]/(X^{n}-1)\otimes_{\Fq} \Fq[X]/(X^{m}-1) \\ & = \Fq[X, Y]/(X^{n}-1, Y^{m}-1).
  \end{align*}
  This isomorphism can actually be made explicit by
  $(k, \ell) \mapsto X^{k}Y^{\ell}$ extended by linearity. \end{example}

\begin{remark}\label{rem:quasi-ab=multivariate}
  More generally, since it is well--known that any finite abelian
  group $G$ is a product of cyclic group
  $\ZZ / d_1 \ZZ \times \cdots \times \ZZ / d_r \ZZ$, the previous
  statement asserts that the group algebra $\Fq[G]$ is isomorphic to a
  quotient of a multivariate polynomial ring, namely:
  \[
    \Fq [G] = \Fq [\ZZ / d_1 \ZZ \times \cdots \times \ZZ / d_r \ZZ]
    \simeq \Fq [X_1, \dots, X_r]/(X_1^{d_1}-1, \dots, X_r^{d_r}-1).
  \]
\end{remark}

\subsubsection{Quasi-Abelian Codes.} Let $\ell > 0$ be any positive
integer, and consider the free $\Fq[G]-$module of rank $\ell$:
\[
  (\Fq[G])^{\ell} \eqdef \Fq[G] \oplus \dots \oplus \Fq[G] = \Bigl\{ (a_{1}, \dots, a_{\ell}) \mid a_{i}\in \Fq[G]\Bigr\}.
\]
Any $\Fq[G]-$submodule of $(\Fq[G])^{\ell}$ is called a {\em
  quasi-group code} of index $\ell$ of $G$ (or quasi-$G$ code). When
the group $G$ is abelian, a quasi-$G$ code is called {\em
  quasi-abelian}. More precisely, given a matrix
\[
  \bfGamma = \begin{pmatrix}\gamma_{1,1}& \dots & \gamma_{1,\ell}\\
                                         \vdots & \ddots & \vdots \\
                                         \gamma_{k, 1} & \dots & \gamma_{k, \ell}
                                    \end{pmatrix}\in (\Fq[G])^{k\times \ell},
\]
the quasi-$G$ code defined by $\bfGamma$ is
\[
  \code{C} \eqdef \{\bfm\bfGamma = (\bfm\bfGamma_{1},\dots,\bfm\bfGamma_{\ell}) \mid \bfm = (m_{1},\dots, m_{\ell})\in(\Fq[G])^{k}\},
\]
where $\bfGamma_{i}$ denotes the column
$\begin{pmatrix}\gamma_{1, i} \\ \vdots \\ \gamma_{k, i}\end{pmatrix}$
and
$\bfm\bfGamma_{i} = m_{1}\gamma_{1, i}+\dots+m_{k}\gamma_{k, i} \in \Fq[G].$
The matrix $\bfGamma$ is said to be {\em systematic} if it is of the
form $\bfGamma = \left(I_{k} \mid \bfGamma'\right)$, where $\bfGamma'\in (\Fq[G])^{k\times (\ell-k)}$ and
$I_{k}\in(\Fq[G])^{k\times k}$ is the diagonal matrix with values
$1_{G}$.

Let $a\in\Fq[G]$ and choose an ordering $g_{0},\dots, g_{n-1}$ of the
elements of $G$. Through the aforementioned isomorphism $\varphi$, the
element $a$ can be represented as a vector
$(a_{0}, \dots, a_{n-1})\in \Fq^{n}$.  Now, consider the matrix
\[
  \bfA = \begin{pmatrix}
        \varphi(a\cdot g_{0}) \\
        \vdots \\
        \varphi(a \cdot g_{n-1})
      \end{pmatrix} \in \Fq^{n\times n},
\]
where each row is the vector representation of a shift of $a$ by some
element $g_{i}\in G$.  In short, the matrix $\bfA$ is the matrix
representing the multiplication--by--$a$ map $m \mapsto am$ in
$\Fq[G]$ in the basis $(g_0, \dots, g_{n-1})$. An easy computation
shows that for $m, a \in \Fq[G]$, the vector representation of the
product $m\cdot a$ is the vector-matrix product
\[
  \varphi(m)\bfA = (m_{0},\dots, m_{n-1})\begin{pmatrix}
                                  \varphi(a\cdot g_{0}) \\
                                  \vdots \\
                                  \varphi(a \cdot g_{n-1})
        \end{pmatrix}.
\]
In other words, any quasi-group code $\code{C}$ of index $\ell$ can be
seen as a linear code of length $\ell\times n$ over $\Fq$. The
$\Fq[G]-$module structure endows $\code{C}$ with an additional action
of the group $G$ on each block of length $n$; and $\code{C}$ (seen as
a linear code over $\Fq$), admits a generator matrix formed out by
$k\times \ell$ square blocks of size $n$.

\begin{example}
  Let us continue with Example \ref{example:group_algebras}.
  \begin{itemize}
    \item If $G = \{1\}$, then any linear code is a quasi-$G$ code.
    \item If $G = \ZZ/n\ZZ$ and $q$ is coprime to $n$. An element of
          $\Fq[G]\simeq \Fq[X]/(X^{n}-1)$ is a polynomial of degree at
          most $n$ which can be represented by the vector of its
          coefficients, and any product $m(X)\cdot a(X) \in \Fq[G]$
          can be represented by the {\em circulant} vector-matrix
          product

          \[
            \begin{pmatrix}
              m_{0} & m_{1} & \dots & m_{n-1}
            \end{pmatrix} \begin{pmatrix}
                                     a_{0} & a_{1} & \dots & a_{n-1} \\
                                     a_{n-1} & a_{0} & \dots & a_{n-2} \\
                                     \vdots & & & \vdots \\
                                     a_{1} & a_{n-1} & \dots & a_{0}
                                   \end{pmatrix}\in \Fq^{n}.
          \]

          For simplicity, assume that $k=1$ and $\ell = 2$. Then, a
          quasi-$\ZZ/n\ZZ$ code of index $2$ is defined over $\Fq$ by
          a double-circulant generator matrix

          \[
\left(
            \begin{array}{c|c}
              \begin{array}{cccc}
                                     a_{0} & a_{1} & \dots & a_{n-1} \\
                                     a_{n-1} & a_{0} & \dots & a_{n-2} \\
                                     \vdots & & & \vdots \\
                                     a_{1} & a_{n-1} & \dots & a_{0}
              \end{array} & \begin{array}{cccc}
                                     b_{0} & b_{1} & \dots & b_{n-1} \\
                                     b_{n-1} & b_{0} & \dots & b_{n-2} \\
                                     \vdots & & & \vdots \\
                                     b_{1} & b_{n-1} & \dots & b_{0}
                                   \end{array}
            \end{array}\right).
          \]

          In other words, a quasi-$\ZZ/n\ZZ$ code is nothing else than
          a usual {\em quasi-cyclic} code with block length $n$.
  \end{itemize}
\end{example}

\iftoggle{longversion}{ \iftoggle{longversion}{}{
In this section, we introduce missing results on quasi-abelian codes and group algebras. First, we prove that the dual of a quasi-abelian code is also a quasi-abelian code. Second, we show that group algebras admit fast operations via a generalization of the FFT, which enables efficient encoding with quasi-abelian codes. Eventually, we show that, for the class of quasi-abelian code used in our construction of \PCGs for \OLEs over $\FF_q\times \cdots \times \FF_q$, the search version of the \QASD problem reduces to the decision version. This provides further support for the security of our \PCG schemes, by showing that their security reduces to the \emph{search} \QASD assumption.
}
\subsection{Duality for Quasi-Abelian Codes}
When dealing with codes, it may be easier to use the language of
parity-check matrices, especially when considering random codes. In this section, we show that this also extends naturally to quasi-abelian
codes.

Let $G$ be an abelian group. The algebra $\Fq[G]$ is naturally endowed
with an inner product $\langle\cdot,\cdot\rangle$ defined as follows:
\[
  \left\langle \sum_{g\in G}a_{g}g, \sum_{g\in G}b_{g}g\right\rangle \eqdef \sum_{g\in G}a_{g}b_{g},
\]
which is simply the usual inner product over $\Fq^{n}$ (this does not
depend on the ordering of $G$). This inner product can be naturally
extended to $(\Fq[G])^{\ell}$:
\[
  \left\langle (a_{1},\dots,a_{\ell}),(b_{1},\dots, b_{\ell})\right\rangle \eqdef \sum_{i=1}^{\ell}\langle a_{i},b_{i}\rangle,
\]
and the notion of the dual $\code{C}^{\perp}$ of a code $\code{C}$
extends to quasi-abelian codes:
\[
  \code{C}^{\perp}\eqdef \left\{x \in (\Fq[G])^{\ell} \mid \langle x, c\rangle = 0 \quad \forall c\in \code{C}\right\}.
\]

\begin{proposition}
  Let $G$ be a finite abelian group and let $\code{C}$ be a quasi-$G$
  code of index $\ell$. Then $\code{C}^{\perp}$ is also a quasi-$G$
  code of index $\ell$.
\end{proposition}
\begin{proof}

  There needs only to prove that $\code{C}^{\perp}$ is kept invariant
  by the action of $\Fq[G]$.

  For any $a= \sum_{g\in G}a_{g}g\in\Fq[G]$, define
  $\bar{a}\eqdef \sum_{g\in G}a_{g}g^{-1}\in\Fq[G]$ and
  $\sigma(a)~\eqdef~a_{1_{G}}\in\Fq$ where $1_{G}$ denotes the
  identity element of $G$. The map $a\mapsto \bar{a}$ is clearly an
  automorphism of $\Fq[G]$ of order 2, and $\sigma:\Fq[G]\mapsto \Fq$
  is a linear form. Moreover, for $a, b \in \Fq[G]$, a simple
  computation shows that $\langle a, b \rangle = \sigma(a\bar{b})$.

  Now, let $x = (x_{1},\dots, x_{\ell})\in \code{C}^{\perp}$. For any
  $c = (c_{1},\dots, c_{\ell})\in\code{C}$ and any $a\in\Fq[G]$,
  \[
    \langle x\cdot a, c \rangle = \sum_{i=1}^{\ell}\sigma((x_{i}a)\bar{c_{i}}) = \sum_{i=1}^{\ell}\sigma(x_{i}\overline{(c_{i}\bar{a})}) = \langle x, c\cdot \bar{a}\rangle = 0,
  \]
  where in the last equality we used the fact that
  $c\cdot\bar{a}\in \code{C}$ since $\code{C}$ is an $\Fq[G]-$module.
  This concludes the proof of the proposition.\qed
\end{proof}

\begin{example}
  Consider a quasi-abelian code $\code{C}$ of index $2$, with a
  systematic generator matrix $\bfGamma = (1 \mid a).$ Then,
  $\code{C}$ admits a parity-check matrix of the form
  $\bfH = (\bar{a}\mid -1)$.
\end{example}

\subsection{Fast-Fourier Transform and Encoding}
\label{subsec:fft}

This Section recalls Fast Fourier Transform algorithms in a general
setting. This encompasses the usual FFT introduced by Cooley and
Tuckey in 1965\cite{CT65}\footnote{Although such an algorithm was
  already probably known by Gauss.} or the Number Theoretic Transform
(NTT) algorithm with which the reader might be more familiar. For a
detailed presentation in the group algebra setting
(see~\cite{Oberst07}).

Let $G$ be a finite abelian group\footnote{Recall than in this work we
  restrict ourselves to the abelian setting, though a Fourier
  Transform theory exists also for non-abelian group algebras, making
  use of the theory of characters.} of cardinality $n$, $\Fq$ a finite
field with $q$ elements, and consider the group algebra $\Fq[G]$. As
explained above, encoding a quasi-$G$ code amounts to computing
multiplications in $\Fq[G]$ which can be done using Discrete Fourier
Transform algorithms (DFT) when $\gcd(n, q)=1$. Indeed, in this case
Maschke theorem ensures that $\Fq[G]$ is semisimple, \ie $\Fq[G]$ is
isomorphic to a direct product of finite fields\footnote{This uses the
  abelianity of $G$, in general $\Fq[G]$ is a direct product of matrix
  algebras}, where the product is now done componentwise. DFT-based
algorithms to compute the products of two elements of $\Fq[G]$ always
follow the same strategy:

\begin{enumerate}
  \item Compute the forward map
        $\Fq[G] \rightarrow \F{q^{\ell_{1}}}\times \cdots \times \F{q^{\ell_{r}}}$
        \footnote{This is what is usually called the Discrete Fourier
        Transform.}.
  \item Compute the componentwise products.
  \item Compute the inverse map
        $\F{q^{\ell_{1}}}\times \cdots \times \F{q^{\ell_{r}}} \rightarrow \Fq[G]$.
\end{enumerate}
Fast Fourier Transform (FFT) algorithms correspond to the case where
steps $1$ and $3$ can be done efficiently (typically in $O(n\log(n))$
operations in $\Fq$ compared to a quadratic {\em naive} approach.)
This operation is all the more efficient when $\ell_{i}=1$ for all
$i$. This happens when $\Fq$ contains a primitive $d$-th root of unity,
where $d = \exp(G)$ is the {\em exponent} of $G$, \ie the {\em lcm} of
the orders of all elements of $G$. For our applications, this will
always be the case.

Recall that any finite group $G$ has a Jordan-Hölder composition
series:

\[
  \{1_{G}\} = G_{0} \lhd G_{1} \lhd \dots \lhd G_{r} = G
\]
such that the quotients $G_{i+1}/G_{i}$ (called the {\em factors} of
the series) are simple groups (\ie in the abelian setting they are
isomorphic to some $\ZZ/p_{i}\ZZ$ where $p_{i}$ is a prime.), and this
composition series is uniquely defined, up to equivalence (\ie all
Jordan-Hölder series have same length and the factors are the same up
to permutation).

\begin{proposition}[{\cite[Section 5]{Oberst07}}]\label{prop:FFT}
  Consider a finite abelian group $G$ of cardinality $n$ with
  $\gcd(n, q)=1$, and exponent $d$. Assume that $\Fq$ contains a
  primitive $d$-th root of unity. Let $p_{1},\dots,p_{r}$ denote all
  the primes (possibly non distinct) appearing in the Jordan-Hölder
  series of $G$ (in particular $n=p_{1}\cdots p_{r}$). Then the
  Discrete Fourier Transform (and its inverse) in $\Fq[G]$ can be
  computed in $O(n\times(p_{1}+\dots+p_{r}))$ operations in $\Fq$.
\end{proposition}

\begin{example}\label{example:FFT}
  Proposition~\ref{prop:FFT} encompasses well-known FFT's from the
  literature.
\begin{itemize}
\item   The usual FFT corresponds to $G = \ZZ/2^{t}\ZZ$. In this case, a
  composition series is given by
          \[
          G_{0} = \{0\} \subset \dots \subset G_{i} = 2^{t-i}\ZZ/2^{t}\ZZ \subset \dots \subset G_{t} = G = \ZZ/2^{t}\ZZ,
        \]
        and each factor $G_{i+1}/G_{i}$ is isomorphic to $\ZZ/2\ZZ$,
        and with the above proposition we recover the usual complexity in
        $O(2^{t}\times t) = O(n\log(n))$. However, $\Fq$ needs to be
        large enough to contain a primitive $2^{t}-$th root of
        unity\footnote{When the characteristic of $\Fq$ is not too
        large, an approach based on the Frobenius Fast Fourier
        Transform  can also be exploited to remove this fact.}.
        \vspace{\baselineskip}
  \item Consider the finite field $\F3$ and the group
  $G = (\ZZ/2\ZZ)^{t}$. Example~\ref{example:tensor_product_group_algebras} entails that
  \[
        \F3[G]\simeq \F3[X_{1},\dots,X_{t}]/(X_{1}^{2}-1,\dots, X_{t}^{2}-1).
        \]
        A composition series of $G$ is given by
        \[
        G_{0} = \{0\}^{t} \subset \dots \subset G_{i} = (\ZZ/2\ZZ)^{i}\times \{0\}^{t-i} \subset \dots \subset G_{t}= G = (\ZZ/2\ZZ)^{t},
      \]
        and the FFT can also be computed in time
        $O(2^{t}\times t) = O(n\log(n))$. This is nothing else than
        a $t$-dimensional FFT in $\F3$.
\end{itemize}
\end{example}

\begin{remark}
  Proposition~\ref{prop:FFT} is {\em asymptotic}, although efficient
  implementations exist for several groups and fields. They are
  particularly efficient when $G$ admits a Jordan-Hölder composition
  series with groups of index $2$, such as in the above two examples,
  which allows a simple divide-and-conquer approach. For a more
  precise description of Multivariate FFT algorithms
  (see~\cite[Section 2.2]{VLS13}).
\end{remark}

\iftoggle{longversion}{}{
\subsection{From Decision-\QASD to Search-\QASD}
\label{app:std}
In this section, we describe a reduction from the search version of
$\QASD$ to the decision version, in the concrete chosen instantiations
(all instances over $\Ring = \FF_q[G]$ where $G = (\ZZ/(q-1)\ZZ)^n$,
which is the group we use to obtain \PCGs for \OLEs over
$\FF_q^{(q-1)^n}$). This reduction is actually a natural extension to
that of \cite{C:BomCouDeb22} to the multivariate setting, and
essentially applies in extreme regime of low rate, \ie more in the
\LPN regime. In \cite{C:BomCouDeb22}, the authors introduced a new
problem they called Function Field Decoding Problem ($\FFDP$) that we
recall below, which is the analogue of \RLWE with function fields
instead of number fields (See Section~\ref{sec:ntff} for a quick
reminder on the theory of algebraic function fields).

Let $K/\Fq(T)$ be a function field with constant field $\Fq$ and ring
of integers $\OO_{K}$, and let $Q(T)\in\Fq[T]$ be irreducible. Let
$\gothP\eqdef Q\OO_{K}$ be the ideal of $\OO_{K}$ generated by $Q$.
\FFDP is parametrized by a secret element $\bfs\in\OO_{K}/\gothP$, and
a noise distribution $\psi$ over $\OO_{K}/\gothP$ which is a finite set.

\begin{definition}[\FFDP distribution] A sample
  $(\bfa, \bfb)\in \OO_{K}/\gothP\times \OO_{K}/\gothP$ is distributed
  according to the \FFDP distribution modulo $\gothP$, with secret
  $\bfs$ and noise distribution $\psi$ if
  \begin{itemize}
    \item $\bfa$ is uniformly distributed over $\OO_{K}/\gothP$;
    \item $\bfb=\bfa\cdot\bfs+\bfe$ where $\bfe$ is distributed
          according to $\psi$.
        \end{itemize}
        A sample drawn according to this distribution will be denoted
        by $(\bfa,\bfb)\getsr\mathcal{F}_{\bfs, \psi}$.
\end{definition}
In its search version, the goal of \FFDP is to recover the secret
$\bfs$ given access to enough samples.
\begin{definition}[\FFDP (search version)] Let
  $\bfs\in\OO_{K}/\gothP$, and let $\psi$ be a probability
  distribution over $\OO_{K}/\gothP$. An instance of \FFDP consists in
  an oracle giving access to independent samples
  $(\bfa, \bfb)\getsr \mathcal{F}_{\bfs, \psi}$. The goal is to
  recover $\bfs$.
\end{definition}
In its decision version, the goal is to distinguish between the \FFDP
distribution and the uniform over
$\OO_{K}/\gothP\times \OO_{K}/\gothP$.
\begin{definition}[\FFDP (decision version)] Let $\bfs$ be drawn {\em
    uniformly} at random in $\OO_{K}/\gothP$, and let $\psi$ be a
  noise distribution over $\OO_{K}/\gothP$. Define the following two distributions:
  \begin{itemize}
    \item $\mathcal{D}_{0}:(\bfa, \bfb)$ uniformly distributed over
    $\OO_{K}/\gothP\times\OO_{K}/\gothP$.
    \item $\mathcal{D}_{1}:(\bfa, \bfa\cdot\bfs+\bfe)$ distributed
          according to the \FFDP distribution
          $\mathcal{F}_{\bfs,\psi}$.
        \end{itemize}
        Let $b\in\{0, 1\}$. Given access to an oracle $\code{O}_{b}$
        providing independent samples from distribution
        $\mathcal{D}_{b}$, the goal of the decision \FFDP is to
        recover $b$.
      \end{definition}
      Recall that a distinguisher between two distributions
      $\code{D}_{0}$ and $\code{D}_{1}$ is a PPT algorithm $\code{A}$
      that takes as input an oracle $\code{O}_{b}$ corresponding to
      distribution $\code{D}_{b}$ with $b\in\{0,1\}$ and outputs a bit
      $\code{A}(\code{O}_{b})\in\{0, 1\}$. The distinguisher wins when
      $\code{A}(\code{O}_{b})=b$. Its distinguishing advantage is
      defined as:
\[
  \adv_{\code{A}}(\code{D}_{0}, \code{D}_{1}) \eqdef \dfrac{1}{2} \Bigl( \Prob(\code{A}(\code{O}_{b}) = 1 \mid b = 1) - \Prob(\code{A}(\code{O}_{b}) = 1 \mid b = 0) \Bigr)
\]
and satisfies
\[
  \Prob(\code{A}(\code{O}_{b}) = b) = \frac{1}{2} + \adv_{\code{A}}(\code{D}_{0}, \code{D}_{1}).
\]

      The crucial remark of \cite{C:BomCouDeb22} was to notice that
      some structured variants of the decoding problem could be
      somehow {\em lifted} to the function field setting, and could be
      directly seen as instances of \FFDP.

    \begin{example}
        Let $n\in\NN$ and consider the polynomial
        \[
          F(T, X)\eqdef X^{n}+T-1\in\Fq(T)[X].
        \]
        Eisenstein criterion proves that $F$ is irreducible over
        $\Fq[T]$. Define the function field
        \[
          K \eqdef \Fq(T)[X]/(F(T,X)).
        \]
        Computing partial derivatives shows that the curve defined by
        $F(T, X)$ is non-singular, and therefore $\Fq[T, X]/(F(T,X))$
        is the full ring of integers $\OO_{K}$ of $K$ (see for
        instance \cite[Chapter VII]{L21}). Now, let
        $Q(T)\eqdef T\in\Fq[T]$. Then,
        \[
          \OO_{K}/T\OO_{K} = \Fq[T, X]/(X^{n}+T-1, T) = \Fq[X]/(X^{n}-1) = \Fq[\ZZ/n\ZZ].
        \]
        Therefore, \QASD with the group $\ZZ/n\ZZ$ can be seen as an
        instanciation of \FFDP with the function field
        $K=\Fq(T)[X]/(X^{n}+T-1)$, and modulus $Q(T)=T$.
      \end{example}

      A general search-to-decision reduction for \FFDP would therefore
      immediately provide a search-to-decision reduction for many
      variants of \QASD. However, adapting the reduction of
      \cite{EC:LyuPeiReg10} the authors of \cite{C:BomCouDeb22} were
      only able to give such a reduction with addition algebraic
      constraints on $K$ and $\gothP$. More precisely, they gave the
      following theorem

      \begin{theorem}[Search to decision reduction for
        \FFDP{}]\label{thm:std}
        Let ${K}/{\Fq(T)}$ be a Galois function field of degree $n$
        with field of constants $\Fq$, and denote by $\OO_{K}$ its
        ring of integers. Let $Q(T)\in\Fq[T]$ be an irreducible
        polynomial. Consider the ideal $\mathfrak{P}\eqdef Q\OO_{K}$.
        Assume that $\mathfrak{P}$ does not ramify in $\OO_K$, and
        denote by $f$ its inertia degree. Let $\distrib$ be a
        probability distribution over ${\OO_{K}}/{\gothP}$, closed
        under the action of $\Gal({K}/{\Fq(T)})$, meaning that if
        $\bfe\sample\distrib$, then for any
        $\sigma \in \Gal(K/\Fq(T))$, we have
        $\sigma(\bfe)\sample\distrib$. Let
        $\bfs\in {\OO_{K}}/{\gothP}$.

        Suppose that we have an access to $\ffd{\bfs, \distrib}$
        and there exists a distinguisher between the uniform
        distribution over ${\OO_{K}}/{\goth{P}}$ and the \FFDP{}
        distribution with uniform secret and error distribution
        $\distrib$, running in time $t$ and having an advantage
        $\eps$. Then there exists an algorithm that recovers
        $\bfs\in {\OO_{K}}/{\gothP}$ (with an overwhelming probability
        in $n$) in time
        \[
          O\left( \frac{n^{4}}{f^{3}}\times \frac{1}{\varepsilon^{2}} \times q^{f \deg(Q)}\times t\right).
        \]
\end{theorem}

Unfortunately, not all group algebras arise from Galois extensions of
function fields. Nonetheless, based on the analogy between cyclotomic
number fields and the Carlitz modules, they proposed to instantiate
their reduction with $K=\Fq(T)[\Lambda_{T}]=\Fq(T)[X]/(X^{q-1}+T)$,
and modulus $Q(T)=T+1$. The theory of Carlitz extensions ensures that $\OO_{K}=\Fq[T][X]/(X^{q-1}+T)$, and therefore
\[
  \OO_{K}/(T+1)\OO_{K} = \Fq[T, X]/(T+1, X^{q-1}+T) = \Fq[X]/(X^{q-1}-1) = \Fq[\ZZ/(q-1)\ZZ].
\]

The key point is the fact that $\Gal(K/\Fq(T))=\Fq^{\times}$ and an
element $\zeta\in\Fq^{\times}$ acts on $P(X)\in\Fq[X]/(X^{q-1}-1)$ by
$\zeta\cdot P(X) \eqdef P(\zeta X)$. In particular, the Galois group
keeps invariant the support of {\em any} element, and therefore any
distribution that only depends on the weight is Galois invariant.
Theorem \ref{thm:std} immediately yields a search-to-decision
reduction for \QASD instantiated with the group $G=\ZZ/(q-1)\ZZ$.

\subsubsection{Extension to the multivariate setting.} Consider the
group $G~=~(\ZZ/(q-1)\ZZ)^{t}$, and let
$\mathcal{R}\eqdef \Fq[G] = \Fq[X_{1},\dots, X_{t}]/(X_{1}^{q-1}-1,\dots,X_{t}^{q-1}-1)$.
Using the heavy machinery of inverse Galois theory, it is possible to
find a Galois extension of $\Fq(T)$ with Galois group $G$. However,
this would induce a large overhead in the complexity of the reduction.
Instead, in this case, building on the case of $\Fq[\ZZ/(q-1)\ZZ]$, we
can directly describe the reduction and get Theorem~\ref{thm:std_inst}
from Section~\ref{subsec:std_inst}.

The reduction works as follows. Recall that by the Chinese Remainder
Theorem,
\[
  \mathcal{R}=\prod_{(\zeta_{1},\dots, \zeta_{t})\in(\Fq^{\times})^{t}}\Fq[X_{1},\dots, X_{t}]/(X_{i}-\zeta_{i}),
\]
and fix an ordering of $(\Fq^{\times})^{t}$, which yields an ordering
$\mathfrak{I}_{1},\dots, \mathfrak{I}_{r}$ of the ideals in the above
decomposition (where $r=(q-1)^{t}$):
\[
  \mathcal{R} = \prod_{i=1}^{r}\Fq[X_{1},\dots, X_{t}]/\mathfrak{I}_{i}.
\]

Let $w\in\{0,\dots, (q-1)^{t}\}$ and $\bfs\in\mathcal{R}$. Consider a
noise distribution $\psi=\psi_{w}$ over $\mathcal{R}$ such that
$\Expect[\HW{x}]=w$ when $x$ is sampled according to $\psi$. A sample
$(\bfa, \bfb)$ is distributed according to $\ffd{\bfs, \psi}$ if
$\bfa$ is uniformly distributed in $\mathcal{R}$, and
$\bfb=\bfa\cdot\bfs+\bfe$ where $\bfe\getsr \psi$.

The idea of the reduction is to recover the secret modulo one of the
factors, and then using the action of some group recover the full
secret. We keep a high level, the first steps of the reduction
following {\em exactly} the same path as that of \cite{C:BomCouDeb22}.
The only difference resides in the last step and the considered group
action.

\begin{description}
  \item[Step 1: Randomizing the secret.] In the {\em decision} version, the secret $\bfs$ is supposed to be {\em uniformly} distributed over $\mathcal{R}$, while in the {\em search} version, the secret is {\em fixed}. In other words, the decision version is an average-case problem, while the search version is worst-case. Fortunately, the secret can be easily randomized by sampling some $\bfs'$ uniformly at random in $\mathcal{R}$. Now, for each sample $(\bfa, \bfb)\sample\ffd{\bfs, \psi}$ with a fixed secret $\bfs$, we can build the sample $(\bfa, \bfb + \bfa\cdot\bfs')$ which is distributed according to $\ffd{\bfs+\bfs', \psi}$, and the secret is now uniformly distributed. Feeding the latter sample to a distinguisher, allows to creates a distinguisher for a fixed secret, with exactly the same advantage.
  \item[Step 2: Hybrid argument.] A sample $(\bfa, \bfb)$ is said to follow the hybrid distribution $\mathcal{H}_{i}$ if it is of the form $(\bfa', \bfb'+\bfh)$ where $\bfh$ is uniformly distributed modulo $\mathfrak{I}_{j}$ for $j\le i$, and is $0$ modulo $\mathfrak{I}_{j}$ for $j>i$. Such an $\bfh$ is easily constructed using the Chinese Remainder Theorem. In particular, $\mathcal{H}_{0}=\ffd{s,\psi}$ and $\mathcal{H}_{r}$ is the uniform distribution over $\mathcal{R}$. A simple hybrid argument proves that a distinguisher between $\mathcal{H}_{0}$ and $\mathcal{H}_{r}$ with advantage $\eps$ can be turned into a distinguisher between $\mathcal{H}_{i_{0}}$ and $H_{i_{0}+1}$ for some $i_{0}$, with advantage at least $\frac{\eps}{r}$.
  \item[Step 3: Guess and search.] Given $i_{0}$, the idea is to make a guess $g_{i_{0}}$ for $\bfs$ modulo $\mathfrak{I}_{i_{0}}$ and to use the previous distinguisher to tell whether this guess is correct, or not.
        Define $\bfg, \bfh$ and $\bfv\in\mathcal{R}$ such that:

        \begin{minipage}{0.45\textwidth}
        \[
        \bfg = \left\lbrace
        \begin{array}{ll}
          g_{i_{0}}& \mod \mathfrak{I}_{i_{0}}\\
          0 & \text{elsewhere}
        \end{array}
        \right.
        \]
        \end{minipage}
\begin{minipage}{0.45\textwidth}
        \[
        \bfv = \left\lbrace
        \begin{array}{ll}
          \text{random} & \mod \mathfrak{I}_{i_{0}}\\
          0 & \text{elsewhere}
        \end{array}
        \right.
        \]
        \end{minipage}\ \\

        and $\bfh$ is uniformly distributed modulo $\mathfrak{I}_{j}$ for $j\le i_{0}+1$ and $0$ elsewhere.
        Then, for each sample $(\bfa, \bfb)$, we can build the sample $(\bfa', \bfb')$ where
        \[
        \left\lbrace
        \begin{array}{l}
          \bfa'=\bfa+\bfv\\
          \bfb'=\bfb+\bfh+\bfv\cdot\bfg=\bfa'\cdot\bfs + \bfe +\bfh + \bfv(\bfg-\bfs).
        \end{array}\right.
        \]
        Using the fact that all the factors $\Fq[X_{1},\dots,X_{t}]/\mathfrak{I}_{j}$ are finite fields (isomorphic to $\Fq$), it is easily seen that $(\bfa', \bfb')$ is distributed according to $\mathcal{H}_{i_{0}}$ when the guess is correct, and according to $\mathcal{H}_{i_{0}+1}$ otherwise. Therefore, by the Chernoff-Hoeffding bound, using our distinguisher $\Theta(n(r/\eps)^{2})$ times, one can detect if the guess is correct or not with probability at least $1-2^{-\Theta(n)}$. An exhaustive search on $\Fq$ yields the value of $\bfs\mod\mathfrak{I}_{i_{0}}$.
  \item[Step 4: A group action] This is the only step that changes from the reduction of \cite{C:BomCouDeb22}. We need to find a group $\widehat{G}$  that will replace the Galois group of their reduction in permuting the factors.

        Inspired by the univariate example, let $\widehat{G}\eqdef (\Fq^{\times})^{t}$. It acts on $\mathcal{R}$ by:
        \[
        (\zeta_{1},\dots,\zeta_{t})\cdot P(X_{1},\dots, X_{t}) \eqdef P(\zeta_{1}X_{1},\dots, \zeta_{t}X_{t}).
        \]

        The key observation here is that
        \begin{enumerate}
          \item This action keeps invariant the support of elements in
                $\mathcal{R}$. In particular, the distribution $\psi$
                is invariant under the action of $\widehat{G}$.
                \item $\widehat{G}$ acts transitively on the factors:

                The action of $(\zeta_{1},\dots,\zeta_{t})$ maps the
                ideal $(X_{1}-\gamma_1,\dots,X_{t}-\gamma_{t})$ onto
                $(X_{1}-\zeta_{1}^{-1}\gamma_1,\dots,X_{t}-\zeta_{t}^{-1}\gamma_{t})$.
              \end{enumerate}

        Now, in order to recover $\bfs\mod\mathfrak{I}_{j}$ for $j\neq i_{0}$, it suffices to take the (unique) element $\bfz\in\widehat{G}$ such that $\bfz\cdot \mathfrak{I}_{j} = \mathfrak{I}_{i_{0}}$, and for any sample $(\bfa, \bfb=\bfa\bfs+\bfe)$ we can build $(\bfz\cdot \bfa, \bfz\cdot\bfb = (\bfz\cdot\bfa)(\bfz\cdot\bfs)+\bfz\cdot \bfe)$. Note that $\bfa'$ (resp. $\bfz\cdot\bfe$) is still uniformly distributed over $\mathcal{R}$ (resp. distributed according to $\psi$ since it is $\widehat{G}$-invariant). In other words, $(\bfa', \bfb')$ is distributed according to $\ffd{\bfz\cdot \bfs, \psi}$, and repeating the first three steps of the reduction will yield $\bfz\cdot\bfs\mod \mathfrak{I}_{i_{0}}$, which is equal to $\bfs \mod \bfz^{-1}\cdot \mathfrak{I}_{i_{0}} = \bfs \mod \mathfrak{I}_{j}$, which concludes the reduction.
\end{description}

 }
  }{
    \subsubsection{Duality.} In this work, we mainly use the language
    of parity-check matrices. The key point is that the dual of a
    quasi-$G$ code is still a quasi-$G$ code, and therefore any
    quasi-$G$ code will have a parity-check matrix of a similar shape
    than the generator matrix introduced above. We refer the reader to
    the long version of this article for an actual proof of this fact.

    \subsubsection{Fast-Fourier Transform and Encoding.}\label{subsec:fft}
    Our construction crucially uses the fact that some quasi-$G$ codes
    can be encoded very efficiently by means of a generalisation of
    the Fast-Fourier Transform (FFT) whose complexity only depends on
    the Jordan-Hölder series of $G$. We refer the reader to the long
    version of this article, and to~\cite{Oberst07} for all the
    details.
}

\subsection{The Quasi-Abelian Decoding Problem}\label{sec:qasd}

In this section, we introduce computationally hard problems related to
random quasi-abelian codes. They are variants of the Syndrome Decoding
Problem, restricted to this class of codes.

Let $G$ be a finite abelian group and $\Fq$ a finite field with $q$
elements. Given an integer $t\in\NN$, we denote by $\Dist_{t}(\Fq[G])$
a noise distribution over $\Fq[G]$ such that $\Expect[\HW{x}]=t$ when
$x\getsr \Dist_{t}$, and
$\Dist_{t,n}(\Fq[G])\eqdef \Dist_{t}(\Fq[G])^{\otimes n}$ will denote
its $n$-fold tensorization, \ie $\bfe\getsr \Dist_{t,n}(\Fq[G])$
is $\bfe\in\Fq[G]^{n}$ and its coordinates are drawn independently
according to $\Dist_{t}(\Fq[G])$. A {\em random} quasi-$G$ code of
index $2$, in {\em systematic form}, will be a quasi-$G$ code whose
parity-check matrix $\bfH\in(\Fq[G])^{2}$ is of the form
$\bfH = (\bf1 \mid \bfa)$, where $a$ is uniformly distributed over $\Fq[G]$.
Equivalently, it is the dual of the code generated by $\bfH$. The
search Quasi-Abelian Syndrome Decoding problem is defined as
follows:

\begin{definition}[(Search) \QASD{} problem]
  Given $\bfH = (\bf1 \mid \bfa)$ a parity-check matrix of a random
  systematic quasi-abelian code, a target weight $t\in \NN$ and a
  syndrome $\bfs\in\Fq[G]$, the goal is to recover an error
  $\bfe = (\bfe_{1}\mid \bfe_{2})$ with
  $\bfe_{i}\getsr \Dist_{t}(\Fq[G])$ such that
  $\bfH\bfe^{T} = \bfs$, \ie $\bfe_{1}+\bfa\cdot \bfe_{2} = \bfs$.
\end{definition}

The problem also has a decisional version.

\begin{definition}[(Decisional) \QASD{} problem]\label{def:QASD
    assumption}
  Given a target weight $t$, the goal of this decisional \QASD problem
  is to distinguish, with a non-negligible advantage, between the
  distributions
  \[
    \begin{array}{lcl}
      \mathcal{D}_{0}: &  (\bfa, \bfs)& \text{ where } \bfa, \bfs \getsr \Fq[G]\\
      \mathcal{D}_{1}: & (\bfa, \bfa\cdot\bfe_{1} + \bfe_{2}) & \text{ where } \bfa \getsr \Fq[G] \text{ and } \bfe_{i}\getsr \Dist_{t}(\Fq[G]).
    \end{array}
  \]
\end{definition}

Both assumptions above generalize immediately to the case of
parity-check matrices with more columns and/or rows of blocks.  When
$\bfH = (\bf1 \mid \bfa_1 \mid \cdots \mid \bfa_{c-1})$, for some
parameter $c$, this corresponds to what has been called {Module-\LPN}
in the literature. This corresponds to the hardness of syndrome
decoding for a quasi-abelian code of larger rate $(c-1)/c$.  We call
(search, decisional) $\QASD(c,\Ring)$ this natural generalization of
$\QASD$.

The \QASD assumption states that the above decisional problem should
be hard (for appropriate parameters). When the group $G$ is the
trivial group, this is the usual {\em plain} \SD-assumption, while
when the group $G$ is cyclic\footnote{and $\gcd(q, |G|)=1$}, this is
the \QCSD assumption at the core of Round 4 NIST submissions \BIKE and
\HQC. Those problems, especially their search version, have been
studied for over 50 years by the coding theory community and to this
day, no efficient algorithm is known to decode a random quasi-abelian
code.  This is even listed as an open research problem in the most
recent Encyclopedia of Coding Theory (from 2021)
\cite[Problem~16.10.5]{willems2021}.

\begin{remark}\label{rmk:systematic_form}
  In Definition~\ref{def:QASD assumption}, we consider quasi-abelian
  codes with a parity-check matrix in systematic form. Indeed, assume
  $\bfH=(\bfa_{1}\mid \bfa_{2})\in \Fq[G]^{1\times 2}$. A syndrome of
  $\bfH$ will be of the form
  $\bfs=~\bfa_{1}\bfe_{1}+\bfa_{2}\bfe_{2}$, and therefore is
  contained in the ideal $\mathcal{I}=(\bfa_{1}, \bfa_{2})$ of $\Fq[G]$
  generated by $\bfa_{1}$ and $\bfa_{2}$\footnote{Beware that $\Fq[G]$ is
    not necessarily principal.}. Therefore, when this ideal is {\em
    not} the full ring, there is an obvious bias. When working over a
  large field $\Fq$, elements of $\Fq[G]$ are invertible with high
  probability, and therefore $\mathcal{I}=\Fq[G]$ with overwhelming
  probability. On the other hand, this is not true anymore when
  working over small fields.  Using parity-check matrices in
  systematic form ensures that $1_{G}\in\mathcal{I}$, which removes
  the bias. This is a standard definition (see for instance
  \cite{AABBBBDGGGGMPRSTVZ22,AABBBBDDGLPRVZ22}), though not always
  formulated like that in the literature.
\end{remark}

\subsection{Security Analysis}

In this paragraph, we provide evidence for the \QASD-assumption. Note
first that for $G = \{1\}$ it is nothing but the $\SD$-assumption,
which is well established. Moreover, we argue for security of \QASD
against linear tests (Definition~\ref{def:security_linear_test}). With
Lemma~\ref{lem:security_linear_test} in hand, it suffices to show that
given the parity-check matrix $\bfH$ of a quasi-$G$ code $\code{C}$,
the minimum distance of the code {\em generated by} $\bfH$, \ie the
{\em dual} of $\code{C}$, is large with high probability (over the
choice of $\bfH$). Note that when $G=\{1\}$, it is well--known that
random codes are good, {\em i.e.} meet the Gilbert-Varshamov (GV)
bound (see for instance~\cite{P67,BF02,D23}).

\begin{proposition}[Gilbert-Varshamov] Let $0<\delta<1-\frac{1}{q}$.
  Let $\eps>0$, and let $\code{C}$ be a random code of rate
  $\frac{k}{n}\le(1-h_{q}(\delta)-\eps)$. Then,
  \[
    \Prob\left(d_{min}(\code{C})>\delta n\right) \ge 1-q^{-\eps n},
  \]
  where the probability is taken over the uniform choice of a
  generator matrix of $\code{C}$, and $h_{q}$ denotes the $q$-ary
  entropy function
  \[
    h_{q}(x) \eqdef -x\log_{q}\left(\frac{x}{q-1}\right)-(1-x)\log_{q}(1-x).
  \]
\end{proposition}
For the past 50 years, it has been a long trend of research in coding
theory to extend such a result for more general quasi-abelian codes.
For the class of quasi-cyclic codes which are, by far, the most used
quasi-abelian codes in cryptography, a GV-like bound was introduced by
Kasami in~\cite{K74}. Gaborit and Zémor even showed in~\cite{GZ06}
that various families of random double-circulant codes asymptotically
satisfied a logarithmic improvement on this bound. More recently, this
state of affairs was extended by Fan and Lin in \cite{FL15} to {\em
  any} quasi-abelian code, even in the modular case where
$char(\FF_{q})$ is {\em not} coprime to $|G|$. The proof of this
result makes use of the theory of representations of finite abelian
groups in $\Fq$.

\begin{theorem}[{\cite[Theorem 2.1]{FL15}}]\label{thm:GV_QA_code}
  Let $G$ be a finite abelian group, and let ${(\code{C}_{\ell})}_\ell$ be a
  sequence of random quasi-$G$ codes of length $\ell \in\NN$ and rate
  $r\in (0, 1)$. Let $\delta\in (0, 1-\frac{1}{q})$. Then,
  \[
    \lim_{\ell\to\infty}\Prob\left(\frac{d_{min}(\code{C}_{\ell})}{|G|} > \delta \ell\right)=
    \left\lbrace\begin{array}{ll}
                  1 & \text{ if } r < 1-h_{q}(\delta);\\
                  0 & \text{ if } r>1-h_{q}(\delta);
                \end{array}\right.
  \]
  and both limits converge exponentially fast. The above probability
  is taken over the uniform choice of a generator matrix
  $\bfG_{\ell}\in\Fq[G]^{k\times \ell}$ of $\code{C}_{\ell}$.
\end{theorem}

As it is often the case in coding theory, this result is stated
asymptotically, but the convergence speed could be made more precise,
the exponent depends on $|G|$: the larger the group $G$, the higher
this probability. Actually, to assert the resistance of \QASD against
linear attacks, it would be more relevant to consider the regime where
$k, \ell$ are constant and $|G|$ goes to infinity as it is done in
\cite{GZ06} but such a development is out of reach of this article and
we leave it as a conjecture.  There is a caveat though. Indeed, as it
was noticed in Remark \ref{rmk:systematic_form}, in the case of
constant $k, \ell$ and growing $|G|$ there is a bias in the \QASD
distribution when the ideal generated by the blocks in the input
parity-check matrix is not the full ring. This corresponds to the
parity-check matrix not being {\em full-rank} when seen as a matrix
over $\Fq[G]$. In this case, the minimum distance could drop, but
heuristically a random quasi-$G$ code will have a minimum distance
linear in its length as long as this bias is removed, which is the
case in our setting since we enforce the systematic form.

\begin{example}\label{example:linear_attack_group_codes} In order to produce \OLE's over the field $\FF_{p}$, ~\cite{C:BCGIKS20} proposed to use a ring $\Ring$ of  the form $\FF_{p}[X]/(F(X))$ where $F(X)$ is totally split in
  $\FF_{p}$.
  \begin{itemize}
  \item The choice of polynomial $F$ that maximizes the number of \OLE
    would be the polynomial $F(X)=X^{p}-X$ which has precisely all its
    roots in $\FF_{p}$ (This is {\em not} the choice recommended by
    the authors, but is still allowed in their framework). However,
    this ring does not fit in our setting, and in fact the \SD-problem
    in this ring is vulnerable to a very simple linear attack: given
    $(a,b)$ where $b$ is either random or equal to
    $a\cdot e+f \bmod X^p-X$, it holds that
    $e(0) = f(0) = 0 \bmod X^q-X$ with high probability, because
    $e(0) = f(0) = 0$ over $\FF_p[X]$ with high probability (since
    $e,f$ are sparse, their constant coefficient is likely to be
    zero), and reduction modulo $X^p-X$ does not change the constant
    coefficient. Hence, the adversary can distinguish $b$ from random
    simply by computing $b(0)$ (since $b(0)$ is nonzero with
    probability $(p-1)/p$ for a random $b$).
    \item However, by simply removing the $X$ factor and setting
          $F(X) = X^{p-1}-1$, which would yield $p-1$ copies of
          $\FF_{p}$ instead of $p$, the ring
          $\mathcal{R} = \FF_{p}[X]/(X^{p-1}-~1)$ is nothing else than
          the group ring $\FF_{p}[\FF_{p}^{\times}]$ and totally fits in
          our framework. In particular, it resists linear attacks.
          Note that the previous evaluation at $0$ does no longer make sense.
  \end{itemize}
\end{example}

\section{Pseudorandom Correlation Generators from \QASD}
\label{sec:pcg}

In the following we always consider $\mathcal{R} = \Fq[G] = \left\{ \sum_{g\in G}a_{g}g \mid a_{g}\in \Fq \right\},$ with $G$ an abelian group.
We refer to $\Ring_t$ as the set of ring elements of $\mathcal{R}$ of maximum weight $t$.

\subsection{A Template for Programmable \PCG for \OLE from \QASD}

\begin{theorem}
\label{pcg:thm_genericPCG}
    Let $\Ring = \FF_q[G] $. Assume that \SPFSS is a secure \FSS{} scheme for sums of point functions, and that the $\QASD(c,\Ring)$ assumption holds. Then there exists a generic construction scheme to construct a \PCG{} to produce one \OLE correlation (described on Fig.~\ref{pcg:figureOLE}). If the \SPFSS is based on a $\PRG : \{0,1\}^{\lambda} \rightarrow \{0,1\}^{2\lambda +2}$ via the \PRG{}-based  construction from~\cite{CCS:BoyGilIsh16}, we obtain:

\begin{itemize}
    \item[$\bullet$] Each party's seed has maximum size around :  $(ct)^2 \cdot ((\log |G| - \log t +1 ) \cdot (\lambda +2) + \lambda + \log q) + ct (\log |G| + \log q) $ bits
    \item[$\bullet$] The computation of \Expand can be done with at most $ (2 + \lfloor(\log q)/\lambda \rfloor ) |G| c^2 t$ \PRG{} operations, and $O(c^2 |G| \log |G|)$ operations in $\Fq$.
\end{itemize}

\end{theorem}
The protocol, adapted from the work of Boyle et al.~\cite{C:BCGIKS20},
is described on Fig.~\ref{pcg:figureOLE}. We first present an overview.
Remind that an instance of the \OLE correlation consists in giving a
random value $x_{\sigma} \in \mathcal{R}$ to party $P_{\sigma}$ as
well as an additive secret sharing of
$x_{0} \cdot x_{1} \in \mathcal{R}$ to both. Formally:
\begin{equation*}
    \left \{ ((x_0,z_0),(x_1,z_1)) | x_0,x_1,z_0 \overset{\$}{\leftarrow} \mathcal{R} , z_1 + z_0 = x_0 \cdot x_1 \right \} .
\end{equation*}

The core idea of the protocol is to give the two parties a random vector $\Vec{e_0}$ or $\Vec{e_1} \in \Ring_t^c$, where each element of the vector is sparse. In addition, parties have access to a vector $\Vec{a} = (1,\Vec{\dot{a}})$, with $\Vec{\dot{a}} = (a_1,\cdots,a_{c-1})$, a vector of random elements of $\Ring$. We see the vector $\bfe_{\sigma}$ of party $P_{\sigma}$ as an error vector. Using the vector $\Vec{a}$, parties can locally extend their error vector and construct $x_{\sigma} = \langle \Vec{a}, \Vec{e_{\sigma}} \rangle$, which is pseudorandom under $\QASD$.

We want to give the parties shares of $x_0 \cdot x_1$. Note that $x_0 \cdot x_1$ is a degree 2 function in $(\Vec{e_0}, \Vec{e_1})$; therefore, it suffices to share $\Vec{e_0} \otimes \Vec{e_1}$. We underline a property of the sparse elements in $\Ring_t$. Let $e,f$ be sparse elements. This means that there exist sets $S_e, S_f \subset G$, such that $e = \sum_{g\in S_e}e_{g}g , f = \sum_{g\in S_f}f_{g}g$ with $e_{g},f_{g}\in \Fq $ and $|S_e| =|S_f| = t \leq |G|$. It follows that the product of $e \cdot f$ can be expressed using only
\(
S_e \cdot S_f \eqdef \{gh ~|~ g \in S_e,~ h\in S_f\}
\)
as basis.
We conclude with $|S_e \cdot S_f| < |S_e| \cdot |S_f| = t^2$, to deduce that the product of sparse vectors in $\Ring$ also gives us sparse vectors (with sparsity $t^2$ instead of $t$). We note that here, we deviate from the original construction of~\cite{C:BCGIKS20}: over a ring of the form $\FF_q[X]/P(X)$ where $P$ is some polynomial, it is not generally true that the product of sparse elements remains sparse. This is circumvented in~\cite{C:BCGIKS20} by sharing the product over $\FF_q[X]$ instead, and reducing locally. When using group algebras as we do, the product preserves sparsity and we can share the product directly within $\FF_q[G]$, which is slightly more efficient.

This result enables us to express each element of
$\Vec{e_0} \otimes \Vec{e_1}$ as a sum of $t^2$ point functions. Then,
we rely on \iftoggle{longversion}{\SPFSS
    (Definition~\ref{def-spfss})}{Sum of Point Function Secret Sharing
    (\SPFSS~\cite[Definition~2.3]{C:BCGIKS20})}. Recall that an \SPFSS
  takes as input a sequence of points as well as a vector of values,
  and produces two keys that can be use to find shares of the sum of
  the implicit point functions. When a party evaluates its key at each
  point in the domain, it obtains a pseudorandom secret sharing of the
  coefficients of the sparse element in $\Ring_t$. The protocol uses
  $c^2$ elements of $\Ring_t$ as a result of the tensor product. This
  means that we need $c^2$ instances of \SPFSS for $t^2$ point
  functions. This gives us a seed size of
  $O(\lambda (ct)^2 \log |G|) = O(\lambda^3 \log
  |G|)$.

\begin{proof}[of Theorem \ref{pcg:thm_genericPCG}]
First, we argue the correctness of the protocol. The coefficient vectors $\Vec{b_{\sigma}^i},\Vec{A_{\sigma}^i}$ define a random element in $\Ring_t$. We can rewrite the product of two of these elements as follows:
\begin{equation*}
    e_0^i \cdot e_1^j = \sum_{k,l \in [0..t)} \Vec{b_0^i}[k] \cdot \Vec{b_1^j}[l] \Vec{A_{0}^i}[k] \Vec{A_{1}^j}[l].
\end{equation*}
This can indeed be described by a sum of point functions. From this point, $\Vec{u} = \Vec{u_0} + \Vec{u_1}$, then $\Vec{u} = \Vec{e_0} \otimes \Vec{e_1}$, each entry being equal to one of those $e_0^i \cdot e_1^j$.
The party obtains $z_{\sigma}$ as an output, and we can verify:
\begin{equation*}
    z_0 + z_1 = \langle \Vec{a} \otimes \Vec{a} , \Vec{u_0} + \Vec{u_1} \rangle = \langle \Vec{a} \otimes \Vec{a} ,  \Vec{e_0} \otimes \Vec{e_1}\rangle = \langle \Vec{a} \otimes \Vec{e_0} \rangle \cdot \langle   \Vec{a} \otimes \Vec{e_1}\rangle = x_0 \cdot x_1.
\end{equation*}
The next-to-last equality is straightforward to check.
Note that here, $\langle \Vec{a} ,\Vec{e_{\sigma}}\rangle$ is a $\QASD$ sample, with fixed random $\Vec{a}$ and independent secret $e_{\sigma}$. We now briefly show sketch security (the analysis is essentially identical to~\cite{C:BCGIKS20} since the construction is ``black-box'' in the ring $\Ring$, we sketch it for completeness). As the two cases are symmetrical, we assume $\sigma= 1$. Let $(k_0,k_1) \overset{\$}{\leftarrow} \PCG.\Gen(1^{\lambda})$ with associated expanded outputs $(x_0,z_0)$ and $(x_1,z_1)$, we need to show that
\begin{equation*}
    \{(k_1,x_0,z_0)\} \equiv \left \{(k_1,\Tilde{x_0},\Tilde{z_0})~|~\Tilde{x_0} \overset{\$}{\leftarrow} \Ring, \Tilde{z_0} = \Tilde{x_0} \cdot x_1 - z_1 \right \}.
\end{equation*}
To show this, we use a sequence of hybrid distributions.
\begin{itemize}
    \item Replace $z_0$ by $x_0 \cdot x_1 - z_1$.
    \item Step by step replace each the \FSS{} key $K_1^{i,j}$ in $k_1$ by a simulated key generated only with the range and the domain of the function. Due of the correctness and the security properties of the \FSS{} scheme, this distribution is indistinguishable from the original distribution.
    \item Replace $x_0$ by a fresh $\Tilde{x}_0$. It is also impossible to distinguish this distribution from the previous one, since the $K_1^{i,j}$ are now completely independent of $x_0$, and we can rely on the $\QASD$ assumption.
    \item Reverse step 2 by using the \FSS{} security property once again.  \qed
\end{itemize}

\end{proof}

Regarding the size of the different parameters, we use the optimization suggested in~\cite{C:BCGIKS20}, such as assuming that the $\QASD$ assumption holds also for \textit{regular error distributions} (we note that our proof of resistance against linear tests holds for very general noise distributions, and in particular for the regular noise distribution). We can thus reduce the seeds size to $(ct)^2 \cdot ((\log |G| - \log t +1 ) \cdot (\lambda +2) + \lambda + \log q) + ct (\log N + \log q)$ bits ; and the number of \PRG{} calls in \Expand down to $(2 + \lfloor(\log q)/\lambda \rfloor )  |G| c^2 t$. Note that to achieve security, choosing $ct = O(\lambda)$ is sufficient. The number of \PRG{} calls can be further reduced to $O(|G|c^2)$ using batch codes to implement the \SPFSS.

\begin{theorem}
    The \PCG construction for \OLE from Fig.~\ref{pcg:figureOLE} is programmable.
\end{theorem}

\begin{proof}
    In order to show that our \PCG{} is programmable we have to transform it a little, as the \Gen functionality takes additional inputs $(\rho_0,\rho_1)$ in the programmability definition. In our case, we can choose $\rho_{\sigma}= \left \{\Vec{A_{\sigma}^i}, \Vec{b_{\sigma}^i} \right\}$. In this way, as explained in the description of the protocol, the additional input of the players can be seen as a vector of elements in $\Ring_t$, $\Vec{e_{\sigma}} =  (\Vec{e_{\sigma}^0},\cdots,\Vec{e_{\sigma}^{c-1}})$. Because $x_{\sigma} = \langle \Vec{a},\Vec{e_{\sigma}} \rangle$, the players can compute their first input locally, after expanding their $\rho_{\sigma}$ into $e_{\sigma}$. This defines functions $\phi_{\sigma}$, and proves the programmability property.
    The proof of the correctness property is the same as in the proof of the Theorem \ref{pcg:thm_genericPCG}. The programmable security property can be proven with s sequence of hybrid distribution as in the proof of Theorem \ref{pcg:thm_genericPCG}, using the reduction to \FSS{} scheme and the $\QASD$ assumption.
    \qed

\end{proof}

\subsubsection{Distributed Seed Generation}

The protocol described in Fig.~\ref{pcg:figureOLE} assumes that a
trusted dealer has given the parties their seed. What we want to do in
practice is to achieve the \Gen phase via a distributive setup
protocol.

\begin{figure}[ht]

\fbox{

\begin{minipage}[c]{0.95\textwidth}

\begin{center}
\textbf{Functionality} $\QASD_{\sf \OLE{}-Setup}$
\end{center}

\textsc{Parameters}: Security parameter $1^{\lambda}$, $\PCG_{\OLE{}} = (\PCG_{\OLE{}}.\Gen, \PCG_{\OLE{}}.\Expand)$ as per Fig.~\ref{pcg:figureOLE}

\textsc{Functionality}:
\begin{enumerate}
    \item Sample $(k_0,k_1) \leftarrow \PCG_{\OLE{}}.\Gen(1^{\lambda})$.
    \item Output $k_{\sigma}$ to party $P_{\sigma}$ for $\sigma \in \{0,1\}$
\end{enumerate}

\end{minipage}
}
\caption{Generic functionality for the distributed setup of \OLE{}\ \PCG seeds}
\label{pcg:figureOLE_setup}
\end{figure}

\begin{figure}[ht]

\fbox{

\begin{minipage}[c]{0.95\textwidth}

\begin{center}
\textbf{Functionality} $\QASD_{\sf \OLE{}-All}$
\end{center}

\textsc{Parameters}: Security parameter, a group $G$, and a ring
$\Ring = \Fq[G]$.

\textsc{Functionality}:

If both parties are honest:
\begin{itemize}
    \item Sample $x_0,x_1 \leftarrow \Ring$
    \item Sample $z_0 \overset{\$}{\leftarrow} \Ring$ and let $z_1 = x_0 \cdot x_1 - z_0$.
    \item Output $(x_{\sigma},z_{\sigma})$ to party $P_{\sigma}$ for $\sigma \in \{0,1\}$.
\end{itemize}

If party $P_{\sigma}$ is corrupted:

\begin{itemize}
    \item Wait for input $(x_{\sigma},z_{\sigma}) \in \Ring^2$ from the adversary.
    \item Sample $x_{1-\sigma} \leftarrow \Ring$ and set $z_{1-\sigma} = x_0 \cdot x_1 - z_{\sigma}$
    \item Output $(x_{1-\sigma},z_{1-\sigma})$ to the honest party.
\end{itemize}

\end{minipage}
}
\caption{\OLE{} Functionality with Corruption}
\label{pcg:figureOLE_all}
\end{figure}

\begin{theorem}[From~\cite{C:BCGIKS20}]\label{pcg:theoremEstimateCostsOLEsetup}
  There exists a protocol securely realizing the functionality
  $\QASD_{\sf \OLE{}-Setup}$ of Fig.~\ref{pcg:figureOLE_setup} against
  malicious adversaries, with complexity:

\begin{itemize}
    \item[$\bullet$] Communication costs per party dominated by $(ct)^2 \cdot ((2 \lambda + 3 ) \log 2|G| + (9 t +2) \log (q-1))$.

    \item[$\bullet$] Computation is dominated by $2 |G|$ \PRG{} evaluations.
\end{itemize}

\end{theorem}

Taking $ct = O(\lambda)$ is enough to achieve exponential security.
With this we can conclude a general result:

\begin{theorem}
\label{pcg:theoremEstimateCostOLE-all}
    Let G be a group, and $\Ring = \Fq[G]$. Suppose that \SPFSS{} is a secure FSS scheme for sums of point functions, and the $\QASD(c,\Ring)$ assumption. Then there exists a protocol securely realizing the  $\QASD_{\sf{\OLE{}-All}}$ functionality over the ring $\Ring$ with the following parameters
   \begin{itemize}
    \item[$\bullet$] Communication costs and size of the seed : $O(\lambda^3 \log |G|)$.
    \item[$\bullet$] Computation costs : $O(\lambda |G|)$ \PRG{} evaluations and $O(c^2 |G|\log |G|)$ operations in $\Fq$.
\end{itemize}
\end{theorem}

\iftoggle{longversion}{\begin{figure}[!htbp]}{\begin{figure}[!h]}

\fbox{
\begin{minipage}[c]{0.95\textwidth}

\begin{center}
\textbf{Construction } $\QASD_{\sf \OLE }$
\end{center}
\textsc{Parameters}: Security parameter $\lambda$, noise weight
$t = t(\lambda)$, compression factor $c \geq 2$, $G$ a finite abelian
group, $\Ring = \FF_q[G]$. An \FSS scheme
(\SPFSS.\Gen,\SPFSS.\FullEval) for sums of $t^2$ point functions, with
domain $[0..|G|)$ and range $\Fq$.

\textsc{Public Input}: $c-1$ random ring elements $a_1, \cdots, a_{c-1} \in \Ring$.

\textsc{Correlation}: After expansion, outputs $(x_0,z_0) \in \Ring^2$ and $(x_1,z_1) \in \Ring^2$ where $z_0 + z_1 = x_0 \cdot x_1$

\textbf{Gen} : On input $1^{\lambda}$ :
\begin{enumerate}[itemsep=3mm]
    \item For $\sigma \in \{0,1\}$ and $i \in [0 .. c)$, sample random vectors $\Vec{A_{\sigma}^{i}}\leftarrow (g_1,\cdots,g_t)_{g_i \in G}$ and $\Vec{b_{\sigma}^{i}}  \leftarrow (\FF_q^*)^t$.
    \item For each $i, j \in [0..c)$, sample \FSS{} keys \\ $(K_0^{i,j}, K_1^{i,j})  \overset{\$}{\leftarrow} \SPFSS.\Gen(1^{\lambda}, \Vec{A_0^i} \otimes \Vec{A_1^j} , \Vec{b_{0}^{i}} \otimes \Vec{b_{1}^{j}}).$
    \item Let $\textsf{k}_{\sigma} = ((K_{\sigma}^{i,j})_{i,j \in [0 .. c)},((\Vec{A_{\sigma}}^{i},\Vec{b_{\sigma}^{i}})_{i \in [0 .. c)})$.
    \item Output $(\textsf{k}_0,\textsf{k}_1)$.
\end{enumerate}
\textbf{Expand} : On input $(\sigma,\textsf{k}_{\sigma}$) :
\begin{enumerate}[itemsep=3mm]
    \item Parse $\textsf{k}_{\sigma}$ as $((K_{\sigma}^{i,j})_{i,j \in [0 .. c)},((\Vec{A_{\sigma}^{i}},\Vec{b_{\sigma}^{i}})_{i \in [0 .. c)})$
    \item For $i \in [0 .. c)$, define the element of $\Ring_t$
    \begin{equation*}
        e_{\sigma}^i = \sum_{j \in [0..t)} \Vec{b_{\sigma}^{i}}[j] \cdot \Vec{A_{\sigma}^{i}}[j].
    \end{equation*}
    \item Compute $x_{\sigma} = \langle \Vec{a},\Vec{e_{\sigma}} \rangle$, where $\Vec{a} = (1,a_1,\cdot,a_{c-1}), \Vec{e_{\sigma}} = (e_{\sigma}^0, \cdots,e_{\sigma}^{c-1})$.
    \item For $i,j \in [0..c)$, compute $u_{\sigma, i + c j} \leftarrow \SPFSS.\FullEval (\sigma, K_{\sigma}^{i,j})$ and view it as a $c^2$ vector of element in $\Ring_{t^2}$.
    \item Compute $z_{\sigma} = \langle \Vec{a} \otimes \Vec{a}, \Vec{u_{\sigma}} \rangle $.
    \item Output $x_{\sigma},z_{\sigma}$    .
\end{enumerate}
\end{minipage}
}
\caption{\PCG for \OLE over $\Ring$, based on \QASD}
\label{pcg:figureOLE}
\end{figure}

\subsection{Instantiating the Group Algebra}

In this section we instantiate our general result with a concrete
construction of a \PCG for \OLE correlation over $\Fq$. Remind that
$G = \prod_{i=1}^n \ZZ/q_i \ZZ$, $q_i \geq 2$ . Using Proposition
\ref{prop:tensor_product_group_algebras} from previous section:
\begin{align*}
    \Fq[G] &= \Fq\left[\prod_{i=1}^n \ZZ/q_i \ZZ\right] \simeq \Fq[\ZZ/q_1 \ZZ] \otimes_{\Fq} \cdots \otimes_{\Fq} \Fq[\ZZ/q_n \ZZ]\\
    &\simeq   \bigotimes_{i = 1}^n \Fq[X_i]/(X_1^{q_i}-1)  \simeq\Fq[X_1,.., X_n] / (X_1^{q_1}-1,..,X_n^{q_n}-1 )  .
\end{align*}

\subsubsection{Batch-\OLE{} over $\FF_q$.}

In the following we let all the $q_i$ be all equal to $q-1$. We therefore use $ \Ring = \Fq[G] \simeq\Fq[X_1,.., X_n] / (X_1^{q-1}-1,..,X_n^{q-1}-1 ) $. Remark that the elements of $\Fq^*$ are the roots of the polynomial $X_i^{q-1}-1$. Therefore, we can write $X_i^{q-1} -1 = \prod_{a \in \Fp^*} (X_i-a)$, for all $1 \leq i \leq n$ and, by the Chinese Remainder Theorem, we get
\begin{equation*}
       \Fq[X_1,.., X_n] / (X_1^{q-1}-1,..,X_n^{q-1}-1 ) \simeq  \prod_{i=1}^T \Fq .
\end{equation*}
where $T = (q-1)^n$ is the number of elements in the group. We can
apply our protocol to construct a \PCG for the \OLE correlation in
$\Ring $. This single \OLE over $\Ring$ can be transformed in $T$
different instances of \OLE over $\Fq$. We get:

\begin{theorem}
  Suppose that \SPFSS is a secure \FSS{} scheme for sums of point
  functions and that the \QASD assumption holds. Let
  $\Ring = \Fq[X_1,.., X_n] / (X_1^{q-1}-1,..,X_n^{q-1}-1 )$, and
  $T = (q-1)^n$. We can construct a \PCG producing $T$ instances for
  \OLE over $\Fp$, using the $\QASD_{\OLE}$ construction.  The
  parameters we obtain are the following.
\begin{itemize}
 \item[$\bullet$] Each party's seed has size at most:  $(ct)^2 \cdot ((n \log (q-1) - \log t +1 ) \cdot (\lambda +2) + \lambda + \log q) + ct (n \log (q-1) + \log q) $ bits
    \item[$\bullet$] The computation of \Expand can be done with at most $ (2 + \lfloor(\log q)/\lambda \rfloor ) n \log (q-1) c^2 t$ \PRG{} operations, and $O(c^2 (q-1)^n n \log (q-1))$ operations in $\Fq$.
\end{itemize}

\end{theorem}

\subsubsection{Concrete Parameters.} We report on
Table~\ref{tab:param} a set of concrete parameters for our new
programmable PCGs from $\QASD$, when generating $T$ instances of a
pseudorandom OLE over $\FF_q$, chosen according to the analysis of
\iftoggle{longversion}{of Section~\ref{sec:cryptanalysis}.}{the long
    version on eprint.} We note that our concrete security parameters
  are very close to the parameters of~\cite{C:BCGIKS20}. This stems
  from two points:

First, \cite{C:BCGIKS20} conservatively chose security bounds based on existing attacks over $\FF_2$, even though their instantiation is over $\FF_p$ with $\log p \approx 128$ (and known attacks on syndrome decoding are less efficient over larger fields). One of the reason for this was to get conservative estimates (syndrome decoding over large fields was less investigated, and attacks could improve in the future); another motivation is that over $\FF_2$, tools have been implemented to automatically evaluate the resistance against various flavors of ISD (whose exact cost can be quite tedious to analyze). Because our PCGs can handle fields as low as $\FF_3$, and to avoid having to pick different parameters for each field size, we also based our analysis on known bounds for $\FF_2$.

Second, the main difference between our analysis and that of~\cite{C:BCGIKS20} is that we must consider folding attacks, which are considerably more diverse in our setting (since an attacker can construct a reduced instance by quotienting with \emph{any} subgroup $G'$, of which there are many). Yet, the \emph{effect} of folding on security does not depend on the fine details of the subgroup $G'$, but only on the \emph{size} of $G'$, which allows to compute the new dimension and the reduced noise weight (via a generalized piling-up lemma). This does not differ significantly from the case of ring-\LPN over cyclotomic rings considered in~\cite{C:BCGIKS20}, since there the adversary could reduce the dimension to any power of two of their choice: our setting allows the adversary to be slightly more fine grained in its dimension reduction (\emph{i.e.} the adversary is not restricted to a power of two), but this does not make a significant difference on the concrete attack cost (essentially because close dimensions yield near-identical noise reduction via the piling-up lemma, and do not have significantly different impact on the concrete attack cost beyond that).

As our table illustrates, our \PCGs{} offer a non-trivial stretch (computed as the ratio between the seed size and the size of storing the output \OLEs{}) from a target number $T = 2^{25}$ of \OLEs{}.

\setlength\tabcolsep{6pt}
\begin{table}[H]
\centering
\caption{Concrete parameters and seed sizes (per party, counted in
  bits) for our \PCG for \OLE{} over $\FF_q$ from $\QASD(\Ring)$,
  using $\Ring = \FF_q[(\ZZ/(q-1)\ZZ)^n]$, $\lambda = 128$, target number $T = (q-1)^n$ of \OLEs{}, syndrome
  compression factor $c\in\{2,4\}$, and number of noisy coordinates
  $t$. `Stretch', computed as $2T/{(\text{seed size})}$, is the ratio
  between storing a full random \OLE{} (i.e., $2T$ field elements) and
  the smaller \PCG{} seed. The parameter $k$ denotes the dimension of
  the \SD{} instance after folding, and $t'$ the (expected) noise
  weight of the folded instance (when heuristically choosing the best
  possible folding for the adversary). $\#\PRG$ calls is computed as
  $4\cdot Tct$. Parameters are chosen to achieve $\lambda$-bits of
  security against known attacks\iftoggle{longversion}{, according to
      the analysis of Section~\ref{sec:cryptanalysis}.}{.}}
\label{tab:param}
\begin{tabular}{@{}lllllccc@{}}
\toprule
$T$ &$c$ &$t$ &$(k,t')$ &Seed size &Stretch &$\#$ $R$-mults &$\#\PRG$ calls\\
  \midrule

$2^{25}$ &$2$ &$152$ &$(2^8, 121)$ &$2^{26.0}/\log q$ &$\log q$ &$4$ &$2^{28.2}\cdot\log q$\\
$2^{25}$ &$4$ &$64$ &$(3\cdot 2^8, 60)$ &$2^{23.6}/\log q$ &$5.3\log q$ &$16$ &$2^{28.0}\cdot\log q$\\
  \midrule

$2^{30}$ &$2$ &$152$ &$(2^8, 121)$ &$2^{26.3}/\log q$ &$26\log q$ &$4$ &$2^{33.2}\cdot\log q$\\
$2^{30}$ &$4$ &$64$ &$(3\cdot 2^8, 60)$ &$2^{24.0}/\log q$ &$128\log q$ &$16$ &$2^{33.0}\cdot\log q$\\
  \midrule

$2^{35}$ &$2$ &$152$ &$(2^8, 121)$ &$2^{26.6}/\log q$ &$676\log q$ &$4$ &$2^{38.2}\cdot\log q$\\
$2^{35}$ &$4$ &$64$ &$(3\cdot 2^8, 60)$ &$2^{24.3}/\log q$ &$3327\log q$ &$16$ &$2^{38.0}\cdot\log q$\\
\bottomrule

\end{tabular}
\end{table}
\setlength\tabcolsep{6pt}

\subsubsection{Discussions on Efficient FFTs.} Operations over the
group algebra can be accelerated using the generalized
FFT\iftoggle{longversion}{.}{, which we cover in the long version
    available on eprint.} Here, we briefly remark that some specific
  values of $q$ yield ``FFT-friendly'' instances, where the
  generalized FFT algorithm is extremely efficient (and could even be
  competitive with the more well-known FFT over cyclotomic rings, with
  proper optimizations): this is the case whenever $q-1$ is a power of
  2, since it enables a very efficient divide and conquer algorithm.
  For example, this is the case over $\FF_3[(\ZZ/2\ZZ)^{2^{n}}]$,
  where the FFT reduces to a $2^n$-dimensional FFT over $\FF_3$.

  \subsubsection{From Decision-\QASD to Search-\QASD.}\label{subsec:std_inst}
  \iftoggle{longversion}{In Appendix~\ref{app:std},}{In the long
      version on eprint} we give a reduction from the search version
    of \QASD to the decision version for all instances over
    $\Ring = \FF_q[G]$ where $G = (\ZZ/(q-1)\ZZ)^n$, which is the
    group which we use to obtain \PCGs for \OLEs over
    $\FF_q^{(q-1)^n}$. This provides further support for the security
    of our \PCG schemes, by showing that their security reduces to the
    \emph{search} \QASD assumption. \iftoggle{longversion}{More
        precisely, we prove the following theorem:

        \begin{restatable}{theorem}{std}\label{thm:std_inst}
          Let $q, t$ be two integers, and let
          $G\eqdef (\ZZ/(q-1)\ZZ)^{t}$. Let $n\eqdef |G| = (q-1)^{t}$
          and $w\in\{0, \dots, n\}$ be an admissible weight, and let
          $\psi$ be an error distribution over
          $\mathcal{R}\eqdef\Fq[G]$ such that $\Expect[\HW{\bfx}]=w$
          when $\bfx$ is sampled according to $\psi$. Let
          $\bfs\in\Fq[G]$ be a fixed secret.

          Suppose that there exists a distinguisher $\mathcal{A}$
          between $(\bfa, \bfy^{\text{unif}})$ and
          $(\bfa, \bfa\cdot \bfs + \bfe)$ where
          $\bfa, \bfy^{\text{unif}}\leftarrow \mathcal{R}$ and
          $\bfe \leftarrow \psi$. Denote by $\tau$ its running time
          and $\varepsilon$ its distinguishing advantage. Then, there
          exists an algorithm that recovers $\bfs\in\mathcal{R}$ (with
          an overwhelming probability in $n$) in time
        \[
          O\left( n^{4}\times \frac{1}{\varepsilon^{2}} \times q \times \tau\right).
        \]
      \end{restatable}
      }{In a nutshell, the reduction follows exactly the
        same proof as that of~\cite{BCD22}.}

\iftoggle{longversion}{
  \section{Concrete Cryptanalysis}
  \label{sec:cryptanalysis}

In this section, we discuss the concrete security of \QASD. Most of
the attacks we discuss in this section fit in the framework of linear
tests, and are therefore asymptotically ruled out by our proof of
resistance against linear tests. However, while the concrete bounds of
the proof are reasonable (in the sense that choosing parameters from
these bounds would yield instances that can be reasonably used in
practice), they are overly pessimistic. This stems from the fact that
the linear test framework rules out \emph{all linear attacks} (even
inefficient ones); equivalently, it considers that the adversary can
always find a vector $\vec v$ that minimizes $\HW{\vec v \cdot
  \bfH}$. However, in practice, \emph{finding} the vector $\vec v$
that minimizes $\HW{\vec v \cdot \bfH}$ is a hard problem. Indeed,
this problem, when instantiated with arbitrary codes is known to be
NP--complete \cite{V97} and is commonly assumed to be hard in average
and the best know algorithm to solve this search problem are nothing
but the algorithms solving \SD{}, {\em i.e.}  all the known variants
of ISD.

When choosing concrete parameters, all previous works that rely on
\LPN{} or \SD{} choose instead to use parameters derived
using the \emph{best possible $\vec v$} which can be obtained using
existing linear attacks, such as ISD. For all known concrete linear
attacks, two codes whose duals have the same minimum distance will
yield the same resistance (measured as $\HW{\vec v \cdot \bfH}$)
against these attacks. In other words, these attacks, which are
combinatorial in nature, only rely at their core on the distance
properties of the code and {\em not} on its general structure. To get
an apple-to-apple efficiency comparison with the state of the art, the
natural rule of thumb is therefore to choose parameters similar to
those chosen for variants of syndrome decoding with the same minimum
distance property: this heuristic was explicitely advocated
in~\cite[Section~ 1.4]{C:CouRinRag21}. In our setting, since
quasi-abelian codes meet the GV bound (\emph{i.e.} have typically the
same minimum distance as random linear codes), this translates to
choosing parameters comparable to those of the standard syndrome
decoding problem with random codes.

In our context, this would however be too aggressive, since there are
known ways in which an attacker \emph{can} exploit the structure of
the code. First, because our codes are quasi-abelian codes and hence,
according to Remark~\ref{rem:quasi-ab=multivariate}, they can be
regarded as codes over a quotient of a multivariate polynomial
ring. Therefore, an attacker can reduce the word modulo some ideal
of the ring, in order to generate an instance of a ``smaller'' decoding
problem. This approach has been considered in \cite{CT19} in the
code--based setting and in \cite{BCV20} in the lattice setting.
This point of view has been considered in \cite{C:BCGIKS20} when studying
the security of \OLEs{} generated using instances of Ring-\LPN{}.

The parameters should therefore be chosen such that
any such ``reduced instance'' remains intractable. Second, due to the
quasi-abelian structure of our codes, one can apply the \DOOM{} attack
from~\cite{S11} to obtain a speedup by a factor $\sqrt{|G|}$, where $G$
denotes the underlying abelian group of the group algebra.

\subsubsection{Our setting.} In the following, we focus on linear
attacks against the $\QASD(n,k)$ assumption instantiated over a ring
$\Ring = \FF_q[X_1,\dots, X_d] / (X_1^{q-1}-1,\dots,X_d^{q-1}-1)$.
Our point is to distinguish pairs
$((a_1, \dots, a_c), a_1 s_1 + \cdots + a_c s_c + e)$ (with possibly
$c = 1$), where $a \getsr \Ring$ and $s_1, \dots, s_c, e \in \Ring$ are
sparse with respect to the basis of monomials.  As already mentioned
in Section~\ref{sec:qasd}, the search version of the problem is
equivalent to solving the \QASD{} problem.  That is to say solving a
decoding problem of the form
\[
  (~\bfA_1 ~|~  \cdots ~|~ \bfA_c ~|~ 1 ~)
  \begin{pmatrix}
    \bfs_1 \\ \vdots \\ \bfs_c \\ \bfe
  \end{pmatrix}
  = 0,
\]
where the $\bfA_i$'s are the matrix representations in the basis of
monomials of the multiplication--by--$a_i$ maps in $\Ring$ and the
$\bfs_i$'s and $\bfe$ are the unknown vector representations of
the $s_i$'s and $e$ in this basis, {\em i.e.} are unknown sparse vectors.

In terms of code parameters, the group codes
have length $n = (c+1)\dim_{\Fq} \Ring$ and dimension
$k = c \dim_{\Fq} \Ring$. Therefore, we always have
$k \geq \frac{n}{2}$ with equality when $c = 1$.  In this setting,
attacks such as Arora-Ge \cite{arora2011new} (which require
$n = \Omega(k^2)$) or BKW (which require $n$ to be subexponential in
$k$, or $n = \Omega(k^{1+\eps})$ using the sample-efficient variant
of~\cite{AR:L05}) do not apply. Furthermore, our codes have rate $c/(c+1)$ with $c \geq 1$. In particular, this implies that the recent results on Statistical Decoding 2.0~\cite{carrier2023statistical}, which improves over ISD when the code rate is sufficiently small, will not yield an efficient attack on our setting (for rates above $1/2$, SD 2.0 is always outperformed by ISD).

\subsection{Instance Projection via Quotient}\label{subsec:proj_quotient}
As previously mentioned, a manner to solve the problem is to solve the
search $\QASD{}(\Ring)$ problem, where
$\Ring = \Fq[X_1, \dots, X_d]/(X_1^{q-1} - 1, \dots, X_d^{q-1} - 1)$.
Given an instance $(a,b)$ of $\QASD(\Ring)$, an attacker may construct
a new decoding instance with smaller length and dimension. In full
generality, the attacker can pick any ideal
$I \subseteq \Fq[X_1, \dots , X_d]$ containing
$(X_1^{q-1}-1, \dots, X_d^{q-1} - d)$ and represented by a Gröbner
basis, and construct a new instance $(a',b') \gets (a,b) \bmod I$,
where the $\bmod$ operation is the reduction modulo $I$ with respect
to the chosen Gröbner basis.
For instance, one can choose a sequence $(F_1(X_1), \dots, F_d(X_d))$
of factors of $X_1^{q-1}-1, \dots, X_d^{q-1}-1$ and reduce modulo them.

However, in general, the projection modulo an arbitrary ideal $I$ can
significantly increase the noise. The way the noise increases is
highly dependent from the density of the generators of $I$. For
example, if $\Ring = \FF_q[X_1,X_2]/(X_1^{q-1}-1, X_2^{q-1}-1)$ and
the attacker reduces modulo $I = (F_1(X_1), F_2(X_2))$ where $F_1, F_2$ are
respective factors of $X_1^{q-1}-1$ and $X_2^{q-1}-1$ of respective
Hamming weight, say, 3 and 5, the noise rate can increase by a factor
up to $(3-1)\cdot(5-1) = 8$. Therefore, we expect this approach to
be useful (to the attacker) only when the noise increase is very
small.

Heuristically the best possible projections of $\Ring$ regarded as the
group algebra $\Fq[G]$ seem to be the projections arising from
quotients of $G$. Namely, given a subgroup $H$ of $G$ the canonical
map $G \rightarrow G/H$ induces a morphism of algebras
\[
 \pi_H: \map{\Fq[G]}{\Fq[G/H]}{\sum_{g \in G} a_g \ g}{
  \sum_{\bar g \in G/H} (\sum_{h \in H} a_{gh}) \bar{g}.}
\]
From a coding theoretic point of view, this operation is nothing but
summing up the entries of a codeword whose index are in a same orbit
under the action of $H$. This operation sends a code of length
$(c+1)|G|$ and dimension $c|G|$ onto a code of length $(c+1)|G/H|$ and
dimension $c|G/H|$.  Moreover, a noisy codeword $c+e$ is sent onto
$\pi_H(c) + \pi_H(e)$ and the weight of $\pi_H(e)$ is bounded from above
by the weight of $e$. In summary, the length and dimensions of the code
are divided by $|H|$ while the weight of the error is preserved or slightly
reduced since some entries of $e$  may sum up to $0$.

Such projections seem optimal in terms of limiting the growth of
the noise.

\begin{remark}
  From the ring theoretic point of view, the map $\pi_H$ can be
  regarded as a quotient map of $\Ring = \Fq [G]$ modulo the ideal
  generated by all the elements $(h-e_G)$ where $h \in H$ and $e_G$
  denotes the unit element of the group $G$.
\end{remark}

\begin{example}
  Following the spirit of \cite{C:BCGIKS20} consider the case
  \[
    \Ring = \Fq[X]/(X^{q-1}-1) \simeq \Fq [\ZZ/(q-1)\ZZ].
  \]
  In this situation, for any $\ell | (q-1)$, one can consider the subgroup
  $H = \ell \ZZ / (q-1) \ZZ$. The corresponding projection can be made
  explicit as
  \[
    \pi_H : \map{\Fq[X]/(X^{q-1}-1)}{\Fq[X]/(X^{\frac{q-1}{\ell}} -
      1)}{\sum_{i=0}^{q-2} a_i X^i}{\sum_{i = 0}^{\frac{q-1}{\ell} - 1} \left(\sum_{j \equiv i \bmod \ell} a_j \right) X^i.}
  \]
  In short, we sum up the entries of the codeword whose indexes are
  congruent modulo $\ell$.
\end{example}

\begin{example}\label{ex:multi_LWE}
  This example is in the spirit of the attacks on multivariate \RLWE{}
  \cite{BCV20}. Consider the ring $\Ring = \Fq[\ZZ / n\ZZ \times \ZZ / n\ZZ]
  \simeq \Fq[X,Y]/(X^n-1, Y^n-1)$ and consider the subgroup
  \[
    H \eqdef \{(x,x) ~|~ x\in \ZZ/n\ZZ\} \subseteq G = \ZZ / n\ZZ
    \times \ZZ / n\ZZ.
  \]
  Here the projection map can be made explicit as
  \begin{equation}\label{eq:BCV}
    \pi_H : \map{\Fq[X,Y]/(X^n-1, Y^n-1)}{\Fq[X]/(X^n-1)}{\sum_{i,j =
        0}^{n-1} a_{ij}X^i Y^j}{\sum_{i = 0}^{n-1} \left(\sum_{u+v
          \equiv i \bmod n} a_{uv} \right) X^i.}
  \end{equation}
\end{example}

This approach is considered in \cite{BCV20} to provide an attack on
multivariate \RLWE{}. In the coding theoretic context, this approach
is analysed in depth in \cite{CT19} where the projection map is called
{\em folding}.

\subsubsection{Computing the new noise weight.}
Following~\cite{C:BCGIKS20}, we consider an instance $(a,ae+f)$ where
each sparse vector $e,f$ has been sampled as a sum of $t/2$ random
monomials. This distribution is very close to the original
distribution, and its choice significantly simplifies the analysis. It
also slightly favor the attacker (since the expected number of noisy
entries will now be slightly below $t$ due to possible collisions). In
this setting, the expected noise rate $t'$ can be computed fairly
simply. Let $R_{m,\ell}$ be the random variable counting the number of
nonzero coefficients in a polynomial with $m$ coefficients over
$\FF_q$ computed as the sum of $\ell$ random monomials. Note that
$t' = \Expect[R_{n, t}]$, where $n$ is the code length. Then, we have
\[\Expect[R_{m,\ell+1}] = \left(1-\frac{\Expect[R_{m,\ell}]}{m}\right)\cdot(\Expect[R_{m,\ell}] + 1) + \frac{\Expect[R_{m,\ell}]}{m}\cdot\left(\Expect[R_{m,\ell}] - \frac{1}{q-1}\right),\]
since adding a new random monomial increases the number of nonzero
coefficients by 1 if it falls in a position with a zero coefficient,
and decreases the expected number of nonzero coefficients by $1/(q-1)$
otherwise (since this is the probability, when summing two random
elements of $\FF_q^*$, to get $0$). Solving the recurrence relation
gives
\[t' = \frac{n\cdot(q-1)}{q} \cdot \left( 1- \left(1-\frac{q}{n\cdot
        (q-1)}\right)^\ell \right).\]

In the rest of the analysis, we will cover standard attacks on
syndrome decoding on instances of a given noise rate and dimension.
Then, when choosing concrete parameters, we will estimate the attacker
cost as the minimum cost of solving any instance obtained by reducing
$\Fq[G]$ to $\Fq[G/H]$, estimating the reduced noise parameter $t'$
using the formula above. We note that this approach ignores the
possibility that for a given instance, $t'$ ends up being much smaller
than its expected value, which would yield some weak instances of the
problem. As in~\cite{C:BCGIKS20}, we observe that this can be avoided
by changing the structure of the noise using rejection sampling: one
can re-sample the noise vectors until the weight $t'$ of the reduced
instance over $\Fq[G/H]$ (using the best possible choice of $|H|$ for
the attacker with the attacks covered below) is at least its expected
value (on average, since the probability of having
$\Expect[t'] \leq t'$ is $1/2$, this reduces by at most a single bit
the entropy of the noise vector).

\subsection{Information Set Decoding}

In this section, we cover standard linear attacks against syndrome
decoding. The most advanced attacks in this category are the
information set decoding (ISD) attacks, initially introduced by
Prange~\cite{P62} and subsequently refined in a long sequence of
works~\cite{stern1988method,AC:FinSen09,C:BerLanPet11,AC:MayMeuTho11,EC:BJMM12,EC:MayOze15}. Evaluating
precisely the effect of each attack on a given instance is complex and
tedious, but a general lower bound on the attack cost was derived
in~\cite{C:HOSS18}, based on similar analysis given
in~\cite{AC:FinSen09,S11,hamdaoui2013non,torres2016analysis}. These
lower bounds build upon the common structure of most ISD variants. In
general, the cost of modern ISD algorithms for a code with
parity-check matrix $\bfH$ over $\FF_2$, with dimension $k$, code
length $n$, and $t$ noisy coordinates, is lower bounded by
\[\min_{p_1,p_2}\left\{\frac{\min\left\{2^{k}, {\binom{n}{t}}\right\}}{{\binom{k-p_2}{t-p_1}}}\cdot \left(\frac{K_1+K_2}{{\binom{k+p_2}{p_{1}}}} + \frac{t\cdot(k-p_2)}{2^{p_2}}\right)\right\},\]
where $(p_1,p_2)$ satisfy $0 \leq p_2 \leq k/2$ and
$0 \leq p_1 \leq k+p_2$, $K_1$ denotes the cost of Gaussian
elimination on a submatrix of $\bfH$ with $n-p_2$ columns, and $K_2$
denotes the running time of a specific sub-algorithm, which varies
accross different attacks. As in~\cite{C:BCGIKS20}, we assume that
performing Gaussian elimination on the submatrix of $\bfH$ can be done
in time $K_1 \approx (k-p_2)^2\log(k-p_2)$,
because $\bfH$ is a structured matrix. According to the analysis
of~\cite{C:HOSS18}, $K_2$ can be lower bounded by
$K_2 \geq {\binom{(k+p_2)/2}{p_1/8}}$ for algorithm
of~\cite{EC:BJMM12}. As in~\cite{C:BCGIKS20},~\cite{EC:BJMM12} seems
to provide the best efficiency in our setting (more recent algorithms
have large hidden constants that render them less practical, or
improve over~\cite{EC:BJMM12} only for very high noise rates).

The above analysis is restricted to the case of $\FF_2$, which is the
easiest to attack using ISD. Over larger fields, one can use the above
costs as a lower bound for the true cost of the attack, but as the
field size grows, this lower bound becomes quite loose. Indeed, this
bound was used to pick concrete parameters in~\cite{C:BCGIKS20}, but a
recent preprint~\cite{EPRINT:LWYY22} estimates that the parameters
recommended in~\cite{C:BCGIKS20} for 80 bits of security actually
achieve 92-112 bits of security, while the parameters recommended for
128 bits of security actually achieve 133-171 bits of security. In our
setting, however, our \PCGs{} can be instantiated over fields as small as
$\FF_3$, in which the costs should be much closer to the lower bounds
used in~\cite{C:BCGIKS20}.

We note that a detailed analysis of ISD over larger fields was given in a recent paper~\cite{SAC:BCDL19}. However, for the sake of avoiding to compute different \QASD parameters for each possible field size $\FF_q$, we stick in this paper to the conservative lower bound that stems from the analysis over $\FF_2$.

In \cite{CT19} a study of the combination of ISD with the {\em
  folding} operation is studied and precises how the use of folding
improves the complexity of the decoder. It turns out that for small
errors rates, which is precisely our setting, the use of folding does not
represent a significant improvement.

\subsection{Prange and statistical decoding (Low-Weight Parity-Check)}
We also consider other standard linear attacks, such as Prange
decoding algorithm~\cite{P62} and low-weight parity
checks~\cite{Zichron17,al2001statistical,fossorier2006modeling,ACISP:Overbeck06,debris2017statistical}
which leads to the so-called {\em statistical decoding}.  The former,
which is just the original ISD algorithm, consists in guessing $k$
noise-free equations and solving the resulting system. It has the
advantage over more recent ISD algorithms that it does not depend on
the field size. The latter is also often more efficient than ISD in
our setting. This is because ISD is a search attack, and executing the
attack involves solving a linear system in each iteration of the
attack. Since typical \PCG{} applications have huge dimensions
(e.g. $k \approx 2^{30}$), this polynomial cost turns out to have a
significant impact on the overall runtime of the attack (even though
ISDs have the lowest exponent in the exponential part of the
attack). Low-weight parity checks, however, work by directly finding
many $\vec v$ such that $\vec v \cdot \bfH$ has low weight, and
declare $\vec b$ to be a syndrome decoding instance if the set
$\{\vec v \cdot \vec b^\Trans\}$ contains too many zeroes. In other
words, these attacks directly target the decision variant of syndrome
decoding (on which our \PCGs{} rely) and require computing only an
inner product per iteration, rather than solving a large linear
system. Concretely, the cost of Prange (when $\bfH$ is a structured
matrix) is given by
$O\left(1/(1-\frac{t}{n})^{k}\cdot k^2\log k \right)$ arithmetic
operations, and the cost of the low-weight parity check attack is
$O\left(n/(k-1)^t\cdot k\right)$ arithmetic operations
(see~\cite{CCS:BCGI18,C:BCGIKS20}).

\subsection{Algebraic Decoding Attacks}
An important line of work in code--based cryptography consists in
recovering a hidden algebraic structure of a code which permits to
decode. See for instance \cite{W10,CGGOT14,COT17,CMP17,CLT19}.  In
general such attacks rest on the fact that the public code $\code C$
or some of its subcodes has a peculiar behaviour with respect to the
component wise product.  Namely that the ``square of $\code{C}$'',
{\em i.e.} the span of the component wise
products of any two words of $\code C$ has small dimension compared
to the square of a random code.

Note that codes sharing this feature of having a ``small square''
benefit from an efficient decoding algorithm usually referred to as
{\em Error Locating Pairs decoder} \cite{P92}. See
\cite[Section~4]{Couvreur21} for further details. Therefore, if a
random quasi--group code had a small square compared to random codes,
then one could deduce an algebraic decoder for quasi--group codes
which is a longstanding open question: even when restricting to the
case of cyclic codes!

Algebraic attacks exploit the structure of the underlying code to decode it efficiently. Many such algebraic decoding attacks have been devised in the literature, and fall in a unified framework developed in~\cite{P92,kotter1992unified} based on componentwise product of codes. Examples of such attacks include~\cite{pellikaan2011evaluation,marquez2012error,faugere2013distinguisher,couvreur2013distinguisher,marquez2014unique} (and many more), and were often used to break some variants of the McEliece cryptosystem. In our context, though, algebraic decoding of quasi-group codes is a well-known and long-standing open problem: it has been studied for over 50 years in the coding
theory community, and to this day no
efficient algorithm is known to decode a random quasi-abelian code.
This is listed as an open research problem in the most recent
Encyclopedia of Coding Theory (from 2021) \cite[Problem~16.10.5]{willems2021}.

\subsection{Attacks on Multivariate LWE}
As already mentioned in Example~\ref{ex:multi_LWE}, an attack on
multivariate \RLWE{} is presented in \cite{BCV20}. This attack is
based on a projection of the form $\Fq[G] \rightarrow \Fq[G/H]$ as
described in Section~\ref{subsec:proj_quotient}.  The attack is
particularly efficient since applying a map of the form (\ref{eq:BCV})
has a very limited impact on the Euclidean norm and hence has a
limited impact on the noise term. In the coding theoretic setting, the
situation is very different since the Hamming weight of the error is
more or less preserved but then the relative weight, {\em i.e.} the ratio
$\frac{t}{n}$ is more or less multiplied by a term $|H|$.
Therefore, reducing with respect to a too large subgroup $H$
leads to shorter codes but provides intractable instances of the decoding
problem.

\subsection{Decoding One-Out-Of Many}
For a code equipped with a non trivial permutation group, which is an
obvious feature of quasi-abelian codes, the decoding problem can be
made easier using Sendrier's {\em Decoding One Out of Many} (\DOOM{})
paradigm \cite{S11}. Indeed, consider a quasi--abelian code
$\code{C} \subseteq \Fq[G]^\ell$ and a noisy codeword $y = c+e$ with
$c \in \code{C}$ and $e \in \Fq [G]^\ell$ of low weight. Then,
for any $g \in G$, we get another instance of the decoding problem
with an error term of the same weight:
\[
  g \cdot y = g \cdot c + g \cdot e.
\]
Here $g \cdot c \in \code{C}$ and $\wt{g \cdot e} = \wt{e}$.
Therefore, given a single instance of \QASD{} we naturally deduce
$|G|$ instances and solving one of them immediately solves the other
ones. Thus, from \cite{S11}, solving one out of $|G|$ instances of
\SD{} permits to divide the work factor of any decoder by $\sqrt{|G|}$.
Therefore the cost of ISD should be divided by $\sqrt{|G|}$ and the cost
of the composition of a projection $\Fq[G] \rightarrow \Fq[G/H]$
with ISD should be divided by $\sqrt{|G/H|}$.

   \section{Applications to Secure Computation}
  \label{sec:applications}

In this part, we explain some of the main applications of our new \PCG's to secure computation. To provide bounds, we will use the following restatement of Theorem \ref{pcg:theoremEstimateCostOLE-all} in the case $\Ring = \Fq[X_1,.., X_n] / (X_1^{q-1}-1,..,X_n^{q-1}-1 )$.

\begin{theorem}
\label{pcg:theoremEstimateCostOLE-allRestate}

Suppose that \SPFSS is a secure \FSS{} scheme for sums of point functions and that the \QASD assumption holds. Let $\Ring = \Fq[X_1,.., X_n] / (X_1^{q-1}-1,..,X_n^{q-1}-1 )$, and $T = (q-1)^n$. We can construct a \PCG producing $T$ instances for \OLE over $\Fp$, using the $\QASD_{\OLE}$  construction with the following parameters

   \begin{itemize}
    \item[$\bullet$] Communication costs and size of the seed : $O(\lambda^3 \log T)$.
    \item[$\bullet$] Computation costs : $O(\lambda T)$ \PRG{} evaluations and $O(c^2 T\log T)$ operations in $\Fq$.
\end{itemize}
\end{theorem}

 First, as explained in \cite{C:BCGIKS20}, a \PCG generating $N$ multiplication triples can be derived from a \PCG generating $2N$ \OLE{}.

\subsubsection{Extension to multiplication triples.} The \OLE{} correlation gives a secret to each party $P_0$ and $P_1$ and an additive secret-sharing of the product of the two secrets. The \OLE{} correlation is interesting in its own right and can be used directly in some applications, but in general, the multiplication triple correlation is used. A (2-party) multiplication triple, gives the parties additive shares of random elements $a$ and $b$, and shares of the product $a \cdot b$. The main advantage of multiplication triples is their usefulness in the setting of 2-party computation of arithmetic circuits over $\FF_q$.

In this setting, each multiplication gate can be evaluated by consuming a single multiplication triple, and with communication costs of two $\FF_q$ elements per party - the additions are free in this setting.
Using two instances of an \OLE{} correlation we can obtain an instance of a multiplication triple correlation. Let $a = a_0 + a_1, b = b_0 + b_1$ and $c = ab = a_0 b_0 + a_0 b_1 + a_1 b_0 + a_1 b_1$, we can distribute $a_{\sigma},b_{\sigma}$ to party $P_{\sigma}$ and run two independent \OLE{} instances to obtain the secret share of the cross terms $a_0 b_1$ and $a_1 b_0$. As the party $P_{\sigma}$ can locally compute $a_{\sigma} b_{\sigma}$ it gets a correct sharing of $a b$. Note that we obtain the correlation in a black-box way. Thus, a \PCG generating $N$ multiplication triples can be derived from a \PCG generating $2N$ \OLE{}.

\subsection{Application : (N-party) multiplication triples
  generation for arithmetic circuit}

\begin{theorem}
  Assume the existence of oblivious transfers and $\QASD(\Ring)$
  assumption, where
  $\Ring = \FF_q[X_1,\cdots, X_n] / (X_1^{q-1}-1,\cdots,X_n^{q-1}-1
  )\simeq \FF_q \times \cdots \times \FF_q, $ with $ q \geq 3$. Let
  $T = (q-1)^n$.
  There exists a
  semi-honest $N$-party protocol for securely evaluating an arithmetic
  circuit $C$ over $\FF_q$ with $T$ multiplication gates, in the
  preprocessing model, such that:
    \begin{itemize}
        \item The preprocessing phase has communication cost $\Tilde{c}(N,\lambda, T) = O(\lambda^3 \cdot N^2 \cdot \log (2T))$, and computation cost $\dot{c}(N,\lambda,T) = O(N^2 \cdot \lambda \cdot 2T)$ \PRG calls; $O( N^2 \cdot 2T \log (2T))$ operations in $\Fq$.
        \item The online phase is non-cryptographic and communication cost $2 \cdot N \cdot T$ elements of $\FF_q$.

     \end{itemize}

\end{theorem}

\begin{proof}

  Consider the parties $P_1,\cdots, P_N $. First, remark that
  programmability enables parties to generate ``correlated'' (2-party)
  multiplication triples, which can used to obtain $N$ -party
  multiplication triples in the following way.

\begin{itemize}
    \item The party $P_i$ gets two random values $(x_i,y_i) $. We define
        $X = \sum_i x_i$ and $Y = \sum_j y_j$.
  \item Each pair of parties $(P_i,P_j)_{1 \leq i,j \leq N, i \neq j}$
        performs the programmable protocol for (2-party) multiplication
        triples with programmable inputs $(x_i,y_j)$, and obtains
        shares of $x_i \cdot y_j$. We indicate the share of $P_i$ as
        $\langle x_i \cdot y_j \rangle_{i}$.
    \item Let $K_i = \sum_{j=1}^N \langle x_i \cdot y_j \rangle_i + \langle x_i  \cdot y_j \rangle_i + x_i \cdot y_i $.
The $K_i$ are shares of the product
\begin{equation*}
    X \cdot Y = \sum_{1 \leq i,j \leq N} x_i \cdot y_j = \sum_{i = 1}^N K_{i}
\end{equation*}
\end{itemize}

The parties use the $\QASD_{\OLE}$ to generate short seeds of each of the $N\cdot(N-1)$ ($2$-party) multiplication triples they need. In the online phase, they locally expand the seeds to obtain $T$ instances of ($N$-party) multiplication triples. The parties can execute the ($N$-party) GMW protocol using the multiplication triples, and evaluate the circuit.

Using Theorem \ref{pcg:theoremEstimateCostOLE-allRestate} we obtain the cost of preprocessing for generating the $2T$ \OLE{} over $\Fp$, namely $O(\lambda^3 \cdot \log (2T))$ in communication cost, and $\dot{c}(N,\lambda,T) = O(N^2 \lambda 2T)$ \PRG calls ; $O(\lambda^2 \cdot 2T \log (2T))$ operations in $\Fq$ in computation cost.

The cost of communication in the online phase is simply derived from the GMW algorithm using the multiplication triples. For each multiplication gate, each party must send two field elements, resulting in a cost of $2 \cdot N \cdot T$.

\end{proof}

\subsection{Secure Computation with Circuit-Dependent Preprocessing}

Circuit-dependent preprocessing is a variation of the standard
Beaver's circuit randomization technique with multiplication triples.
It has been investigated in recent works, such
as~\cite{C:DNNR17,EC:Couteau19}. The idea is to preprocess
multiplications in a way that depends on the structure of the circuit
and leads to an online phase that requires just \textit{one opening
  per multiplication gate}, instead of two when using multiplication
triples. \PCGs for \OLEs do not directly enable reducing the
preprocessing phase of secure computation with circuit-dependent
correlated randomness: at a high level, this stems from the fact that
since the correlated randomness depends on the topology of the
circuit, it cannot be compressed beyond the description size of this
topology. Nevertheless, \PCGs enable \emph{batch} secure computation
(\emph{i.e.} securely computing many copies of the same circuit on
different input) with silent preprocessing in the circuit-dependent
correlated randomness setting, by using \PCGs to compress a batch of
correlations for a given gate accross all circuits.

\begin{theorem}
     Assume the existence of oblivious transfer and the $\QASD(\Ring)$ assumption, where  $\Ring = \Fq[X_1,\cdots, X_n] / (X_1^{q-1}-1,\cdots,X_n^{q-1}-1 )\simeq \Fq \times \cdots \times \FF_q, $ with $  q \geq 3$. Let $T = (q-1)^n$.
    There exists a semi-honest 2-party protocol for securely evaluating $T$ copies of an arithmetic circuit $C$ over $\FF$ with $S$ multiplication gates, in the preprocessing model, such that:
    \begin{itemize}
        \item The preprocessing phase has communication cost $c(T,\lambda, S) = O(\lambda^3 \cdot S \cdot \log (2T))$ and a computation cost  $\dot{c}(T,\lambda,S) = O(\lambda \cdot S \cdot 2T)  $ \PRG calls ; $ O( S \cdot 2T \log (2T))$ operations in $\Fq$.

        \item The onTine phase is non-cryptographic and communication costs $2 \cdot S \cdot T$ elements of $\FF$.

     \end{itemize}

\end{theorem}

\begin{proof}
 Let $C$ be an arithmetic circuit over $\FF$ consisting of fan-in two addition and multiplication gates. Each wire $w$ is assigned a mask $r_w$ during the offline phase. The masks are designed as follows.

\begin{itemize}
    \item[$\bullet$] if $w$ is an input wire, $r_w$ is chosen at random.
    \item[$\bullet$] if $w$ is the output wire of a multiplication gate, $r_ w \leftarrow \FF $ is chosen at random
    \item[$\bullet$] if $w$ is the output wire of an addition gate with input wires $u$ and $v$, then $r_w = r_u + r_v$.
    \item[$\bullet$] for each multiplication gate, we assigned a value $s_{u,v}$, such that on input wires $u$ and $v$, $s_{u,v} = r_u \cdot r_v$.
\end{itemize}

The masks are not known by the parties, but they obtain random additive shares of each $r_w$ for any input and output wire of multiplication gates, as well as $s_{u,v}$ for the multiplication gates.

When the online phase begins, both parties hide their secret values with random masks. The party that is not the one giving the input for a given wire $w$ gives to the other one his shares $\langle r_w\rangle$.
The invariant of the online phase is that through the protocol, for each wire, parties know exactly the value $x + r_x$ where $r_x$ is the mask of this wire (for which parties have additive sharing), and $x$ is the real value that is computed by the circuit passing through the wire $w$.
The invariant is preserved through each gate because of the following:
\begin{itemize}
    \item[$\bullet$]  For an addition gate, parties know  $x+r_x$ and $y + r_y$. Then the parties add locally those values to obtain $x + y + r_x + r_y$ , with $r_x + r_y$ being indeed the output mask for the addition gate.
    \item[$\bullet$] For a multiplication gate with $r_{w}$ denoting its output wire's mask, parties know  $x+r_x$ and $y + r_y$. The parties can locally compute their share $ \langle (x+r_x)\cdot r_y + (y+r_y) \cdot r_x + r_x \cdot r_y + r_w \rangle$ (the formula can be a little bit different if we are not in $\FF_2$). By both exchanging one bit of information, they reconstitute that value. Adding up $(x+r_x) \cdot (y + r_y)$, they obtain in clear $x \cdot y + r_w$ where $r_w$ is the mask of the output wire of this multiplication gate.
\end{itemize}

In the end, we have to perform $2S$ different calls to our \PCG to create the ($2$-party) multiplication triples seeds.
Again we use Theorem \ref{pcg:theoremEstimateCostOLE-allRestate} to get the estimation of the costs in communication and space, per instances, and we multiply it by $S$. In the online phase, we gain a factor 2 in communication because each party only has to send a bit of information for each of the multiplication gates. As there are $S \cdot T$ multiplication gates in total, the communication cost in the online phase is $2 \cdot S \cdot T$. \qed
\end{proof}

 }{}

\bibliographystyle{alpha}
\newcommand{\etalchar}[1]{$^{#1}$}

\appendix

\chapter*{Appendices}
\section{Additional Preliminaries}
\label{app:prelims}

\subsection{Function Secret Sharing}

Function secret sharing (\FSS{}), introduced
in~\cite{EC:BoyGilIsh15,CCS:BoyGilIsh16}, allows to succinctly share
functions. In this section, we largely follow the presentation from
the preliminaries of~\cite{C:BCGIKS20} (in particular, the definitions
are reproduced almost verbatim from~\cite{C:BCGIKS20}). An \FSS{}
scheme splits a secret function $f:I\to\GG$, where $\GG$ is some
Abelian group, into two functions $f_0,f_1$, each represented by a key
$K_0,K_1$, such that: (1) $f_0(x) + f_1(x)=f(x)$ for every input
$x\in I$, and (2) each of $K_0,K_1$ individually hides $f$.

\begin{definition}[Function Secret Sharing] Let $\class = \{f:I\to\GG\}$ be a class of function descriptions, where the description of each $f$ specifies the input domain $I$ and  an Abelian group $(\GG,+)$ as the output domain.
A (2-party) \emph{function secret sharing} (\FSS{}) scheme for $\class$ is a pair of algorithms $\FSS = (\FSS.\Gen,\FSS.\Eval)$ with the following syntax:
\begin{itemize}
  \item $\FSS.\Gen(1^\secpar,f)$ is a Probabilistic Polynomial Time (PPT) algorithm that given security parameter $\secpar$ and description of $f\in\class$
outputs a pair of keys $(K_0,K_1)$. We assume that the keys specify $I$ and $\GG$.
  \item $\FSS.\Eval(b,K_b,x)$ is a polynomial-time algorithm that, given a key $K_b$ for party $b\in\bit$, and an input $x\in I$, outputs a group element $y_b\in\GG$.
\end{itemize}
The scheme should satisfy the following requirements:
\begin{itemize}
  \item {\bf Correctness:} For any $f\in\class$ and $x\in I$, we have
  \[
        \Pr\left[(K_0,K_1)\getsr\FSS.\Gen(1^\secpar,f) ~\bigg|~\sum_{b\in\bit}\FSS.\Eval(b,K_b,x) = f(x)\right] = 1.
  \]
  \item {\bf Security:} For any $b\in\bit$, there exists a PPT
        simulator $\Sim$ such that for any polynomial-size function
        sequence $f_\secpar\in\class$, the distributions
        $\{(K_0,K_1)\getsr\FSS.\Gen(1^\secpar,f_\secpar) ~|~ K_b\}$
        and $\{K_b\getsr\Sim(1^\secpar,\Leak(f_\secpar))\}$ are
        computationally indistinguishable.
\end{itemize}
In the constructions we use, the leakage function $\Leak:\bit^*\to\bit^*$ is given by $\Leak(f_\secpar) = (I,\GG)$, namely it outputs a description of the input and output domains of $f$.
\end{definition}

We also define a full-domain evaluation algorithm, $\FSS.\FullEval(b,K_b)$, which outputs a vector of $\abs{I}$ group elements, corresponding to running $\Eval$ on every element $x$ in the domain $I$. For the type of \FSS{} we consider, $\FSS.\FullEval$ is significantly faster than the generic solution of running $|I|$ instances of $\Eval$. We will use \FSS{} for point functions and sums of point functions, as defined below.

\begin{definition}[Distributed Point Function (DPF)~\cite{EC:GilIsh14,EC:BoyGilIsh15}]
  \label{def-dpf}
  Denote by $[n]$ the set of integers $\{0, \dots, n-1\}$.
For an Abelian group $\GG$, $\alpha \in [n]$, and $\beta \in \GG$, the {\em point function} $f_{\alpha,\beta}$ is the function $f_{\alpha,\beta}:[n]\to\GG$ defined by $f_{\alpha,\beta}(x)=0$ whenever $x \ne \alpha$, and $f_{\alpha,\beta}(x) = \beta$ if $x = \alpha$.
A \emph{distributed point function} (DPF) is an \FSS{} scheme for the class of point functions $\{f_{\alpha,\beta} : [n] \rightarrow \GG \mid \alpha\in [n],\beta\in\GG\}$.
\end{definition}

The best known DPF construction~\cite{CCS:BoyGilIsh16} can use any pseudorandom generator (\PRG{}) $G : \bit^\secpar \to \bit^{2\secpar+2}$ and has the following efficiency features.
For $m=\lceil \frac{\log \abs{\GG}}{\secpar+2} \rceil$, the key generation algorithm $\Gen$ invokes $G$ at most $2(\lceil\log n\rceil+m)$ times, the evaluation algorithm $\Eval$ invokes $G$ at most $\lceil\log n\rceil+m$ times, and the full-domain evaluation algorithm $\FullEval$ invokes $G$ at most $n \cdot (1+m)$ times.
The size of each key is at most $\lceil\log n\rceil\cdot(\secpar+2) + \secpar + \lceil\log_2|\GG|\rceil$ bits. We will use a simple and generic extension of DPF to sums of point functions.

\begin{definition}[\FSS{} for sum of point functions (\SPFSS{})]
\label{def-spfss}
For $S = (s_1, \dots, s_t) \in [n]^t$ and $\vec y = (y_1, \dots, y_t) \in \GG^t$, define the \emph{sum of point functions} $f_{S,\vec y} : [n] \rightarrow \GG$ by
\[ f_{S,\vec y}(x) = \sum_{i=1}^t f_{s_i,y_i}(x). \]
An \SPFSS{} scheme is an \FSS{} scheme for the class of sums of point functions.
\end{definition}

Note that for $S = (s_1, \dots, s_t)$, the function $f_{S,\vec y}$ is non-zero on {\em at most} $t$ points.
If the elements of $S$ are distinct, $f_{S,\vec y}$ coincides with a {\em multi-point function} for the set of points in $S$. A simple realization of \SPFSS{} is by summing $t$ independent instances of DPF. This will typically be good enough for our purposes. To simplify notation, when generating keys for a scheme $\SPFSS = (\SPFSS.\Gen,\allowbreak \SPFSS.\Eval)$, we write $\SPFSS.\Gen(1^\secpar, S, \vec y)$, instead of explicitly writing $f_{S,\vec y}$.

\subsection{Pseudorandom Correlation Generators}

We recall the notion of pseudorandom correlation generator (\PCG{})
from~\cite{C:BCGIKS19}. At a high level, a \PCG{} for some target ideal
correlation takes as input a pair of short, correlated seeds and
outputs long correlated pseudorandom strings, where the expansion
procedure is deterministic and can be applied locally. The definitions below are taken almost verbatim from~\cite{C:BCGIKS20}.

\begin{definition}[Correlation generator] A PPT algorithm $\cC$ is called a \emph{correlation generator}, if $\cC$ on input $1^\secpar$ outputs a pair of elements in $\{0,1\}^n\times\{0,1\}^n$ for $n\in\poly(\lambda)$.
\end{definition}

The security definition of \PCGs requires the target correlation to satisfy a technical requirement, which roughly says that it is possible to efficiently sample from the conditional distribution of $\R_0$ given $\R_1=r_1$ and vice versa. It is easy to see that this is true for the correlations considered in this paper.

\begin{definition}[Reverse-sampleable correlation generator]
\label{def:rsc} Let
$\cC$ be a correlation generator. We say $\cC$ is \emph{reverse sampleable} if there exists a PPT algorithm $\RSample$ such that for $\sigma\in\{0,1\}$ the correlation obtained via:
\begin{align*}\{(\R_0^\prime,\R_1^\prime)\mid  &(\R_0,\R_1)\getsr \cC(1^\secpar), \R^\prime_\sigma:=\R_\sigma,\R^\prime_{1-\sigma}\getsr \RSample(\sigma,\R_\sigma)\}\end{align*}
is computationally indistinguishable from $\cC(1^\secpar)$.
\end{definition}

\begin{definition}[Pseudorandom Correlation Generator (PCG)]
\label{def:pcg}
Let $\cC$ be a reverse-sam\-pleable correlation generator.
A \emph{pseudorandom correlation generator (PCG) for $\cC$} is a pair of algorithms $(\PCG.\Gen,\allowbreak\PCG.\Expand)$ with the following syntax:
\begin{itemize}
	\item $\PCG.\Gen(1^\lambda)$ is a PPT algorithm that given a security parameter $\lambda$, outputs a pair of seeds $(\key_0,\key_1)$;

	\item $\PCG.\Expand(\sigma,\key_\sigma)$ is a polynomial-time algorithm that given a party index $\sigma\in\bit$ and a seed $\key_\sigma$, outputs a bit string $\R_\sigma\in \{0,1\}^n$.
\end{itemize}
The algorithms $(\PCG.\Gen,\PCG.\Expand)$ should satisfy the following:
\begin{itemize}
\item {\bf Correctness.} The correlation obtained via:
\begin{align*}\{(\R_0,\R_1) \mid (\key_0,\key_1)\getsr\PCG.\Gen(1^\lambda),R_\sigma\gets \PCG.\Expand(\sigma,\key_\sigma)\text{ for }\sigma\in\{0,1\}\}\end{align*}
is computationally indistinguishable from $\cC(1^\secpar)$.

\item {\bf Security.} For any $\sigma\in\bit$, the following two distributions are computationally indistinguishable:
	\begin{align*}
	\{ (\key_{1-\sigma},\R_\sigma) \mid (\key_0,\key_1)\getsr\PCG.\Gen(1^\lambda), &\R_\sigma\gets \PCG.\Expand(\sigma,\key_\sigma) \} \; \textrm{and} \\
	\{ (\key_{1-\sigma},\R_\sigma) \mid (\key_0,\key_1)\getsr\PCG.\Gen(1^\lambda), &\R_{1-\sigma}\gets \PCG.\Expand(\sigma,\key_{1-\sigma}), \\
	&R_\sigma\getsr\RSample(\sigma,\R_{1-\sigma}) \} \\
	\end{align*}
	where $\RSample$ is the reverse sampling algorithm for correlation $\cC$.
\end{itemize}
\end{definition}

Note that $\PCG.\Gen$ could simply output a sample from $\cC$. To avoid this trivial construction, we also require that the seed size is significantly shorter than the output size.

\subsubsection{Programmable PCG's.} At a high level, a programmable \PCG allows generating multiple \PCG keys such that part of the correlation generated remains the same accross different instances. Programmable \PCGs are necessary to construct $n$-party correlated randomness from the $2$-party correlated randomness generated via the \PCG. Informally, this is because when expanding $n$-party shares (e.g. of Beaver triples) into a sum of $2$-party shares, the sum will involve many ``cross terms''; using programmable \PCGs allows maintaining consistent pseudorandom values accross these cross terms. We recall the formal definition below.

\begin{definition}[Programmable \PCG{}]
\label{def:programmablePCG}
A tuple of algorithms \\
$\PCG = (\PCG.\Gen, \PCG.\Expand)$ following the syntax of a standard
\PCG{}, but where $\PCG.\Gen(1^{\lambda})$ takes additional random
inputs $\rho_0, \rho_1 \in \{0,1\}^{\kappa}$, for a fixed parameter $\kappa$ of size $ \text{\sf poly}(\lambda)$, is a {\em programmable
  \PCG{}} for a simple bilinear 2-party correlation $C_e^n$ (specified
by a bilinear pairing $e: \GG_{1} \times \GG_2 \rightarrow \GG_T$ for
some groups $\GG_1, \GG_2$ and $\GG_T$) if the following holds:

\begin{itemize}
    \item[$\bullet$] \textbf{Correctness.} The correlation obtained via:
\begin{equation*} \left \{
       \begin{matrix}
         ((R_0,S_0),(R_1,S_1))
     \end{matrix}
     \Bigg |
     \begin{matrix}
       \rho_0, \rho_1 \overset{\$} {\leftarrow} \{0,1\}^{\kappa}, (k_0,k_1) {\leftarrow} \PCG.\Gen(1^{\lambda}, \rho_0,\rho_1),\\
       (R_{\sigma},S_{\sigma}) {\leftarrow} \PCG.\Expand(\sigma, k_{\sigma}) \text{ for }  \sigma \in \{0,1\}  \\
  \end{matrix}  \right \}
\end{equation*}

        is computationally indistinguishable from
        $C_{e}^{n}(1^{\lambda})$.

  \item[$\bullet$] \textbf{Programmability} There exist public
        efficiently computable functions
        $\phi_0 : \{0,1\}^{*} \rightarrow \GG_1^n$,
        $\phi_1 : \{0,1\}^{*} \rightarrow \GG_2^n$ such that

\begin{equation*}
    \Pr \left [
      \begin{matrix}
        \rho_0, \rho_1 \overset{\$} {\leftarrow} \{0,1\}^{\kappa} , (k_0,k_1) {\leftarrow} \PCG.\Gen( 1^{\lambda},\rho_0,\rho_1)\\
     (R_0,S_0) \leftarrow \PCG.\Expand(0,k_0),\\
     (R_1,S_1) \leftarrow \PCG.\Expand(1,k_1)\\
  \end{matrix}
  \colon
      \begin{matrix}
       R_0 = \phi_0(\rho_0) \\
       R_1 = \phi_1(\rho_1)  \\
  \end{matrix} \right ] \geq  1 -  \text{\sf negl}(\lambda),
\end{equation*}

        where
        $e : \GG_1^n \times \GG_2^n \rightarrow \GG_T^n$
        is the bilinear map obtained by applying $e$ componentwise.

\item[$\bullet$] \textbf{Programmable security} The following pair of distributions are computationally indistinguishable
\begin{align*}
  &\Biggl\{ ({ k_1},(\rho_0,\rho_1)) \hspace{0.1cm} \Bigg|  \hspace{0.15cm}
    \rho_0, \rho_1 \overset{\$} {\leftarrow} \{0,1\}^{\kappa}, ({k}_0,{k}_1) {\leftarrow} \PCG.\Gen(1^{\lambda},\rho_0, \rho_1) \hspace{0.15cm} \Biggl \}   \hspace{0.3cm} \text{and} \\
    &\Biggl \{ ({k}_1,(\rho_0,\rho_1))  \hspace{0.1cm} \Bigg|  \hspace{0.15cm}
      \rho_0, \rho_1, \tilde{\rho_0}  \overset{\$} {\leftarrow} \{0,1\}^{\kappa}, ({k}_0,{k}_1) {\leftarrow} \PCG.\Gen(1^{\lambda}, \tilde{\rho_0},  \rho_1)  \hspace{0.1cm}\Biggl \}
\end{align*}
as well as the pair of distributions:
\begin{align*}
  &\Biggl\{ ({ k_0},(\rho_0,\rho_1)) \hspace{0.1cm} \Bigg|  \hspace{0.15cm}
    \rho_0, \rho_1 \overset{\$} {\leftarrow} \{0,1\}^{\kappa}, ({k}_0,{k}_1) {\leftarrow} \PCG.\Gen(1^{\lambda},\rho_0, \rho_1) \hspace{0.15cm} \Biggl \}   \hspace{0.3cm} \text{and} \\
    &\Biggl \{ ({k}_0,(\rho_0,\rho_1))  \hspace{0.1cm} \Bigg|  \hspace{0.15cm}
      \rho_0, \rho_1, \tilde{\rho_1}  \overset{\$} {\leftarrow} \{0,1\}^{\kappa}, ({k}_0,{k}_1) {\leftarrow} \PCG.\Gen(1^{\lambda}, \rho_0, \tilde{ \rho_1})  \hspace{0.1cm}\Biggl \}.
\end{align*}
\end{itemize}
\end{definition}

\section{From Decision-\QASD to Search-\QASD}
\label{app:std}
In this section, we describe a reduction from the search version of
$\QASD$ to the decision version, in the concrete chosen instantiations
(all instances over $\Ring = \FF_q[G]$ where $G = (\ZZ/(q-1)\ZZ)^n$,
which is the group we use to obtain \PCGs for \OLEs over
$\FF_q^{(q-1)^n}$). This reduction is actually a natural extension to
that of \cite{C:BomCouDeb22} to the multivariate setting, and
essentially applies in extreme regime of low rate, \ie more in the
\LPN regime. In \cite{C:BomCouDeb22}, the authors introduced a new
problem they called Function Field Decoding Problem ($\FFDP$) that we
recall below, which is the analogue of \RLWE with function fields
instead of number fields (See Section~\ref{sec:ntff} for a quick
reminder on the theory of algebraic function fields).

Let $K/\Fq(T)$ be a function field with constant field $\Fq$ and ring
of integers $\OO_{K}$, and let $Q(T)\in\Fq[T]$ be irreducible. Let
$\gothP\eqdef Q\OO_{K}$ be the ideal of $\OO_{K}$ generated by $Q$.
\FFDP is parametrized by a secret element $\bfs\in\OO_{K}/\gothP$, and
a noise distribution $\psi$ over $\OO_{K}/\gothP$ which is a finite set.

\begin{definition}[\FFDP distribution] A sample
  $(\bfa, \bfb)\in \OO_{K}/\gothP\times \OO_{K}/\gothP$ is distributed
  according to the \FFDP distribution modulo $\gothP$, with secret
  $\bfs$ and noise distribution $\psi$ if
  \begin{itemize}
    \item $\bfa$ is uniformly distributed over $\OO_{K}/\gothP$;
    \item $\bfb=\bfa\cdot\bfs+\bfe$ where $\bfe$ is distributed
          according to $\psi$.
        \end{itemize}
        A sample drawn according to this distribution will be denoted
        by $(\bfa,\bfb)\getsr\mathcal{F}_{\bfs, \psi}$.
\end{definition}
In its search version, the goal of \FFDP is to recover the secret
$\bfs$ given access to enough samples.
\begin{definition}[\FFDP (search version)] Let
  $\bfs\in\OO_{K}/\gothP$, and let $\psi$ be a probability
  distribution over $\OO_{K}/\gothP$. An instance of \FFDP consists in
  an oracle giving access to independent samples
  $(\bfa, \bfb)\getsr \mathcal{F}_{\bfs, \psi}$. The goal is to
  recover $\bfs$.
\end{definition}
In its decision version, the goal is to distinguish between the \FFDP
distribution and the uniform over
$\OO_{K}/\gothP\times \OO_{K}/\gothP$.
\begin{definition}[\FFDP (decision version)] Let $\bfs$ be drawn {\em
    uniformly} at random in $\OO_{K}/\gothP$, and let $\psi$ be a
  noise distribution over $\OO_{K}/\gothP$. Define the following two distributions:
  \begin{itemize}
    \item $\mathcal{D}_{0}:(\bfa, \bfb)$ uniformly distributed over
    $\OO_{K}/\gothP\times\OO_{K}/\gothP$.
    \item $\mathcal{D}_{1}:(\bfa, \bfa\cdot\bfs+\bfe)$ distributed
          according to the \FFDP distribution
          $\mathcal{F}_{\bfs,\psi}$.
        \end{itemize}
        Let $b\in\{0, 1\}$. Given access to an oracle $\code{O}_{b}$
        providing independent samples from distribution
        $\mathcal{D}_{b}$, the goal of the decision \FFDP is to
        recover $b$.
      \end{definition}
      Recall that a distinguisher between two distributions
      $\code{D}_{0}$ and $\code{D}_{1}$ is a PPT algorithm $\code{A}$
      that takes as input an oracle $\code{O}_{b}$ corresponding to
      distribution $\code{D}_{b}$ with $b\in\{0,1\}$ and outputs a bit
      $\code{A}(\code{O}_{b})\in\{0, 1\}$. The distinguisher wins when
      $\code{A}(\code{O}_{b})=b$. Its distinguishing advantage is
      defined as:
\[
  \adv_{\code{A}}(\code{D}_{0}, \code{D}_{1}) \eqdef \dfrac{1}{2} \Bigl( \Prob(\code{A}(\code{O}_{b}) = 1 \mid b = 1) - \Prob(\code{A}(\code{O}_{b}) = 1 \mid b = 0) \Bigr)
\]
and satisfies
\[
  \Prob(\code{A}(\code{O}_{b}) = b) = \frac{1}{2} + \adv_{\code{A}}(\code{D}_{0}, \code{D}_{1}).
\]

      The crucial remark of \cite{C:BomCouDeb22} was to notice that
      some structured variants of the decoding problem could be
      somehow {\em lifted} to the function field setting, and could be
      directly seen as instances of \FFDP.

    \begin{example}
        Let $n\in\NN$ and consider the polynomial
        \[
          F(T, X)\eqdef X^{n}+T-1\in\Fq(T)[X].
        \]
        Eisenstein criterion proves that $F$ is irreducible over
        $\Fq[T]$. Define the function field
        \[
          K \eqdef \Fq(T)[X]/(F(T,X)).
        \]
        Computing partial derivatives shows that the curve defined by
        $F(T, X)$ is non-singular, and therefore $\Fq[T, X]/(F(T,X))$
        is the full ring of integers $\OO_{K}$ of $K$ (see for
        instance \cite[Chapter VII]{L21}). Now, let
        $Q(T)\eqdef T\in\Fq[T]$. Then,
        \[
          \OO_{K}/T\OO_{K} = \Fq[T, X]/(X^{n}+T-1, T) = \Fq[X]/(X^{n}-1) = \Fq[\ZZ/n\ZZ].
        \]
        Therefore, \QASD with the group $\ZZ/n\ZZ$ can be seen as an
        instanciation of \FFDP with the function field
        $K=\Fq(T)[X]/(X^{n}+T-1)$, and modulus $Q(T)=T$.
      \end{example}

      A general search-to-decision reduction for \FFDP would therefore
      immediately provide a search-to-decision reduction for many
      variants of \QASD. However, adapting the reduction of
      \cite{EC:LyuPeiReg10} the authors of \cite{C:BomCouDeb22} were
      only able to give such a reduction with addition algebraic
      constraints on $K$ and $\gothP$. More precisely, they gave the
      following theorem

      \begin{theorem}[Search to decision reduction for
        \FFDP{}]\label{thm:std}
        Let ${K}/{\Fq(T)}$ be a Galois function field of degree $n$
        with field of constants $\Fq$, and denote by $\OO_{K}$ its
        ring of integers. Let $Q(T)\in\Fq[T]$ be an irreducible
        polynomial. Consider the ideal $\mathfrak{P}\eqdef Q\OO_{K}$.
        Assume that $\mathfrak{P}$ does not ramify in $\OO_K$, and
        denote by $f$ its inertia degree. Let $\distrib$ be a
        probability distribution over ${\OO_{K}}/{\gothP}$, closed
        under the action of $\Gal({K}/{\Fq(T)})$, meaning that if
        $\bfe\sample\distrib$, then for any
        $\sigma \in \Gal(K/\Fq(T))$, we have
        $\sigma(\bfe)\sample\distrib$. Let
        $\bfs\in {\OO_{K}}/{\gothP}$.

        Suppose that we have an access to $\ffd{\bfs, \distrib}$
        and there exists a distinguisher between the uniform
        distribution over ${\OO_{K}}/{\goth{P}}$ and the \FFDP{}
        distribution with uniform secret and error distribution
        $\distrib$, running in time $t$ and having an advantage
        $\eps$. Then there exists an algorithm that recovers
        $\bfs\in {\OO_{K}}/{\gothP}$ (with an overwhelming probability
        in $n$) in time
        \[
          O\left( \frac{n^{4}}{f^{3}}\times \frac{1}{\varepsilon^{2}} \times q^{f \deg(Q)}\times t\right).
        \]
\end{theorem}

Unfortunately, not all group algebras arise from Galois extensions of
function fields. Nonetheless, based on the analogy between cyclotomic
number fields and the Carlitz modules, they proposed to instantiate
their reduction with $K=\Fq(T)[\Lambda_{T}]=\Fq(T)[X]/(X^{q-1}+T)$,
and modulus $Q(T)=T+1$. The theory of Carlitz extensions ensures that $\OO_{K}=\Fq[T][X]/(X^{q-1}+T)$, and therefore
\[
  \OO_{K}/(T+1)\OO_{K} = \Fq[T, X]/(T+1, X^{q-1}+T) = \Fq[X]/(X^{q-1}-1) = \Fq[\ZZ/(q-1)\ZZ].
\]

The key point is the fact that $\Gal(K/\Fq(T))=\Fq^{\times}$ and an
element $\zeta\in\Fq^{\times}$ acts on $P(X)\in\Fq[X]/(X^{q-1}-1)$ by
$\zeta\cdot P(X) \eqdef P(\zeta X)$. In particular, the Galois group
keeps invariant the support of {\em any} element, and therefore any
distribution that only depends on the weight is Galois invariant.
Theorem \ref{thm:std} immediately yields a search-to-decision
reduction for \QASD instantiated with the group $G=\ZZ/(q-1)\ZZ$.

\subsubsection{Extension to the multivariate setting.} Consider the
group $G~=~(\ZZ/(q-1)\ZZ)^{t}$, and let
$\mathcal{R}\eqdef \Fq[G] = \Fq[X_{1},\dots, X_{t}]/(X_{1}^{q-1}-1,\dots,X_{t}^{q-1}-1)$.
Using the heavy machinery of inverse Galois theory, it is possible to
find a Galois extension of $\Fq(T)$ with Galois group $G$. However,
this would induce a large overhead in the complexity of the reduction.
Instead, in this case, building on the case of $\Fq[\ZZ/(q-1)\ZZ]$, we
can directly describe the reduction and get Theorem~\ref{thm:std_inst}
from Section~\ref{subsec:std_inst}.

The reduction works as follows. Recall that by the Chinese Remainder
Theorem,
\[
  \mathcal{R}=\prod_{(\zeta_{1},\dots, \zeta_{t})\in(\Fq^{\times})^{t}}\Fq[X_{1},\dots, X_{t}]/(X_{i}-\zeta_{i}),
\]
and fix an ordering of $(\Fq^{\times})^{t}$, which yields an ordering
$\mathfrak{I}_{1},\dots, \mathfrak{I}_{r}$ of the ideals in the above
decomposition (where $r=(q-1)^{t}$):
\[
  \mathcal{R} = \prod_{i=1}^{r}\Fq[X_{1},\dots, X_{t}]/\mathfrak{I}_{i}.
\]

Let $w\in\{0,\dots, (q-1)^{t}\}$ and $\bfs\in\mathcal{R}$. Consider a
noise distribution $\psi=\psi_{w}$ over $\mathcal{R}$ such that
$\Expect[\HW{x}]=w$ when $x$ is sampled according to $\psi$. A sample
$(\bfa, \bfb)$ is distributed according to $\ffd{\bfs, \psi}$ if
$\bfa$ is uniformly distributed in $\mathcal{R}$, and
$\bfb=\bfa\cdot\bfs+\bfe$ where $\bfe\getsr \psi$.

The idea of the reduction is to recover the secret modulo one of the
factors, and then using the action of some group recover the full
secret. We keep a high level, the first steps of the reduction
following {\em exactly} the same path as that of \cite{C:BomCouDeb22}.
The only difference resides in the last step and the considered group
action.

\begin{description}
  \item[Step 1: Randomizing the secret.] In the {\em decision} version, the secret $\bfs$ is supposed to be {\em uniformly} distributed over $\mathcal{R}$, while in the {\em search} version, the secret is {\em fixed}. In other words, the decision version is an average-case problem, while the search version is worst-case. Fortunately, the secret can be easily randomized by sampling some $\bfs'$ uniformly at random in $\mathcal{R}$. Now, for each sample $(\bfa, \bfb)\sample\ffd{\bfs, \psi}$ with a fixed secret $\bfs$, we can build the sample $(\bfa, \bfb + \bfa\cdot\bfs')$ which is distributed according to $\ffd{\bfs+\bfs', \psi}$, and the secret is now uniformly distributed. Feeding the latter sample to a distinguisher, allows to creates a distinguisher for a fixed secret, with exactly the same advantage.
  \item[Step 2: Hybrid argument.] A sample $(\bfa, \bfb)$ is said to follow the hybrid distribution $\mathcal{H}_{i}$ if it is of the form $(\bfa', \bfb'+\bfh)$ where $\bfh$ is uniformly distributed modulo $\mathfrak{I}_{j}$ for $j\le i$, and is $0$ modulo $\mathfrak{I}_{j}$ for $j>i$. Such an $\bfh$ is easily constructed using the Chinese Remainder Theorem. In particular, $\mathcal{H}_{0}=\ffd{s,\psi}$ and $\mathcal{H}_{r}$ is the uniform distribution over $\mathcal{R}$. A simple hybrid argument proves that a distinguisher between $\mathcal{H}_{0}$ and $\mathcal{H}_{r}$ with advantage $\eps$ can be turned into a distinguisher between $\mathcal{H}_{i_{0}}$ and $H_{i_{0}+1}$ for some $i_{0}$, with advantage at least $\frac{\eps}{r}$.
  \item[Step 3: Guess and search.] Given $i_{0}$, the idea is to make a guess $g_{i_{0}}$ for $\bfs$ modulo $\mathfrak{I}_{i_{0}}$ and to use the previous distinguisher to tell whether this guess is correct, or not.
        Define $\bfg, \bfh$ and $\bfv\in\mathcal{R}$ such that:

        \begin{minipage}{0.45\textwidth}
        \[
        \bfg = \left\lbrace
        \begin{array}{ll}
          g_{i_{0}}& \mod \mathfrak{I}_{i_{0}}\\
          0 & \text{elsewhere}
        \end{array}
        \right.
        \]
        \end{minipage}
\begin{minipage}{0.45\textwidth}
        \[
        \bfv = \left\lbrace
        \begin{array}{ll}
          \text{random} & \mod \mathfrak{I}_{i_{0}}\\
          0 & \text{elsewhere}
        \end{array}
        \right.
        \]
        \end{minipage}\ \\

        and $\bfh$ is uniformly distributed modulo $\mathfrak{I}_{j}$ for $j\le i_{0}+1$ and $0$ elsewhere.
        Then, for each sample $(\bfa, \bfb)$, we can build the sample $(\bfa', \bfb')$ where
        \[
        \left\lbrace
        \begin{array}{l}
          \bfa'=\bfa+\bfv\\
          \bfb'=\bfb+\bfh+\bfv\cdot\bfg=\bfa'\cdot\bfs + \bfe +\bfh + \bfv(\bfg-\bfs).
        \end{array}\right.
        \]
        Using the fact that all the factors $\Fq[X_{1},\dots,X_{t}]/\mathfrak{I}_{j}$ are finite fields (isomorphic to $\Fq$), it is easily seen that $(\bfa', \bfb')$ is distributed according to $\mathcal{H}_{i_{0}}$ when the guess is correct, and according to $\mathcal{H}_{i_{0}+1}$ otherwise. Therefore, by the Chernoff-Hoeffding bound, using our distinguisher $\Theta(n(r/\eps)^{2})$ times, one can detect if the guess is correct or not with probability at least $1-2^{-\Theta(n)}$. An exhaustive search on $\Fq$ yields the value of $\bfs\mod\mathfrak{I}_{i_{0}}$.
  \item[Step 4: A group action] This is the only step that changes from the reduction of \cite{C:BomCouDeb22}. We need to find a group $\widehat{G}$  that will replace the Galois group of their reduction in permuting the factors.

        Inspired by the univariate example, let $\widehat{G}\eqdef (\Fq^{\times})^{t}$. It acts on $\mathcal{R}$ by:
        \[
        (\zeta_{1},\dots,\zeta_{t})\cdot P(X_{1},\dots, X_{t}) \eqdef P(\zeta_{1}X_{1},\dots, \zeta_{t}X_{t}).
        \]

        The key observation here is that
        \begin{enumerate}
          \item This action keeps invariant the support of elements in
                $\mathcal{R}$. In particular, the distribution $\psi$
                is invariant under the action of $\widehat{G}$.
                \item $\widehat{G}$ acts transitively on the factors:

                The action of $(\zeta_{1},\dots,\zeta_{t})$ maps the
                ideal $(X_{1}-\gamma_1,\dots,X_{t}-\gamma_{t})$ onto
                $(X_{1}-\zeta_{1}^{-1}\gamma_1,\dots,X_{t}-\zeta_{t}^{-1}\gamma_{t})$.
              \end{enumerate}

        Now, in order to recover $\bfs\mod\mathfrak{I}_{j}$ for $j\neq i_{0}$, it suffices to take the (unique) element $\bfz\in\widehat{G}$ such that $\bfz\cdot \mathfrak{I}_{j} = \mathfrak{I}_{i_{0}}$, and for any sample $(\bfa, \bfb=\bfa\bfs+\bfe)$ we can build $(\bfz\cdot \bfa, \bfz\cdot\bfb = (\bfz\cdot\bfa)(\bfz\cdot\bfs)+\bfz\cdot \bfe)$. Note that $\bfa'$ (resp. $\bfz\cdot\bfe$) is still uniformly distributed over $\mathcal{R}$ (resp. distributed according to $\psi$ since it is $\widehat{G}$-invariant). In other words, $(\bfa', \bfb')$ is distributed according to $\ffd{\bfz\cdot \bfs, \psi}$, and repeating the first three steps of the reduction will yield $\bfz\cdot\bfs\mod \mathfrak{I}_{i_{0}}$, which is equal to $\bfs \mod \bfz^{-1}\cdot \mathfrak{I}_{i_{0}} = \bfs \mod \mathfrak{I}_{j}$, which concludes the reduction.
\end{description}

\section{Algebraic number theory in function fields}
\label{sec:ntff}
There is a well-known analogy between the theory of finite extensions
of $\QQ$, the so-called {\em number fields}, and that of finite
separable extensions of $\Fq(T)$, the field of rational functions with
coefficients in a finite field $\Fq$. The latter algebraic extensions
are called {\em function fields}, because they can be realized as
fields of rational functions on curves over finite fields. In this
section, we recall the minimal requirements about the arithmetic of
function fields that are needed in the sequel. A dictionnary
summarizing the analogies between function fields and number fields is
represented in Table \ref{table:analogy_nf_ff} below.

\begin{table}[ht]
  \centering
\[
  \begin{array}{|c|c|}
    \hline
    \text{ Number fields } & \text{ Function fields } \\
    \hline
    \QQ & \Fq(T) \\
    \ZZ & \Fq[T] \\
    \text{Prime numbers } q \in \ZZ & \text{Irreducible polynomials } Q \in \Fq[T]\\
                           & \\
    K = \fract{\QQ[X]}{(f(X))} & K = \fract{\Fq(T)[X]}{(f(T, X))}\\
                           & \\
    \makecell{\OO_{K} \\ = \text{Integral closure of $\ZZ$} \\ \text{\emph{Dedekind} domain}}  & \makecell{\OO_{K} \\ = \text{Integral closure of $\Fq[T]$} \\ \text{\emph{Dedekind} domain}} \\
                           & \\
    { \textbf{characteristic 0}} & {\bf \textbf{characteristic} >0}\\
    \hline
  \end{array}
\]
\caption{A Number-Function fields analogy}
\label{table:analogy_nf_ff}
\end{table}

\subsection{Algebraic function fields.} Starting from a finite
field $\Fq$, a {\em function field} is a finite extension $K$ of
$\Fq(T)$ of the form
\[
  K = \Fq(T)[X]/(P(T,X)),
\]
where $P(T, X)\in\Fq(T)[X]$ is irreducible. The field
$K\cap \overline{\mathbb{F}}_q$ is referred to as {\em the field of constants}
of $K$. In general, this is a (finite) extension of $\Fq$, but when
$\Fq$ is the full field of constants of $K$, the extension $K/\Fq(T)$
is said to be {\em geometric}. This is equivalent for the modulus
$P(T,X)$ to be irreducible, even regarded as a polynomial in
$\overline{\Fq}(T)[X]$ (\cite[Cor, 3.6.8]{S09}). This will always be
assumed in our setting.

Similarly to the number field setting, the integral closure of
$\Fq[T]$ in $K$ is called the {\em ring of integers} of $K$, and
denoted by $\OO_{K}$. This is a {\em Dedeking domain}. In particular, for any ideal
$\mathfrak{P}$ of $\OO_{K}$, there exist unique prime ideals
$\mathfrak{P}_{i}$ and integers $e_{i}$ such that
$\mathfrak{P}=\mathfrak{P}_{1}^{e_{1}}\dots\mathfrak{P}_{r}^{e_{r}}$,
and the quotients $\OO_{K}/\gothP_{i}$ are finite extensions of $\Fq$.
When the ideal $\mathfrak{P}$ is of the form $P\OO_{K}$ where
$P(T)\in \Fq[T]$ is an irreducible polynomial, the primes
$\mathfrak{P}_{i}$ are said to be {\em lying above}
$P$\footnote{Rigourously, there exists another prime element in
  $\Fq(T)$, which is $1/T$. This element does not belong to $\Fq[T]$,
  and corresponds to the point at infinity on the projective line. To
  take into account this additional point, we could consider
  the ring $\Fq[1/T]$ (or its localization $(\Fq[1/T])_{1/T}$ to avoid
  redundancy), and its integral closure $\OO_{K, \infty}$ in $K$. It
  is also a Dedekind domain, and the primes of $\OO_{K,\infty}$ lying
  above $(1/T)$ are known as the {\em places at infinity}. The main
  difference with the number field setting being that this place at
  infinity plays a similar role as the other primes (which are called
  {\em finite places} in opposition), while in number fields the {\em
    places at infinity} are called {\em archimedean places} would
  correspond to the complex embeddings of $K$ in $\CC$.}. The extension
degrees
$f_{i} \eqdef [\OO_{K}/\gothP_{i} \colon \Fq[T]/(P(T))] = [\OO_{K}/\gothP_{i}\colon \FF_{q^{\deg P}}]$
are called the {\em inertia degrees} of $P$, and $e_{i}$ are known as
its {\em ramification} indexes. When the $e_{i}$'s are all equal to
$1$, the extension is said to be {\em unramified} at $P$. In that
case, the Chinese Remainder Theorem entails that $\OO_{K}/\gothP$ is
isomorphic to $\prod_{i=1}^{r}\OO_{K}/\gothP_{i}$ which is a product
of finite fields. All those quantities are related through the
well-known formula
\begin{equation}\label{eq:fondamental_relation}
  n \eqdef [K:\Fq(T)] = \sum_{i=1}^{r}e_{i}f_{i}.
\end{equation}
\subsection{Galois extensions.} Recall that the extension $K/\Fq(T)$
is said to be {\em Galois} when the automorphism group
\[
  \Aut(K/\Fq(T)) \eqdef \{\sigma\colon K \mapsto K \mid \sigma \text{
    is an isomorphism with } \sigma(a)=a \;\; \forall a\in \Fq(T)\}
\]
has cardinality $[K:\Fq(T)]$. In that case, this group is usually
denoted by $\Gal(K/\Fq(T))$ and known as the {\em Galois} group of
$K$. Galois extensions whose Galois group is abelian are called {\em
  abelian} extensions. This Galois group adds more symmetry to the
function field. More specifically, $G$ keeps $\OO_{K}$ globally
invariant and given an irreducible polynomial $Q(T)\in\Fq[T]$, it acts
transitively on the prime ideals lying above $Q$ (\ie it permutes the
factors). In particular, all the ramification indexes $e_{i}$ (resp.
the inertia degrees $f_{i}$) are equal, and denoted by $e$ (resp.
$f$):
\[
  Q\OO_{K} = (\gothP_{1}\dots\gothP_{r})^{e},
\]
and Equation \eqref{eq:fondamental_relation} simply becomes $n=efr$.
In this work, we sometimes need to work with different extensions. When
the context is not clear, we will put the irreducible polynomial in
index, and the considered function field in brackets: $e_{Q}(K)$ and
$f_{Q}(K)$. Another consequence is that the action of $G$ on $\OO_{K}$
is well defined on the quotient $\OO_{K}/Q\OO_{K}$ and simply permutes
the factors $\OO_{K}/\gothP_{i}^{e}$. The {\em decomposition group}
$D_{\gothP_{i}/Q}$ of $\gothP_{i}$ over $Q$ is the subgroup of Galois
automorphisms keeping $\gothP_{i}$ globally invariant
\[
  D_{\gothP_{i}/Q} \eqdef \{\sigma\in G \mid \sigma(\gothP_{i})=\gothP_{i}\}.
\]
It has cardinality $e\times f$. When $K$ is unramified at $Q$, the
ring $\OO_{K}/\gothP_{i}$ is the finite field $\F{q^{f\deg(Q)}}$ and
the action of $D_{\gothP_{i}/Q}$ is that of the Frobenius
endomorphism: the reduction modulo $\gothP_{i}$ yields an isomorphism
\[
  D_{\gothP_{i}/Q} \simeq \Gal(\F{q^{f\deg(Q)}}/\FF_{q^{\deg(Q)}}).
\]
The decomposition groups of all the primes above $Q$ are conjugate in
$\Gal(K/\Fq(T))$: For any $i\neq j$ there exists $\sigma\in G$ such
that $D_{\gothP_{i}/Q} = \sigma D_{\gothP_{j}/Q}\sigma^{-1}$. In
particular, when the extension is {\em abelian}, they are all equal
and referred to as the {\em decomposition group} of $Q$, and denoted by
$D_{Q}$. The subfield
\[
  L\eqdef K^{D_{Q}} = \{x\in K \mid \sigma(x)=x\quad\forall \sigma \in D_{Q}\}
\]
of all elements of $K$ fixed {\em pointwise} by $D_{Q}$ is called the
{\em decomposition field} of $Q$. It is an algebraic function field,
with ring of integers $\OO_{L}=\OO_{K}^{D_Q}$ consisting in all the
elements of $\OO_{K}$ pointwise fixed by $D_{Q}$. Moreover, it is a
Galois extension with Galois group $G/D_{Q}$. This is the largest
subextension of $K$ in which $Q$ {\em totally splits}\footnote{Hence
  the name decomposition field} (\ie
$f_{Q}(K^{D_{Q}})=e_{Q}(K^{D_{Q}})=1$ and
$r_{Q}(K^{D_{Q}}) = r_{Q}(K) = r$).

\begin{center}
  \begin{tikzpicture} \matrix (m) [matrix of math nodes,row
    sep=3em,column sep=4em,minimum width=2em] {
      \OO_K & K \\
      \OO_K^{D_Q} & K^{D_Q}\\
      \Fq[T] & \Fq(T) \\};

    \path[-] (m-1-1) edge (m-2-1) edge (m-1-2);
    \path[-] (m-2-1) edge (m-3-1) edge (m-2-2);
    \path[-] (m-2-2) edge (m-1-2) edge (m-3-2);
    \path[-] (m-3-1) edge (m-3-2);

    \draw (m-3-1) node[left,xshift=-0.5cm] {$(Q) \subset$};
    \draw (m-2-1) node[left,xshift=-0.5cm] {$(Q) = \mathfrak{p}_{1}\dots\mathfrak{p}_{r} \subset$};
    \draw (m-1-1) node[left,xshift=-0.5cm] {$(Q) = \mathfrak{P}_{1}\dots\mathfrak{P}_{r} \subset$};

    \draw (m-1-2) node[right,xshift=1cm] {$\OO_{K}/\gothP_{i} = \F{q^{f\deg(Q)}}$};
    \draw (m-2-2) node[right,xshift=1cm] {$\OO_{K}/\mathfrak{p}_{i} = \F{q^{\deg(Q)}}$};
    \draw (m-3-2) node[right,xshift=1cm] {$\Fq[T]/Q = \F{q^{\deg(Q)}}$};
\end{tikzpicture}
\end{center}

\subsection{The Carlitz module}
\label{sec:carlitz}
In classical algebraic number theory, the cyclotomic number fields
play a major role. For instance, all abelian extensions of $\QQ$ can
be realized as subfields of some cyclotomic number fields. This is
known as the Kronecker-Webber Theorem, and is the cornerstone of the
very important class field theory.

In the theory of algebraic function fields, the analogues of the
cyclotomic extensions of $\QQ$ are known as the {\em Carlitz
  extensions}. They were discovered by Carlitz in the late 1930's and
the analogy with the cyclotomic number fields was made explicit by his
student Hayes about $40$ years later in \cite{H74} to give an analogue
of the Kronecker-Webber Theorem for the rational function field
$\Fq(T)$. This result was later generalized by Drinfeld and Goss to
yield a complete solution to Kronecker's Jugendtraum\footnote{``childhood dream'' in German.} for function fields, \ie an explicit class
field theory. In the number field setting, such an explicit
construction is only known for $\QQ$, via the cyclotomic number
fields, and for imaginary quadratic number fields, via the theory of
elliptic curves with complex multiplication.

In this section, we just give a quick presentation of the Carlitz
modules, keeping the same notations as \cite{C:BomCouDeb22}. We refer
to \cite[\S~V]{C:BomCouDeb22} paper for a self-contained presentation
(without proofs). For an in-depth exposition, the interested reader
can refer to \cite[Chapter 12]{R02}, \cite[Chapter 12]{VS06}, or the
survey \cite{Conrad-Carlitz}.

A dictionnary summarizing the analogies between cyclotomic number
fields and Carlitz extensions is given in Table~\ref{table:carlitz}.

\begin{table}[!ht]
  \centering
      \[
      \begin{array}{|c|c|}
        \hline
        \QQ & \Fq(T) \\
        \ZZ & \Fq[T] \\
        \text{Prime numbers } q \in \ZZ & \text{Irreducible polynomials } Q\in\Fq[T]\\
            & \\
        \mu_{m} = \ideal{\zeta} \simeq \ZZ/m\ZZ \text{ (groups) }& \Lambda_{M} =\ideal{\lambda}\simeq \Fq[T]/(M) \text{ (modules) }\\
            & \\
        d \mid m \Leftrightarrow \mu_{d} \subset \mu_{m} \text{ (subgroups) } & D\mid M \Leftrightarrow \Lambda_{D} \subset \Lambda_{M} \text{ (submodules) }\\
            & \\
        a \equiv b \mod m \Rightarrow \zeta^{a} = \zeta^{b} & A \equiv B \mod M \Rightarrow [A](\lambda) = [B](\lambda)\\
& \\
        K = \QQ[\zeta] & K = \Fq(T)[\lambda] \\
        \OO_{K} = \ZZ[\zeta] & \OO_{K} = \Fq[T][\lambda] \\
            & \\
        \Gal(K/\QQ) \simeq (\ZZ/m\ZZ)^{\times} & \Gal(K/\Fq(T))\simeq (\Fq[T]/(M))^{\times}\\
            & \\
        \textbf{Cyclotomic} & \textbf{Carlitz}\\
        \hline
      \end{array}
    \]
    \caption{\label{table:carlitz} Analogies between cyclotomic and
      Carlitz}
\end{table}

If one wants to build cyclotomic extensions of $\Fq(T)$, the most
natural idea is to mimic the construction of cyclotomic number fields
and to add roots of unity to $\Fq(T)$. However, the crucial difference
with $\QQ$ is that roots of unity are already {\em algebraic} over
$\Fq$, and adjoining them to $\Fq(T)$ only yields a function field of
the form $\FF_{q^{m}}(T)$, {\ie an extension of the constants}.

Instead, one needs to look deeper into the algebraic structure
adjoined to $\QQ$. Notice that roots of unity form an abelian group,
that is to say a $\ZZ$-module. More precisely, consider the action of
$\ZZ$ on $\overline{\QQ}^{\times}$ by exponentiation:
$m\cdot z \eqdef z^{m}$. Then, the $m$-th roots of unity are nothing
else than the {\em torsion elements} of the action of $m\in\ZZ$:
\[
  \mu_{m}=\{z\in\overline{\QQ^{\times}}\mid m\cdot z = 1\}.
\]
At a high level, the philosophy behind the construction of Carlitz
extension is to replace $\ZZ$ by $\Fq[T]$ when {\em that makes sense},
and therefore {\em abelian groups} by {\em $\Fq[T]$-modules}. In
particular, the analogue of the exponentiation will be a new action of
$\Fq[T]$ on $\overline{\Fq(T)}$, called the {\em Carlitz action}. This
yields another structure of $\Fq[T]$-module on $\overline{\Fq(T)}$,
which is called the {\em Carlitz module}. If $M\in\Fq[T]$, the elements of
$M$--torsion are denoted
$\Lambda_{M}\eqdef \{ \lambda \in \overline{\Fq(T)}\mid M\cdot \lambda = 0\}$,
and form a {\em cyclic} $\Fq[T]$-module, generated by some element
denoted $\lambda_{0}$, which is an analogue of a {\em primitive} root
of unity. The Carlitz extension by $M$ will then be
$\Fq(T)[\Lambda_{M}]=\Fq(T)[\lambda_{0}]$. It is a {\em Galois}
extension of $\Fq(T)$ of Galois group isomorphic to
$(\Fq[T]/(M))^{\times}$.

One key fact about Carlitz extensions is that their ring of integers is
simply $\Fq[T][\lambda_{0}]$ and the decomposition of primes is well
understood:

\begin{theorem}[{\cite[Th.~12.10]{R02}}]\label{thm:carlitz_splitting}
  Let $M\in\Fq[T]$, $M\neq 0$, and let $Q\in\Fq[T]$ be a monic,
  irreducible polynomial. Consider the Carlitz extension $K_{M}$ and
  let $\OO_{M}$ denote its ring of integers. Then,
  \begin{itemize}[label=\textbullet]
    \item If $Q$ divides $M$, then $Q\OO_{M}$ is totally ramified.
    \item Otherwise, let $f$ be the smallest integer $f$ such that
          $Q^{f} \equiv 1 \mod M$. Then $Q\OO_{M}$ is unramified and
          has inertia degree $f$. In particular, $Q$ splits completely
          if and only if $Q\equiv 1 \mod M$.
  \end{itemize}
\end{theorem}

Those results are completely analogue to the cyclotomic case.

\section{The Curious Case of \texorpdfstring{$\F{2}$}{F2}}
\label{sec:f2}
In Section \ref{sec:pcg}, we showed how to produce batch \OLEs over all
finite fields $\Fq$ for $q\ge 3$. However, this approach cannot be
applied as is to build \OLEs over $\F2$. The most natural approach to
mimic previous construction is to consider the ring
$\mathbb{B}_{n}\eqdef \F2[X_{1},\dots, X_{n}]/(X_{i}^{2}-X_{i})$ of
Boolean functions, for which efficient algorithmics exist and which is
isomorphic to a direct product of $n$ copies of $\F2$. However, there
is a strong bias which is very similar to the one mentioned
in Example~\ref{example:linear_attack_group_codes}. Indeed suppose we are given
a pair $(a, as + e)$ where $a \getsr \mathbb{B}_n$
and $s,e$ are sparse with respect to the basis of monomials.
Then, the constant term of both $s, e$ is very likely to be zero and the constant
term is nothing but their evaluation at $0$. Moreover, the evaluation at $0$
map commutes with the reduction modulo $((X_i^2-X_i))_i$.
Consequently one can evaluate our sample at $0$ and the result is highly
biased since $(as+e)(0) = a(0)s(0) + e(0)$ and hence is $0$ whenever $s,e$
both vanish at $0$ which is highly probable.
Here again, we have a distinguisher on codes from a multivariate ring
which is {\em not} a group algebra.

More generally, we
have the following simple, but powerful, impossibility result.

\begin{theorem}[Impossibility result]\label{thm:impossibility}
  Let $G$ be a finite group and let $\mathcal{R} = \FF[G]$ be its group
  algebra with coefficients in a finite field $\FF$. Assume that
  $\mathcal{R}$ is isomorphic, as algebra, to $\F2^{N}$ for $N\ge 1$.
  Then, $N=1$ and $G = \{1\}.$
\end{theorem}

\begin{proof}
  Note that $G$ can be embedded in the invertible elements
  $\mathcal{R}^{\times}$ of $\mathcal{R}$. Indeed, any $g\in G$, when
  regarded as an element of $\Ring$ is invertible, with inverse
  $g^{-1}$. In particular, $|G| \le |\mathcal{R}^{\times}|$. But the
  algebra isomorphism $\mathcal{R} \simeq \F2^{N}$ induces a group
  isomorphism
  $\mathcal{R}^{\times} \simeq \F2^{\times}\times \cdots \times \F2^{\times} = \{(1,\dots, 1)\}$.
  In particular, $|\mathcal{R}^{\times}| = 1$, and $|G| = 1$, \ie
  $G = \{1\}$, which concludes the proof.
\end{proof}
Theorem \ref{thm:impossibility} shows that we cannot adapt directly
our approach based on \QASD to efficiently build \OLEs over $\F2$. In
this Section though, we propose a way to overcome this limitation. In
a nutshell, our approach is to consider the group algebra $\F2[G]$ of
some well-chosen finite abelian group $G$, as was done previously,
such that there is an isomorphism of {\em modules} between $\F2[G]$
and $\F2^{N}$, but not of {\em algebras}. It turns out that this
approach is not so different from the proposal of \cite{C:BCGIKS20}
which uses the ring $\Fp[X]/(P(X))$ where $P(X) = X^{2^{\ell}}+1$ is a
{\em cyclotomic polynomial} and $p$ is a prime such that
$p\equiv 1\mod 2^{\ell+1}.$ Indeed, our proposal uses the theory of
 {\em Carlitz extensions} (see Section \ref{sec:carlitz})
 which are function fields analogues of cyclotomic number fields.

 \subsection{An attempt based on the Carlitz module.}

 In \cite{C:BCGIKS20}, the authors propose to use a cyclotomic ring
 modulo some prime $p$. The natural idea to mimic their construction
 would be to make use of Carlitz extensions.

Consider the rational function field $\F2(T)$, endowed with the
 Carlitz action, and let
 \[
   K_{\ell}\eqdef \F2(T)[\Lambda_{T^{\ell+1}}]
 \]
 for some positive integer $\ell$ to be detailed later. The theory of
 Carlitz modules asserts that $K_{\ell}$ is a Galois extension of
 $\FF_{2}(T)$ of degree $2^{\ell}$, and of Galois group
 \[
   G\eqdef \left(\F2[T]/(T^{\ell+1})\right)^{\times}.
 \]
 The first idea that comes to mind is to find an irreducible modulus
 $Q(T)\in\F2[T]$ that splits completely in $\OO_{K}$ so that
 \[
   \OO_{K}/Q\OO_{K} \simeq \FF_{2^{\deg(Q)}}\times \dots \times \FF_{2^{\deg(Q)}} \simeq \FF_{2} \times \dots \times \FF_{2}.
 \]
 On the one hand, this shows that a necessary condition is
 $\deg(Q)=1$. On the other hand, by
 Theorem~\ref{thm:carlitz_splitting}, the ideal $Q\OO_{K_{\ell}}$
 splits completely if and only if $Q\equiv 1 \mod T^{\ell+1}$. In
 particular, $\deg(Q)$ needs to be large enough, and both conditions
 are incompatible.

 Therefore, one needs to relax some of the hypotheses above in order
 to make this idea somehow work. Clearly, the first condition
 ($\deg(Q)=1$) cannot be released, because all factors of
 $\OO_{K}/(Q\OO_{K})$ are extension fields of $\FF_{2^{\deg(Q)}}$, of
 dimension the inertia degree of the ideal generated by $Q$.
 Therefore, the only condition that can be relaxed is the second one.

Let $Q\in\F2[T]$ be an irreducible polynomial of degree $1$. There are
only two possibilities, namely $Q=T$ or $Q=T+1$. However, by
Theorem~\ref{thm:carlitz_splitting}, $T$ ramifies in $\OO_{K}$, and
therefore the only possible choice for $Q$ is $T+1$. Now, we need to
compute the inertia degree. By the aformentionned theorem, it is
characterized by the multiplicative order of $T+1$ modulo
$T^{\ell+1}$. It is not hard to see that it is the least power of $2$
greater (or equal) than $\ell+1$.

In the sequel, we make a concrete choice for the parameter $\ell$.
Assume that we want to produce $2^{20}$ \OLEs correlations. This is an
estimation of the order of magnitude of the number of multiplicative
gates in a concrete arithmetic circuit.
The number of \OLEs produced being the number of factors, we need to
set $\ell>20$. As we will see, setting $\ell=25$ is
enough.

Indeed, the least power of $2$ greater than $26$ is $32=2^{5}$.
Therefore, $(T+1)$ has inertia degree $32$ in $\OO_{K}$.
Theorem~\ref{thm:carlitz_splitting} and
Equation~\ref{eq:fondamental_relation} entail that
  \[
    \OO_{K}/(T+1)\OO_{K}\equiv \underbrace{\FF_{2^{32}}\times\dots\times\FF_{2^{32}}}_{2^{20}\text{
        times }}.
  \]
  With only Carlitz extensions, this is the best that we can
  produce\footnote{With finite places}. However, we are not required
  to use the {\em full} Carlitz extension: We could consider an
  intermediate one, that would cancel the inertia. This is precisely
  the {\em decomposition field}:

  \begin{proposition}
    Let $\ell=25$, and let $K$ be the Carlitz extension by
    $T^{\ell+1}\eqdef T^{26}$. Let $D_{T+1}$ be the decomposition group of
    $T+1$, and let $L\eqdef K^{D_{T+1}}$ denote the fixed field by
    $D_{T+1}$. Then we have:
    \[
      \OO_{L}/(T+1)\OO_{L}=\underbrace{\F2\times \dots \F2}_{2^{20}\text{
          times }}.
    \]
  \end{proposition}
  Although we will not provide it here because it would only obfuscate
  the speech, $\OO_{K}$ has an explicit description of the form
  \[
    \OO_{K} = \F2[X]/(1+P(X))
  \]
  where $P$ is a linearized polynomial of degree $2^{25}$, \ie the
  only monomials that appear in $P$ are powers of $2$. In particular,
  $P$ has a very sparse description. On the other hand, the ring
  $\OO_{L}$ does not seem to inherit this property. It is only defined
  as the subring of $\OO_{K}$ fixed by $D_{T+1}$, and we need to
  understand how $D_{T+1}$ acts on $\OO_{K}$. Recall that by
  definition, $D_{T+1}$ acts as the Frobenius of each of the factors
  of $\OO_{K}/(T+1)\OO_{K}$. In other words, $\OO_{L}/(T+1)\OO_{L}$ is
  the subring of $\OO_{K}/(T+1)\OO_{K}$ fixed by the Frobenius on each
  factor (after applying the Chinese Remainder Theorem). Note that the
  action of $D_{T+1}$ can be directly understood on $\OO_{K}$ (before
  CRT): Indeed, it is isomorphic to the cyclic group (of order $32$)
  generated by $(T+1)\in\left(\F2[T]/T^{26}\right)^{\times}$, where
  $T+1$ acts on $F(T, X)\in\OO_{K}$ by $F(T, (T+1)\cdot X)$ via the
  Carlitz action on the second variable. In other words, $\OO_{L}$ is
  the subring of $\OO_{K}$ fixed by this Carlitz action.

  \subsection{Building \OLEs.}

  It suffices to build {\em one} \OLE over $\OO_{L}/(T+1)\OO_{L}$ to
  generate $2^{20}$ \OLEs over $\F2$. Let
  \[
    \mathcal{R}\eqdef \OO_{L}/(T+1)\OO_{L}.
  \]
  Following \cite{C:BCGIKS20}, in order to build an \OLE over
  $\mathcal{R}$, we could generate $U, V\in \mathcal{R}$ pseudorandom
  such that they admit a sparse description of the form
  $U=a\cdot e_{1}+f_{1}$ and $V=a\cdot e_{2}+f_{2}$ with $e_{i},f_{i}$
  somehow {\em sparse}. However, there are two issues here:

    \begin{itemize}
      \item How to assert pseudorandomness here ?
      \item What does it mean to have a sparse description in
            $\mathcal{R}$ ?
    \end{itemize}

    If in \cite{C:BCGIKS20} the sparsity is well defined in the {\em
      canonical} basis, it is not clear what basis to choose in
    $\mathcal{R}$. Note that $\OO_{K}/(T+1)\OO_{K}$ admits a monomial
    basis (this is a consequence of the fact that $K$ being a Carlitz
    extension, $\OO_{K}$ is generated over $\Fq[T]$ by a unique
    element $\lambda_{0}$), but it is no longer true for $\mathcal{R}$.

    However, as it was recalled in \cite{C:BomCouDeb22}, since $T+1$
    is not ramified in $\OO_{L}$, a result of Noether (see ~\cite{Noether32}
    for the original paper (in German)) entails that $\OO_{L}$ admits
    a {local normal basis}, \ie there exists $a\in\mathcal{R}$ such
    that $(\sigma(a))_{\sigma\in\Gamma}$ forms an $\F2$-basis of
    $\mathcal{R}$, where
    \[
      \Gamma=\Gal(L/\F2(T))=\fract{\left(\F2[T]/(T^{26})\right)^{\times}}{(T+1)}
    \]
    is the Galois group of $L$. In fact, here, the normal basis is
    easy to find: Indeed, $\Gamma$ acts transitively on the factors of
    $\mathcal{R}$. This means that starting from
    $e_{1}=(1, 0,\dots, 0)$, $\sigma(e_{1})$ is another element of the
    canonical basis of $\F2^{2^{20}}$ for all $\sigma\in\Gamma$. In
    particular, $e_{1}$ generates a normal basis of $\F2^{2^{20}}$,
    and therefore its inverse through the CRT generates a normal basis
    of $\mathcal{R}$. Let us call this polynomial $\eps(X)$.

    It is tantalizing to define the sparsity with respect to this
    basis. Moreover, the existence of this normal basis has a
    powerful consequence. Indeed, let $a\in \mathcal{R}$. Written in
    the normal basis, we have that
    \[
      a = \sum_{\sigma\in\Gamma}a_{\sigma} \sigma(\eps) = \underbrace{\left(\sum_{\sigma\in\Gamma}a_{\sigma}\sigma\right)}_{\eqdef A}(\eps(X)).
    \]
    In other words, we can write $a\in \mathcal{R}$ as $A(\eps(X))$
    where $A$ now belongs to the group algebra $\F2[\Gamma]$. This
    exactly means that $\mathcal{R} = \F2[\Gamma]\cdot \eps$, \ie that
    $\mathcal{R}$ is a free module of rank one over $\F2[\Gamma]$, \ie that
    $\mathcal{R}$ is {\em isomorphic} to the group algebra $\F2[\Gamma]$ as
    {\em modules}.

    \subsection{\QASD to the rescue.} With this result in hand, it is
    very appealing to define our \OLE over $\F2[\Gamma]$, and {\em then
      only} map it to $\mathcal{R}$, since hardness of \QASD would
    provide security.

    \begin{proposition}\label{prop:pseudorandom_in_R}
      Let $a$ be uniformly distributed over $\F2[\Gamma]$ and
      $e, f\in\F2[G]$ of Hamming weight $t$. Then, $a\cdot e + f$ is
      pseudorandom assuming the hardness of \QASD over $\F2[\Gamma]$.
    \end{proposition}

    \begin{remark}
    Note that Theorem~\ref{thm:GV_QA_code} also holds in the modular
    setting, and therefore it holds for $\F2[\Gamma]$. In particular,
    according to our analysis, \QASD with this instantiation is secure
    against linear tests.
  \end{remark}
  Now, since $\mathcal{R}$ is isomorphic (as a module) to $\F2[G]$, if
  $U\in\F2[\Gamma]$ is pseudorandom, then $U(\eps)\in\mathcal{R}$ is
  pseudorandom. Everything seems to be there for building an \OLE over
  $\mathcal{R}$: Let $U, V\in\mathcal{R}$ be such that
  $U = (a\cdot e_{1}+f_{1})(\eps(X))$ and
  $V = (a\cdot e_{2}+f_{2})(\eps(X))$ with $a$ uniformly distributed
  over $\F2[G]$, and $e_{i}, f_{i}$ sparse (as elements of $\F2[G]$).
  Proposition~\ref{prop:pseudorandom_in_R} entails that $U$ and $V$
  are pseudorandom in $\mathcal{R}$. Following~\cite{C:BCGIKS20}, if
  we can distribute additive shares of the product
  $U\cdot V \in \mathcal{R}$, we would win. However, here the
  operations do not commute, and we cannot use \FSS for point
  functions to distribute shares of the cross products. Indeed,
\begin{align*}
  U\times V & = (a\cdot e_{1}+f_{1})(\eps)\times(a\cdot e_{2}+f_{2})(\eps) \\
            & = (a\cdot e_{1})(\eps)\times(a\cdot e_{2})(\eps) + (a\cdot e_{1})(\eps)\times (f_{2})(\eps) + (a\cdot e_{2})(\eps)\times (f_{1})(\eps) \\
            & + (f_{1})(\eps)\cdot(f_{2})(\eps),
\end{align*}
and if every term admits a sparse presentation, it is not clear to us
how to distribute additive shares of them.

\subsection{A note on efficiency.} Even if the previous problem is
solved, there remains the question of efficiency. Indeed, fast
encoding of quasi-abelian codes, \ie fast multiplication in the group
algebra, is usually done through the Fast Fourier Transform, which
does not extend {\em a priori} to the modular setting since it is not
semisimple. However, a recent work of Hong, Viterbo and Belfiore
(\cite{HVB16}) developped a {\em modular} FFT over $\F2$ for the
specific group $(\ZZ/2\ZZ)^{s}$. Their algorithm is particularly
efficient because it only involves {\em additions}, and could be
optimized on hardware.

Our group $\Gamma$ is a little bit more complicated (see \cite[Proposition 2.4]{CL17}):
\[
  \Gamma \eqdef \fract{\left(\F2[T]/T^{26}\right)^{\times}}{(T+1)} \eqdef (\ZZ/2\ZZ)^{6}\times (\ZZ/4\ZZ)^{3}\times (\ZZ/8\ZZ)\times (\ZZ/16\ZZ).
\]

However, note that $\F2[\ZZ/2^{k}\ZZ]\simeq \F2[X]/(X^{2^{k}})$. In
other words, multiplication in this group algebra can be thought as a
truncated multiplication in $\F2[X]$. Now, many algorithms have been
developped as analogues of FFT in characteristics $2$. They are known
as additive fast Fourier transform, and even benefit from very
efficient implementation ~\cite{Coxon21,LCH18}. If they actually work
over extensions of $\F2$, recent works such as
\cite{DLSCCKMCYY18} suggest that multiplying polynomials
over $\F2$ could be made very efficient.

Finally, as the group algebra of a direct product, $\F2[\Gamma]$ is a
tensor product, \ie isomorphic to a multivariate ring, where the
degrees of the variables are bounded by the corresponding power of
$2$. The existence of multivariate FFT also suggests the existence of
efficient multivariate {\em additive} FFT in characteristics $2$.

Moreover, this description is very naive, and further work may
actually directly design efficient algorithms for multiplication in
modular group algebras over $\F2$ in the spirit of what has been done
for $(\ZZ/2\ZZ)^{s}$.

\end{document}